\documentclass[11pt,a4paper]{article}

\usepackage[utf8]{inputenc}
\usepackage[english]{babel}
\usepackage[T1]{fontenc}
\usepackage{amsmath}
\usepackage{tikz}
\usetikzlibrary{calc, intersections, fpu}
\usetikzlibrary{arrows.meta} %cool arrows
\usetikzlibrary{decorations.text}
\usepackage{tikz-cd}
\usepackage{amsfonts}
\usepackage{amssymb}
\usepackage{amsthm}
\usepackage{graphicx}
\usepackage{stmaryrd}
\usepackage{comment}
\usepackage{algorithm}
\usepackage{imakeidx}
\usepackage[noend]{algpseudocode}
\usepackage{listings}
\usepackage{enumitem}
\usepackage{mathtools}
\usepackage{dsfont}
\usepackage{xcolor}
\usepackage{geometry}
\usepackage[colorlinks=true, linkcolor=blue, hyperindex=true, citecolor = green!60!purple]{hyperref}
\usepackage{pgfmath}
\usepackage{standalone}

\newcommand\ens[1]{\left\{ #1 \right\}}
\newcommand\enstq[2]{\left\{ #1~|~#2 \right\}}
\newcommand\inter[2]{\llbracket #1,#2 \rrbracket}
\newcommand\oper[4]{#1 \limits_{\substack{#2}}^{#3}{#4}}
\newcommand{\N}{\mathds N}
\newcommand{\To}{\mathds T}
\newcommand{\Sp}{\mathbb S}

\newcommand{\bfdex}[1]{\textbf{#1}\index{#1}}
\newcommand{\tand}{\text{ and }}

\newcommand{\T}{\mathcal{T}}
\renewcommand{\P}{\mathcal{P}}
\newcommand{\control}[2]{ ($#2+#1$) .. #1}
\newcommand{\comprep}{\text{c-rep}}
\newcommand{\cT}{\text{c-}\mathcal{T}}
\newcommand{\B}{\mathcal{B}}
\newcommand{\obullet}[1]{{\bullet}({#1})}%\overset{\hbox{\tiny $\bullet$}}{#1}

\newtheorem{theorem}{Theorem}[section]

\newtheorem*{thm*}{Theorem}

\newtheorem{defn}[theorem]{Definition}
\newtheorem{proposition}[theorem]{Proposition}
\newtheorem{corollary}[theorem]{Corollary}
\newtheorem{remark}[theorem]{Remark}

\newtheorem{lem}[theorem]{Lemma}

\title{A Structural Approach to Tree Decompositions of Knots and Spatial Graphs\thanks{This research was partially supported by the ANR project SoS (ANR-17-CE40-0033).}} 

\author{Corentin Lunel\thanks{LIGM, CNRS, Univ. Gustave Eiffel, ESIEE Paris, F-77454 Marne-la-Vall\'ee, France, corentin.lunel2@univ-eiffel.fr} \and Arnaud de Mesmay\thanks{LIGM, CNRS, Univ. Gustave Eiffel, ESIEE Paris, F-77454 Marne-la-Vall\'ee, France, arnaud.de-mesmay@univ-eiffel.fr}}
\date{}

\makeindex

\definecolor{dred}{rgb}{0.6,0,0}

\definecolor{dblue}{rgb}{0.3,0.3,0.8}

\makeatletter
\def\BState{\State\hskip-\ALG@thistlm}
\makeatother

\begin{document}

\maketitle

\begin{abstract}
  Knots are commonly represented and manipulated via diagrams, which are decorated planar graphs. When such a knot diagram has low treewidth, parameterized graph algorithms can be leveraged to ensure the fast computation of many invariants and properties of the knot. It was recently proved that there exist knots which do not admit any diagram of low treewidth, and the proof relied on intricate low-dimensional topology techniques. In this work, we initiate a thorough investigation of tree decompositions of knot diagrams (or more generally, diagrams of spatial graphs) using ideas from structural graph theory. We define an obstruction on spatial embeddings that forbids low tree width diagrams, and we prove that it is optimal with respect to a related width invariant. We then show the existence of this obstruction for knots of high representativity, which include for example torus knots, providing a new and self-contained proof that those do not admit diagrams of low treewidth. This last step is inspired by a result of Pardon on knot distortion.
\end{abstract}

\section{Introduction}

A (tame) \textbf{knot} is a polygonal embedding of the circle $S^1$ into $\mathbb{R}^3$, and two knots are considered equivalent if they are \textbf{isotopic}, i.e., if they can be continuously deformed one into the other without introducing self-intersections. The \textbf{trivial knot}, or \textbf{unknot}, is, up to equivalence, the embedding of $S^1$ as a triangle. The investigation of knots and their mathematical properties dates back to at least the nineteenth century~\cite{knotbook} and has developed over the years into a very rich and mature mathematical theory. From a computational perspective, a fundamental question is to figure out the best algorithm testing whether a given knot is the unknot. Note that it is neither obvious from the definitions that a non-trivial knot exists, nor that the problem is decidable. This was famously posed as an open problem by Turing~\cite{turing}. The current state of the art on this problem is that it lies in \textbf{NP}~\cite{Hass_trivial_NP} and co-\textbf{NP}~\cite{Lackenby_co-NP}, a quasipolynomial time algorithm has been announced~\cite{Lackenby_talk} but no polynomial-time algorithm is known. More generally, algorithmic questions surrounding knots typically display a wide gap between the best known algorithms (which are almost never polynomial-time, and sometimes the complexity is a tower of exponentials) and the best known complexity lower bounds. We refer to the survey of Lackenby for a panorama of algorithms in knot theory~\cite{Lackenby_Algorithms}.

In recent years, many attempts have been made to attack such seemingly hard problems via the route of parameterized algorithms. In particular, the \textbf{treewidth} of a graph is a parameter quantifying how close a graph is to a tree, and thus algorithmic problems on graphs of low treewidth can often be solved very efficiently using dynamic programming techniques on the underlying tree structure of instance. The concept of \textbf{branchwidth}, which we also use below, is somewhat equivalent and always within a constant factor of treewidth~\cite{Graph_Minors_X}. One approach is to study a knot via one of its \textbf{diagrams} (see Figure~\ref{pic_knots_spatial}), that is, a decorated graph obtained by a planar projection where it is indicated on each vertex which strand goes over or under. Then, if such a diagram has low treewidth, one can apply these standard dynamic programming techniques to solve seemingly hard problems very efficiently. While this approach has not yet been successful for unknot recognition beyond treewidth $2$~\cite{Bodlaender_Treewidth_2}, it has proved effective for the computation of many knot invariants, including: Jones and Kauffman polynomials~\cite{Malowsky_Jones} (which are known to be $\#P$-hard to compute in general~\cite{Jaeger}), HOMFLY-PT polynomials~\cite{Burton_Homfly}, and quantum invariants~\cite{maria,Burton_2018}. Since any knot admits infinitely many diagrams, these algorithms naturally lead to the following question raised by Burton~\cite[p.2694]{burton}, and Makowsky and Maris\~{n}o~\cite[p.755]{Malowsky_Jones}: do all knots admit diagrams of constant treewidth, or conversely does there exist a family of knots for which \textit{all} the diagrams have treewidth going to infinity. This question was answered recently by de Mesmay, Purcell, Schleimer and Sedgwick~\cite{Mesmay_treewidth} who proved that, among other examples, torus knots $T_{p,q}$ are such a family. The proof relies at its core on an intricate result of Hayashi and Shimokawa~\cite{Hayashi_spliting} on \emph{thin position} of \emph{multiple Heegaard splittings}.

\begin{figure}[h]
\begin{center}
%\tikzsetnextfilename{example_knots}
\begin{tikzpicture}[scale = 0.75]
\def\ec{1.5}
\clip (-2.2,-2) rectangle (16,3);
\draw [thick] (0,0.5) circle (2cm);

%Trivial knot
\begin{scope}[xshift = 3cm]
\coordinate (a) at (1,2);
\coordinate	(b) at (0,1);
\coordinate	(c) at (1,0);
\coordinate	(d) at (2,1);
\coordinate	(e) at (0.5,-1);
\coordinate	(f) at (1.5,-1);
\node [fill,white,circle, inner sep = \ec pt] (na) at (a) {};
\node [fill,white,circle, inner sep = \ec pt] (nb) at (b) {};
\node [fill,white,circle, inner sep = \ec pt] (nc) at (c) {};
\node [fill,white,circle, inner sep = \ec pt] (nd) at (d) {};
\node [fill,white,circle, inner sep = \ec pt] (ne) at (e) {};
\node [fill,white,circle, inner sep = \ec pt] (nf) at (f) {};

\begin{scope}[thick]
%Crossings
\draw (na.south east) -- (na.north west);
\draw (nb.south east) -- (nb.north west);
\draw (nc.south west) -- (nc.north east);
\draw (nd.south west) -- (nd.north east);
\draw (ne.south) -- (ne.north);
\draw (nf.east) -- (nf.west);

%strands
\draw (na.south west) -- (nb.north east);
\draw (nb.south east) -- (nc.north west);
\draw (nc.north east) -- (nd.south west);
\draw (nd.north west) -- (na.south east);
\draw (ne.east) -- (nf.west);
\draw (nc.south west) .. controls +(-0.4,-0.4) and \control{(ne.north)}{(0,0.3)};
\draw (nc.south east) .. controls +(0.4,-0.4) and \control{(nf.north)}{(0,0.3)};
\draw (na.north east) .. controls +(1.4,1.4) and \control{(nd.north east)}{(1.4,1.4)};
\draw (na.north west) .. controls +(-1.4,1.4) and \control{(nb.north west)}{(-1.4,1.4)};
\draw (ne.south) .. controls +(0,-1) and \control{(nf.south)}{(0,-1)};
\draw (nb.south west) .. controls +(-1,-1) and \control{(ne.west)}{(-1,0)};
\draw (nd.south east) .. controls +(1,-1) and \control{(nf.east)}{(1,0)};
\end{scope}
\end{scope}

%Figure eight knot
\begin{scope}[xshift = 7cm]
\coordinate (a) at (1,2);
\coordinate	(b) at (0,1);
\coordinate	(c) at (1,0);
\coordinate	(d) at (2,1);
\coordinate	(e) at (0.5,-1);
\coordinate	(f) at (1.5,-1);
\node [fill,white,circle, inner sep = \ec pt] (na) at (a) {};
\node [fill,white,circle, inner sep = \ec pt] (nb) at (b) {};
\node [fill,white,circle, inner sep = \ec pt] (nc) at (c) {};
\node [fill,white,circle, inner sep = \ec pt] (nd) at (d) {};
\node [fill,white,circle, inner sep = \ec pt] (ne) at (e) {};
\node [fill,white,circle, inner sep = \ec pt] (nf) at (f) {};

\begin{scope}[thick]
%Crossings
\draw (na.south east) -- (na.north west);
\draw (nb.south west) -- (nb.north east);
\draw (nc.south east) -- (nc.north west);
\draw (nd.south west) -- (nd.north east);
\draw (ne.south) -- (ne.north);
\draw (nf.east) -- (nf.west);

%Strands
\draw (na.south west) -- (nb.north east);
\draw (nb.south east) -- (nc.north west);
\draw (nc.north east) -- (nd.south west);
\draw (nd.north west) -- (na.south east);
\draw (ne.east) -- (nf.west);
\draw (nc.south west) .. controls +(-0.4,-0.4) and \control{(ne.north)}{(0,0.3)};
\draw (nc.south east) .. controls +(0.4,-0.4) and \control{(nf.north)}{(0,0.3)};
\draw (na.north east) .. controls +(1.4,1.4) and \control{(nd.north east)}{(1.4,1.4)};
\draw (na.north west) .. controls +(-1.4,1.4) and \control{(nb.north west)}{(-1.4,1.4)};
\draw (ne.south) .. controls +(0,-1) and \control{(nf.south)}{(0,-1)};
\draw (nb.south west) .. controls +(-1,-1) and \control{(ne.west)}{(-1,0)};
\draw (nd.south east) .. controls +(1,-1) and \control{(nf.east)}{(1,0)};
\end{scope}
\end{scope}

%Spatial qraph
\begin{scope}[xshift = 13cm, yshift = 0.5cm]
\clip (-3,-2.5) rectangle (3,2.5);
%\draw[help lines] (-2.5,-2.5) grid (2.5,2.5);
\coordinate (1) at (-1,2);
\coordinate	(2) at (1,2);
\coordinate	(3) at (-2,1);
\coordinate	(4) at (-1,1);
\coordinate	(5) at (1,1);
\coordinate	(6) at (2,1);
\coordinate	(7) at (0,0.5);
\coordinate	(8) at (0,-1);
\coordinate	(9) at (0,-2);
\foreach \i in {1,...,9}
{
	\node (n\i) [fill, circle, white, inner sep =\ec pt] at (\i) {};
}

\begin{scope}[thick]
%Crossings
\node [fill, circle, inner sep =\ec pt] at (1) {};
\node [fill, circle, inner sep =\ec pt] at (2) {};
\draw (n3.south) -- (n3.north);
\draw (n4.south) -- (n4.north);
\draw (n5.west) -- (n5.east);
\draw (n6.west) -- (n6.east);
\draw (n7.north west) -- (n7.south east);
\draw (n8.north east) -- (n8.south west);
\draw (n9.north east) -- (n9.south west);

%Strands
\draw (n1.east) -- (n2.west);
\draw (n1.west) .. controls +(-0.5,0) and \control{(n3.north)}{(0,0.5)};
\draw (n1.south) -- (n4.north);
\draw (n2.east) .. controls +(0.5,0) and \control{(n6.north)}{(0,0.5)};
\draw (n2.south) -- (n5.north);
\draw (n3.east) -- (n4.west);
\draw (n3.south) .. controls +(0,-1.5) and \control{(n8.north west)}{(-0.5,0.5)};
\draw (n3.west) .. controls +(-2,0) and \control{(n9.south west)}{(-1,-1)};
\draw (n4.south) .. controls +(0,-0.5) and \control{(n7.south west)}{(-0.3,-0.3)};
\draw (n4.east) .. controls +(0.5,0) and \control{(n7.north west)}{(-0.3,0.3)};
\draw (n5.west) .. controls +(-0.5,0) and \control{(n7.north east)}{(0.3,0.3)};
\draw (n5.south) .. controls +(0,-0.5) and \control{(n7.south east)}{(0.3,-0.3)};
\draw (n5.east) -- (n6.west);
\draw (n6.south) .. controls +(0,-1.5) and \control{(n8.north east)}{(0.5,0.5)};
\draw (n6.east) .. controls +(2,0) and \control{(n9.south east)}{(1,-1)};
\draw (n8.south east) .. controls +(0.3,-0.3) and \control{(n9.north east)}{(0.3,0.3)};
\draw (n8.south west) .. controls +(-0.3,-0.3) and \control{(n9.north west)}{(-0.3,0.3)};
\end{scope}
\end{scope}
\end{tikzpicture}
\caption{Diagrams of two trivial knots on the left, a bowline knot and a knotted spatial graph.}
\label{pic_knots_spatial}
\end{center}
\end{figure}
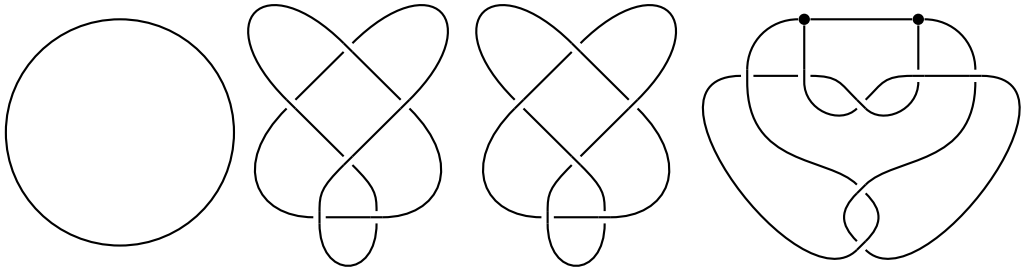

\subparagraph*{Our results.} The main purpose of this work is to provide new techniques to characterize which knots, or more generally which \textbf{spatial graphs} (polygonal embeddings of graphs into $\mathbb{R}^3$, considered up to isotopy, see for example Figure~\ref{pic_knots_spatial}), do not admit diagrams of low treewidth. Our starting point is similar to the one in~\cite{Mesmay_treewidth}: we first observe that if a knot or a spatial graph admits a diagram of low treewidth, then there is a way to sweep $\mathbb{R}^3$ using spheres arranged in a tree-like fashion which intersect the knot a small number of times (Proposition~\ref{P:swtw}). This corresponds roughly to a map $f: \mathbb{R}^3 \rightarrow T$ where $T$ is a trivalent tree, where the preimage of each point interior to an edge is a sphere with a small number of intersections with the knot (we refer to Section~\ref{S:prelim} for the precise technical definitions of all the concepts discussed in this introduction). We call this a \textbf{sphere decomposition}\footnote{Our sphere decompositions are different from the ones in~\cite{Mesmay_treewidth} but functionally equivalent for knots.}, and the resulting measure (maximal number of intersections) the \textbf{spherewidth} of the knot.

Thus, in order to lower bound the treewidth of all the diagrams of a knot, it suffices to lower bound its spherewidth. We provide a systematic technique to do so using a perspective taken from structural graph theory. In the proof of the celebrated Graph Minor Theorem of Robertson and Seymour~\cite{Graph_Minors_XX}, handling families of graphs with bounded treewidth turns out not to be too hard~\cite{Graph_Minors_V}, and in contrast, a large part is devoted to analyzing the structure shared by graphs of large treewidth. There, a fundamental contribution is the concept of \textbf{tangle}\footnote{It turns out that the word \emph{tangle} holds a completely different meaning in knot theory, and, to avoid confusion, in this article we will always use it with the graph theory meaning.}. We refer to Diestel~\cite{monalisa} or Grohe~\cite{grohetangles} for nice introductions to tangles and their applications. Informally, a tangle of order $k$ in a graph $G$ is a choice, for each separation of size at most $k$, of a ``big side'' of that separation, where the highly-connected part of the graph lies. In addition, there are some compatibility properties, in particular no three ``small sides'' should cover the whole graph. Such a tangle turns out to be exactly the structure dual to branchwidth, in the sense that, as is proved in~\cite{Graph_Minors_X}, for any graph $G$, the maximal possible order of a tangle is exactly equal to its branchwidth. We develop a similar concept dual to sphere decompositions which we call a \textbf{bubble tangle}. Informally, a bubble tangle of order $k$ for a knot or spatial graph $K$ is a map that, for each sphere intersecting $K$ at most $k$ times, chooses a ``big side'' indicating where the complicated part of $K$ lies. There are again some compatibility conditions which add topological information to the collection of ``small sides''. Then our first result is the following. 

\begin{theorem}\label{T:duality}
For any knot or spatial graph $K$, the maximum order of a bubble tangle for $K$ is equal to the spherewidth of $K$.
\end{theorem}

This provides a convenient and systematic pathway to prove lower bounds on the spherewidth, and thus on the treewidth of all possible diagrams: it suffices to prove the existence of a bubble tangle of high order. However, making choices for the uncountable family of spheres with a small number of intersections with $K$, and then verifying the needed compatibility conditions is very unwieldy. Our second contribution is to provide a way to define such a bubble tangle in the case of knots (or spatial graphs) which are embedded on some surface $\Sigma$ in $\mathbb{R}^3$. Given a surface $\Sigma$ in $\mathbb{R}^3$, a \textbf{compression disk} is a disk properly embedded in $\mathbb{R}^3 \smallsetminus \Sigma$  whose boundary is a non-contractible curve on $\Sigma$. The \textbf{compression-representativity of an embedding} of a knot or spatial graph $K$ on a surface $\Sigma$ in $\mathbb{R}^3$ is the smallest number of intersections between $K$ and a cycle on $\Sigma$ that bounds a compression disk. The \textbf{compression-representativity} of a knot or spatial graph is the supremum of that quantity over all embeddings on surfaces (this was originally defined by Ozawa \cite{ozawa}). Our second theorem is the following.

\begin{theorem}\label{T:representativity}
For any knot or spatial graph $K$ embedded on a surface $\Sigma$ in $\mathbb{R}^3$, there exists a bubble tangle of order $2/3$ of the compression-representativity of the embedding. Therefore, for any knot or spatial graph $K$, there exists a bubble tangle of order $2/3$ of the compression-representativity of $K$.
\end{theorem}

Combining together Theorems~\ref{T:duality} and~\ref{T:representativity} with Proposition~\ref{P:swtw} provides a large class of knots of high spherewidth, and our tools are versatile enough to apply to spatial graphs, while previous ones did not. In particular, observing that torus knots $T_{p,q}$ have high compression-representativity, we obtain the following corollary, which improves the lower bound obtained by~\cite{Mesmay_treewidth}, without relying on deep knot-theoretical tools.

\begin{corollary}\label{C:torusknots}
A torus knot $T(p,q)$ has spherewidth at least $2/3\min(p,q)$, and thus any diagram of $T(p,q)$ has treewidth at least $1/3 \min(p,q)$.
\end{corollary}

\subparagraph*{Related work and proof techniques.} The results in this article and many of their proof techniques stem from two very distinct lineages in quite distant communities, the first one being knot theory or more generally low-dimensional topology, and the second one being structural graph theory. While there have been some recent works aiming at building bridges between combinatorial width parameters and topological quantities (for example the aforementioned~\cite{Mesmay_treewidth}, but also~\cite{huszar2,huszar1,mariapurcell} for related problems in $3$-manifold theory), the main contribution in this article is that we dive deeper in the structural graph theory perspective via the concept of a tangle. The latter has now proved to be a fundamental tool in graph theory and beyond (see for example Diestel~\cite[Preface to the 5th edition]{Diestel_Graph_Theory}).

The duality theorem of Robertson and Seymour between branchwidth and tangles in~\cite{Graph_Minors_X} has been generalized many times since its inception, for example in order to encompass other notions of decompositions and their obstructions~\cite{amini,lyaudet}, to apply more generally to matroids~\cite{Geelen_Matroid} and to the wide-ranging setting of abstract separations systems~\cite{Diestel_Tangle,diestelACS}. The key difference in our work, and why it does not fit into these generalizations, is that our notions of sphere decomposition and bubble tangles inherently feature the \emph{topological} constraint of working with $2$-spheres. This is a crucial constraint, as it would be easy to sweep any knot with width at most $2$ if one were allowed to use arbitrary surfaces during the sweeping process. It should also be noted that in planar graphs, it was shown by Seymour and Thomas~\cite{Seymour_Ratcatcher} that the separations involved in an optimal branch decomposition can always be assumed to take the form of $1$-spheres, i.e., Jordan curves. This property led to the celebrated ratcatcher algorithm to compute the branchwidth of planar graphs in polynomial time~\cite{Seymour_Ratcatcher} as well as to sphere-cut decompositions and their algorithmic applications (see for example~\cite{dorn2005efficient}). Our sphere decompositions are the generalization one dimension higher of these sphere-cut decompositions, and Theorem~\ref{T:duality} identifies bubble tangles as a correct notion of dual obstruction for those. We believe that these notions could be of further interest beyond knots, in the study of graphs embedded in $\mathbb{R}^3$ with some topological constraints, e.g., linkless graphs~\cite{sachs}.

The \textbf{representativity} (also called \textbf{facewidth}) of a graph embedded on a surface $S$ is the smallest number of intersections of a non-contractible curve with that graph. Theorem~\ref{T:representativity} will not come as a surprise for readers accustomed to graph minor theory, as Robertson and Seymour proved a very similar-looking theorem in Graph Minors XI~\cite[Theorem~4.1]{Graph_minors_XI}, showing that that the branchwidth of a graph embedded on a surface is lower bounded by its representativity, which they prove by exhibiting a tangle. The key difference is that our notion of compression-representativity only takes into account the length of cycles bounding compression disks, instead of all the non-contractible cycles. Here again, this topological distinction is crucial to give a meaningful concept for knots, as for example the graph-theoretical representativity of a torus knot is zero. Due to this difference, the proof technique of Robertson and Seymour does not readily apply to prove Theorem~\ref{T:representativity}; instead we have to rely on more topological arguments.

From the knot theory side, there is a long history in the study of the ``best'' way to sweep a knot while trying to minimize the number of intersections in this sweepout. One of the oldest knot invariants, the bridge number, can be seen through this lens (see for example~\cite{schultens}). A key concept in modern knot theory, introduced by Gabai in his proof of the Property R conjecture~\cite{gabai}, is the notion of \textbf{thin position} which more precisely quantifies the best way to place a knot to  minimize its width. It is at the core of many advances in modern knot theory (see for example Scharlemann~\cite{scharlemann}). Recent developments in thin position have highlighted that in order to obtain the best topological properties, it can be helpful to sweep the knot in a tree-like fashion compared to the classical linear one. This approach leads to definitions bearing close similarities to our sphere decompositions (this is one of the ideas behind generalized Heegaard splittings~\cite{scharlemannthompson,genHeegaard}, see also~\cite{Hayashi_spliting,widthtrees,tayloradditive}). The concept of compression-representativity of a knot or a spatial graph finds its roots in the works of Ozawa~\cite{ozawa}, and Blair and Ozawa~\cite{blairozawa} who defined it under the simple name of representativity, taking inspiration from graph theory. They proved that it provides a lower bound on the bridge number and on more general linear width quantities. Our Theorem~\ref{T:representativity} strengthens their results by showing that it also lower bounds the width of tree-like decompositions. Furthermore, while specific tools have been developed to lower bound various notions of width of knots or $3$-manifolds, we are not aware of duality theorems like our Theorem~\ref{T:duality}. It shows that our bubble tangles constitute an obstruction that is, in a precise sense, the optimal tool for the purpose of lower bounding spherewidth.

Finally, an important inspiration for our proof of Theorem~\ref{T:representativity} comes from a seemingly unrelated breakthrough of Pardon~\cite{Pardon_distortion}, who solved a famous open problem of Gromov~\cite{gromov} by proving the existence of knots with arbitrarily high \textbf{distortion}. The distortion for two points on an embedded curve in $\mathbb{R}^3$ is the ratio between the intrinsic and Euclidean distance between the points. The distortion of the entire curve is the supremum over all pairs of points. The distortion of a knot is the minimal distortion over all embeddings of the knot. While this metric quantity seems to have nothing to do with tree decompositions, it turns out that the technique developed by Pardon can be reinterpreted in our framework. With our terminology, his proofs amounts to first lower bounding the distortion by a constant factor of the spherewidth, and then defining a bubble tangle for knots of high representativity. The lower bound is nicely explained by Gromov and Guth~\cite[Lemma~4.2]{gromovguth}, where the simplicial map is similar to our sphere decompositions, up to a constant factor. Then our proof of Theorem~\ref{T:representativity} is inspired by the second part of Pardon's argument, with a quantitative strengthening to obtain the $2/3$ factor, whereas his argument would only yield $1/2$.

\subparagraph*{Organization of this paper.} After providing background and defining our key concepts in Section~\ref{S:prelim}, we prove Theorem~\ref{T:duality} in Section~\ref{S:duality}, and Theorem~\ref{T:representativity} in Section~\ref{S:representativity}. We provide examples in Section~\ref{S:examples}. %Due to the line limitations, all the proofs not included in the main document are deferred to appendices.

\section{Preliminaries}\label{S:prelim}

We include the most relevant definitions, but some familiarity with low-dimensional topology will help, see for example in the textbook of Schultens~\cite{schultensbook}. We refer to Diestel~\cite{Diestel_Graph_Theory} for a nice introduction to graph theory and in particular its structural aspects. We denote by $V(G)$, $E(G)$, and $L(G)$ the vertices, edges and leaves (degree one vertices) of a graph $G$.

\subparagraph*{Low-dimensional topology.} Following standard practice, instead of considering knots and spatial graphs within $\mathbb{R}^3$, we compactify it and work within $\mathbb{S}^3$.  We denote by $C(A)$ the connected components of a subset $A$ of $\Sp^3$, and thus by $|C(A)|$ its number of connected components. As is standard in low-dimensional topology, we work in the Piecewise-Linear (PL) category, which means that \textit{all the objects that we use in this article are assumed to be piecewise-linear}, i.e., made of a finite number of linear pieces with respect to a fixed triangulation of $\mathbb{S}^3$. This allows us to avoid pathologies such as wild knots or the Alexander horned sphere. An \bfdex{embedding} of a compact topological space $X$ into another one $Y$ is a continuous injective map, and it is \bfdex{proper} if it maps the boundary $\partial X$ within the boundary $\partial Y$. A $3$-dimensional version of the Schoenflies theorem guarantees that for any PL $2$-sphere $S$ embedded in $\Sp^3$, both components of $\Sp^3 \smallsetminus S$ are balls (see for example~\cite[Theorem~XIV.1]{bing}). A \bfdex{knot} is a PL embedding of $S^1$ into $\Sp^3$, a \bfdex{link} is a disjoint union of knots, and a \bfdex{spatial graph} is a PL embedding of a graph $G$ into $\Sp^3$. All these objects are considered equivalent when they are ambient isotopic, i.e., when there exists a continuous deformation preserving the embeddedness. Knots and links are a special instance of spatial graphs, and henceforth we will mostly focus on spatial graphs, generally denoted by the letter $G$. For technical reasons, it is convenient to thicken our embedded graphs as follows. A \bfdex{thickened embedding} $\varphi$ of a graph $G$ is an embedding of $G$ in $\mathbb{S}^3$ where each vertex is thickened to a small ball, two balls are connected by a polygonal edge if and only if they are adjacent in the graph $G$, and pairs of edges are disjoint. We emphasize that we do not thicken edges, which might be considered nonstandard.  We will also work with graphs embedded on surfaces which are themselves embedded in $\mathbb{S}^3$: such embeddings will also always be thickened, that is, vertices on the surface are thickened into small disks. 
\textit{From now on, all the graph embeddings will be thickened, and thus to ease notations we will omit the word thickened}.

As mentioned in the introduction, for $\Sigma$ a surface embedded in $\Sp^3$, a \bfdex{compression disk} is a properly embedded disk $D$ in $\mathbb{S}^3 \smallsetminus \Sigma$ such that the boundary $\partial D$ is a non-contractible curve on $\Sigma$. A \bfdex{compressible curve} $\gamma$ of $\Sigma$ is the boundary of a compressing disk of $\Sigma$. 
 For a spatial graph $G$ embedded on a surface $\Sigma$ in $\mathbb{S}^3$, the \bfdex{compression representativity}  of $G$ on $\Sigma$, written $\comprep(G,\Sigma)$ is $\min \enstq{|C(\alpha \cap G|)}{\alpha \text{ compressible curve of } \Sigma}$ (we count connected components to correctly handle thickened vertices). The compression representativity $\comprep(G)$ of $G$ is the supremum of $\comprep(G,\Sigma)$ over all nested embeddings $G \hookrightarrow \Sigma \hookrightarrow \mathbb{S}^3$.

In order to define spherewidth and bubble tangles, we require a precise control of the event when two spheres merge together to yield a third one, which is mainly encapsulated in the concept of double bubble. A \bfdex{double bubble} is a triplet of closed disks $(D_1,D_2,D_3)$ in $\mathbb{S}^3$, disjoint except on their boundaries, that they share: $D_1 \cap D_2 = D_1 \cap D_3 = D_2 \cap D_3 = D_1 \cap D_2 \cap D_3 = \partial D_1 = \partial D_2 = \partial D_3$, see Figure~\ref{pic_double_bubble}. Such a double bubble defines three spheres, which, by the PL Schoenflies theorem, bound three balls.

Two surfaces (resp. a knot and a surface) embedded in $\mathbb{S}^3$ are \bfdex{transverse} if they intersect in a finite number of connected components, where the intersection is locally homeomorphic to the intersection of two orthogonal planes (resp. to the intersection of a plane and an orthogonal line). Likewise, we say that a surface is transverse to a ball if it is transverse to its boundary. A surface is transverse to a graph if it is transverse to all the thickened vertices and edges it intersects. A double bubble is transverse to a graph or a surface if each of its three spheres is and if the vertices of the graph do not intersect the spheres on their shared circle $\partial D_i$. Intersections are \bfdex{tangent} when they are not transverse, and a sphere $S$ is said \bfdex{finitely tangent} to a graph $G$ embedded in $\Sp^3$ if they do not intersect transversely but the number of intersections $|E(G) \cap S|$ is still finite.

\begin{figure}[h]
\begin{center}
%\tikzsetnextfilename{double_bubble}
\begin{tikzpicture}
\coordinate (clsp) at (2,1.4);
\coordinate (ulsp) at ($(clsp)+(60:1)$);
\coordinate (dlsp) at ($(clsp)+(300:1)$);

\filldraw [purple, fill opacity = 0.1] (ulsp) arc (60:300:1);
\filldraw [purple, fill opacity = 0.1] (ulsp) arc (120:-120:1);

\draw [purple] (ulsp) .. controls +(-40:0.3) and \control{(dlsp)}{(40:0.3)} node (mr) [pos = 0.8, inner sep = 0] {};
\draw [purple, dotted] (ulsp) .. controls +(-140:0.3) and \control{(dlsp)}{(140:0.3)}node (ml) [pos = 0.55, inner sep = 0] {};
\fill [purple, opacity = 0.1] (ulsp) .. controls +(-40:0.3) and \control{(dlsp)}{(40:0.3)} -- (dlsp) .. controls +(140:0.3) and \control{(ulsp)}{(-140:0.3)};
\end{tikzpicture}
\caption{A double bubble: two spheres that intersect in a single disk.}
\label{pic_double_bubble}
\end{center}
\end{figure}
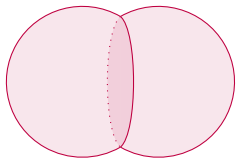

\subparagraph*{Spherewidth.} In this paragraph, we introduce sphere decompositions, which are the main way that we use to sweep knots and spatial graphs using spheres.

\begin{defn}[Sphere decomposition]
  Let $G$ be a graph embedded in $\Sp^3$. A \bfdex{sphere decomposition} of $G$ is a continuous map $f: \Sp^3 \rightarrow T$ where $T$ is a trivalent tree with at least one edge: 
	\begin{itemize}
		\item For all $x \in L(T), f^{-1}(x)$ is a point disjoint from $G$.
		\item For all $x \in V(T) \smallsetminus L(T), f^{-1} (x)$ is a PL double bubble transverse to $G$.
		\item For all $x$ interior to an edge, $f^{-1} (x)$ is a sphere transverse or finitely tangent to $G$.
	\end{itemize}
\end{defn}

The \bfdex{weight} of a sphere $S$ (with respect to $G$) is the number of  connected components in its intersection with $G$. The \bfdex{width} of a sphere decomposition $f$ is the supremum of the weight of $f^{-1}(x)$ over all points $x$ interior to edges of the tree $T$.  The \bfdex{spherewidth} of the graph $G$, denoted by $sw(G)$, is the infimum, over all sphere decompositions $f$, of the width of $f$: $sw(G) = \inf_{f: \Sp^3 \mapsto T} \sup_{x \in \mathring e \in E(T)} |C(f^{-1}(x) \cap G)|$. Therefore, a sphere decomposition is a way to continuously sweep $\Sp^3$ using spheres, which will occasionally merge or split in the form of double bubbles, and the spherewidth is a measure of how well we can sweep a graph $G$ using sphere decompositions. This is similar to the level sets given by a Morse function, but note that our double-bubble singularities are not of Morse type, and those are key to the proof of Theorem~\ref{T:duality}.

\begin{figure}[ht]
\begin{center}
%%\tikzsetnextfilename{Pretzel_decomp}
\begin{tikzpicture}[scale = 1]
%define intersection
\coordinate (b11) at (0,1);
\coordinate (b12) at (0,2);
\coordinate (b13) at (0,3);
\coordinate (b21) at (1,0.5);
\coordinate (b22) at (1,1.5);
\coordinate (b23) at (1,2.5);
\coordinate (b24) at (1,3.5);
\coordinate (b31) at (2,0);
\coordinate (b32) at (2,1);
\coordinate (b33) at (2,2);
\coordinate (b34) at (2,3);
\coordinate (b35) at (2,4);
\coordinate (bu1) at (0,5);
\coordinate (bu2) at (1,5);
\coordinate (bu3) at (2,5);
\def\p{0.4}
\def\sr{45};%north east
\pgfmathparse{180-\sr}\edef\sl{\pgfmathresult};%northwest
\pgfmathparse{-\sr}\edef\er{\pgfmathresult};%south east
\pgfmathparse{-\sl}\edef\el{\pgfmathresult};%south west
\def\r{2pt};
\begin{scope}[white]
\node (n11) at (b11) [fill, circle, fill, inner sep = \r] {};
\node (n12) at (b12) [fill, circle, fill, inner sep = \r] {};
\node (n13) at (b13) [fill, circle, fill, inner sep = \r] {};
\node (n21) at (b21) [fill, circle, fill, inner sep = \r] {};
\node (n22) at (b22) [fill, circle, fill, inner sep = \r] {};
\node (n23) at (b23) [fill, circle, fill, inner sep = \r] {};
\node (n24) at (b24) [fill, circle, fill, inner sep = \r] {};
\node (n31) at (b31) [fill, circle, fill, inner sep = \r] {};
\node (n32) at (b32) [fill, circle, fill, inner sep = \r] {};
\node (n33) at (b33) [fill, circle, fill, inner sep = \r] {};
\node (n34) at (b34) [fill, circle, fill, inner sep = \r] {};
\node (n35) at (b35) [fill, circle, fill, inner sep = \r] {};
\node (nu1) at (bu1) [fill, circle, fill, inner sep = \r] {};
\node (nu2) at (bu2) [fill, circle, fill, inner sep = \r] {};
\node (nu3) at (bu3) [fill, circle, fill, inner sep = \r] {};
\end{scope}

%draw strands
\begin{scope}[thick]
\draw (n11.north east) .. controls +(\sr:\p) and \control{(n12.south east)}{(\er:\p)};
\draw (n11.north west) .. controls +(\sl:\p) and \control{(n12.south west)}{(\el:\p)};
\draw (n12.north east) .. controls +(\sr:\p) and \control{(n13.south east)}{(\er:\p)};
\draw (n12.north west) .. controls +(\sl:\p) and \control{(n13.south west)}{(\el:\p)};
\draw (n21.north east) .. controls +(\sr:\p) and \control{(n22.south east)}{(\er:\p)};
\draw (n21.north west) .. controls +(\sl:\p) and \control{(n22.south west)}{(\el:\p)};
\draw (n22.north east) .. controls +(\sr:\p) and \control{(n23.south east)}{(\er:\p)};
\draw (n22.north west) .. controls +(\sl:\p) and \control{(n23.south west)}{(\el:\p)};
\draw (n23.north east) .. controls +(\sr:\p) and \control{(n24.south east)}{(\er:\p)};
\draw (n23.north west) .. controls +(\sl:\p) and \control{(n24.south west)}{(\el:\p)};
\draw (n31.north east) .. controls +(\sr:\p) and \control{(n32.south east)}{(\er:\p)};
\draw (n31.north west) .. controls +(\sl:\p) and \control{(n32.south west)}{(\el:\p)};
\draw (n32.north east) .. controls +(\sr:\p) and \control{(n33.south east)}{(\er:\p)};
\draw (n32.north west) .. controls +(\sl:\p) and \control{(n33.south west)}{(\el:\p)};
\draw (n33.north east) .. controls +(\sr:\p) and \control{(n34.south east)}{(\er:\p)};
\draw (n33.north west) .. controls +(\sl:\p) and \control{(n34.south west)}{(\el:\p)};
\draw (n34.north east) .. controls +(\sr:\p) and \control{(n35.south east)}{(\er:\p)};
\draw (n34.north west) .. controls +(\sl:\p) and \control{(n35.south west)}{(\el:\p)};
\draw (nu1.north east) .. controls +(\sr:\p) and \control{(nu2.north west)}{(\sl:\p)};
\draw (nu1.south east) .. controls +(\er:\p) and \control{(nu2.south west)}{(\el:\p)};
\draw (nu2.north east) .. controls +(\sr:\p) and \control{(nu3.north west)}{(\sl:\p)};
\draw (nu2.south east) .. controls +(\er:\p) and \control{(nu3.south west)}{(\el:\p)};
\draw (n13.north east) .. controls +(\sr:\p) and \control{(n24.north west)}{(\sl:\p)};
\draw (n24.north east) .. controls +(\sr:\p) and \control{(n35.north west)}{(\sl:\p)};
\draw (n11.south east) .. controls +(\er:\p) and \control{(n21.south west)}{(\el:\p)};
\draw (n21.south east) .. controls +(\er:\p) and \control{(n31.south west)}{(\el:\p)};
\draw (n35.north east) .. controls +(\sr:\p) and \control{(nu3.south east)}{(\er:\p)};
\draw (nu3.north east) .. controls +(\sr:4*\p) and \control{(n31.south east)}{(\er:4*\p)};
\draw (nu1.north west) .. controls +(\sl:4*\p) and \control{(n11.south west)}{(\el:4*\p)};
\draw (n13.north west) .. controls +(\sl:2*\p) and \control{(nu1.south west)}{(\el:2*\p)};

\draw (n11.south west) -- (n11.north east);
\draw (n12.south west) -- (n12.north east);
\draw (n13.south west) -- (n13.north east);
\draw (n21.south east) -- (n21.north west);
\draw (n22.south east) -- (n22.north west);
\draw (n23.south east) -- (n23.north west);
\draw (n24.south east) -- (n24.north west);
\draw (n31.south west) -- (n31.north east);
\draw (n32.south west) -- (n32.north east);
\draw (n33.south west) -- (n33.north east);
\draw (n34.south west) -- (n34.north east);
\draw (n35.south west) -- (n35.north east);
\draw (nu1.north west) -- (nu1.south east);
\draw (nu2.north west) -- (nu2.south east);
\draw (nu3.north west) -- (nu3.south east);
\end{scope}

%decomp
\begin{scope}[purple]
\draw (-0.55,-0.25) rectangle (1.5,4.25);
\draw (0.5,-0.25) -- (0.5,4.25);

\fill [green!70!purple] (0,1.5) circle (1.5pt);
\fill [blue] (1,2) circle (1.5pt);
\fill [brown] (2,2.5) circle (1.5pt);
\fill (0.5,5) circle (1.5pt);
\draw (-0.65,-0.35) rectangle (1.6,4.5);
\draw (1.6,4.5) -- (2.45,4.5) -- (2.45,-0.35) -- (1.6,-0.35);
\begin{scope}[dotted]
\draw (-0.75,-0.45) rectangle (2.55,4.6);
\draw (-0.6,-0.3) rectangle (1.55,4.4);
\draw (-0.3,1.3) rectangle (0.3,1.7);
\draw (-0.37,1) rectangle (0.37,2.5);
\draw (-0.42,0.8) rectangle (0.42,3.5);
\draw (0.7,1.8) rectangle (1.3,2.2);
\draw (0.63,0.7) rectangle (1.37,3.2);
\draw (0.58,0.2) rectangle (1.42,3.8);
\draw (1.7,2.3) rectangle (2.3,2.7);
\draw (1.63,0.7) rectangle (2.37,3.7);
\draw (0.4,4.7) rectangle (0.6,5.3);
\draw (0,4.63) rectangle (1.7,5.4);
\end{scope}
\end{scope}

%Tree of the sweeping
\draw [->] (3.2,2.5) -- (6,2.5) node [above, align = center, midway] {Sphere \\ decomposition};
\begin{scope}[purple]
\draw (6,1.5) -- (6.5,2.5) -- (7,2);
\draw (6.5,2.5) -- (7.25,3.5) -- (8,2.5);
\draw (6.5,5) -- (7.25,3.5);
\end{scope}
\begin{scope}[xshift = 6cm]
\fill [green!70!purple] (0,1.5) circle (1.5pt);
\fill [blue] (1,2) circle (1.5pt);
\fill [brown] (2,2.5) circle (1.5pt);
\fill [purple](0.5,5) circle (1.5pt);
\end{scope}
\end{tikzpicture}
\caption{A width-$4$ sphere decomposition of a pretzel knot.}
\end{center}
\end{figure}
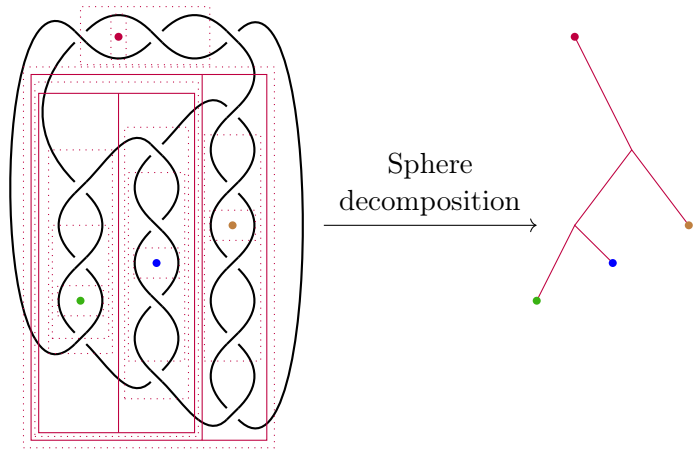

 The point of using thickened embeddings instead of usual ones is that this allows disjoint spheres of a sphere decomposition to intersect a same vertex of a graph embedding. This is motivated by the following proposition, which provides a bridge between sphere decompositions and tree decompositions of diagrams of knots and spatial graphs.

\begin{proposition}\label{P:swtw}
Let $G$ be a knot or a graph embedded in $\Sp^3$, and $D$ be a diagram of $G$. Then the spherewidth of $G$ is at most twice the tree-width of $D$.
\end{proposition}

The proof is very similar to that of Lemma~3.4 in~\cite{Mesmay_treewidth}.

\begin{proof}[Proof of Proposition~\ref{P:swtw}]
  If $D$ is a tree, then the spherewidth of $G$ is one and the proposition is trivial. So we assume henceforth that $G$ contains at least one cycle. We will refrain from providing a precise definition of treewidth (we refer to Diestel~\cite{Diestel_Graph_Theory}) as we will rely on a variant with more structure, more adapted to planar graphs: sphere-cut decompositions~\cite{dorn2005efficient}, which we first introduce. Given a planar graph $D$, a \bfdex{sphere-cut decomposition} is a trivalent tree $T$, whose leaves are in bijection with the edges of $G$, and a family of Jordan curves $\gamma_e$ parameterized by the edges of $T$. An edge $e$ of $T$ partitions the leaves of $T$ and thus the edges of $D$ into two subsets $E^e_1$ and $E^e_2$. Each Jordan curve $\gamma_e$ is required to intersect $D$ only at its vertices (such a curve is often called a \bfdex{noose}), and thus partitions the edges of $D$ into two subsets: we require that this partition matches the partition $E^e_1$ and $E^e_2$. The \bfdex{weight} of a Jordan curve $w(\gamma_e)$ is the number of vertices that it intersects. The width of a sphere-cut decomposition $w(T,(\gamma_e))$ is the maximum weight $w(\gamma_e)$ over all the edges of $T$, and the sphere-cut width is the minimum width over all sphere-cut decompositions. Given a graph $D$ that is not a tree, if we denote by $k$ its treewidth, it is known that the sphere-cut width of $D$ is at most $k$: this follows from the fact that the branchwidth of $D$ is at most its treewidth~\cite{Graph_Minors_X} and the results of Seymour and Thomas on bond carving decompositions~\cite{Seymour_Ratcatcher} (see for example~\cite[Lemma~2.2]{twinwidth} for the exact statement that we use connecting branchwidth and sphere-cut width).

  We now work with a sphere-cut decomposition $(T,(\gamma_e))$ achieving the optimal width. There are two types of vertices in $D$: those corresponding to vertices of $G$ and those corresponding to crossings of $G$ in the projection. Since $G$ is a thickened embedding in $\Sp^3$, it makes sense to also thicken $D$, i.e., we thicken a bit the vertices of the planar graph $D$ so that they are small disks, and so that the Jordan curves $\gamma_e$ intersect the interior of these disks. Then, by perturbing them a little and if necessary removing bigons, we can assume that the Jordan curves are all pairwise disjoint, while still intersecting the graph $D$ only at the vertex disks. Now, since $D$ is a planar projection of $G$, each Jordan curve can be lifted to a sphere in $\Sp^3$ by capping it off with one disk above and one disk below the projection. The resulting spheres, which we denote by $S_e$, are pairwise disjoint and intersect the graph $G$ in two different ways: whenever $\gamma_e$ traverses one of the vertices of $G$, $S_e$ intersects the corresponding thickened vertex; and whenever $\gamma_e$ traverses one of the crossings of $G$, $S_e$ intersects transversely one or two of the edges of $G$ (see Figure~\ref{pic_gamma_G}). 
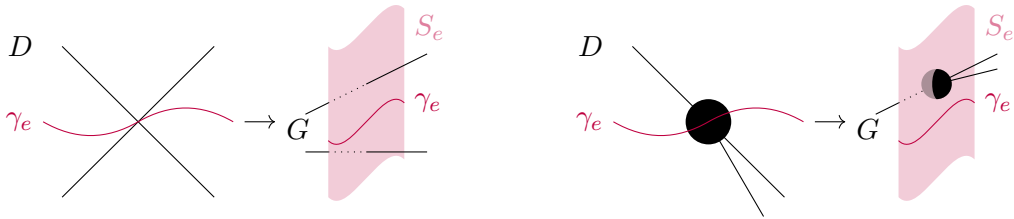
\begin{figure}[h]
\begin{center}
\begin{tikzpicture}[scale = 1]
\draw (-1,1) -- (1,-1);
\draw (-1,-1) -- (1,1);
\draw [purple] (-1.25,0) to [out=-30, in=-150] (0,0) to [out=30, in=150] (1.25,0);
\node at (-1.25,0) [left,purple] {$\gamma_e$};
\node at (-1.25,1) [left] {$D$};

\draw [->] (1.4,0) -- (1.8,0);

\fill [opacity = 0.2, purple] (2.5,-1) to [out=-45, in=-135] (3,-0.75)  to [out=45, in=135] (3.5,-0.5) -- (3.5,1.5) to [out=135, in=45] (3,1.25)  to [out=-135, in=-45] (2.5,1) -- cycle;
\draw [purple] (2.5,-0.25) to [out=-45, in=-135] (3,0)  to [out=45, in=135] (3.5,0.25);
\node at (3.5,0.25) [right,purple] {$\gamma_e$};
\node at (2.1,-0.1) {$G$};
\node at (3.5,1.25) [right,purple!50!] {$S_e$};
\draw (2.2,0.1) -- (2.5,0.25);
\draw [dotted](2.5,0.25) -- (3,0.5);
\draw (3,0.5) -- (3.8,0.9);
\draw (2.2,-0.4) -- (2.5,-0.4);
\draw [dotted](2.5,-0.4) -- (3,-0.4);
\draw (3,-0.4) -- (3.8,-0.4);

\begin{scope}[xshift = 7.5cm ]
\draw (-1,1) -- (1,-1);
\draw (0,0) -- (-60:1.45);
\fill (0,0) circle (0.3cm);
\draw [purple] (-1.25,0) to [out=-30, in=-150] (0,0) to [out=30, in=150] (1.25,0);
\node at (-1.25,0) [left,purple] {$\gamma_e$};
\node at (-1.25,1) [left] {$D$};

\draw [->] (1.4,0) -- (1.8,0);

\fill [opacity = 0.2, purple] (2.5,-1) to [out=-45, in=-135] (3,-0.75)  to [out=45, in=135] (3.5,-0.5) -- (3.5,1.5) to [out=135, in=45] (3,1.25)  to [out=-135, in=-45] (2.5,1) -- cycle;
\draw [purple] (2.5,-0.25) to [out=-45, in=-135] (3,0)  to [out=45, in=135] (3.5,0.25);
\node at (3.5,0.25) [right,purple] {$\gamma_e$};
\node at (2.1,-0.1) {$G$};
\node at (3.5,1.25) [right,purple!50!] {$S_e$};
\draw (2.2,0.1) -- (2.5,0.25);
\draw [dotted](2.5,0.25) -- (2.85,0.425);
\draw (3,0.5) -- (3.8,0.9);
\draw (3,0.5) -- (3.8,0.7);
\fill [opacity = 0.3] (3,0.5) circle (0.2cm);
\fill (3,0.3) arc (-90:90:0.2) .. controls +(-140:0.1) and \control{(3,0.3)}{(140:0.1)};
\end{scope}
\end{tikzpicture}
\caption{The two cases for the intersection between $\gamma_e$ and $D$ and their lift in $\Sp^3$.}
\label{pic_gamma_G}
\end{center}
\end{figure}

Therefore, the weight of a sphere $S_e$ is upper bounded by $c^e_1+2c^e_2$, where $c^e_1$ is the number of vertices of $G$ traversed by $\gamma_e$, and $c^e_2$ is the number of crossings of $G$ traversed by $\gamma_e$. The spheres $S_e$ are used as the backbone of a sphere decomposition, there just remains to interpolate between them:

  \begin{itemize}
\item A sphere $S_e$ for $e$ adjacent to a leaf of $T$ encloses exactly one edge of $D$, therefore it bounds a $G$-trivial ball. It is straightforward to define a continuous family of nested spheres of width $2$ and converging to a point disjoint from $G$, which together sweep such a $G$-trivial ball.
\item A vertex $v$ of $T$ is adjacent to three edges, corresponding to three spheres $S_{1}$, $S_{2}$ and $S_{3}$. The space inbetween these three spheres is homeomorphic to a solid pair of pants $P$, i.e., $\Sp^3$ with three balls removed. We look at $P \cap G$ and observe that since every edge of $D$ is in bijection with one of the leaves of $T$, it is enclosed by a ball in the previous item, and thus does not appear in $P\cap G$. Therefore, the only components that we see in $P\cap G$ are preimages under the projection map of thickened vertices or crossings. By construction, both these preimages are topologically trivial (they do not involve any knotting). The thickened vertices connect one, two or the three boundaries of $P$. So we can define a double bubble $DB$ in the middle of $P$ and transverse to $G$, and three families of nested spheres interpolating between the three spheres of $DB$ and $S_1$, $S_2$, and $S_3$, with the following property: each of these interpolating spheres intersects at most twice each thickened vertex connecting the three boundaries, and at most once any of the other preimages (see Figure~\ref{pic_solid_pant}). 
  \end{itemize}

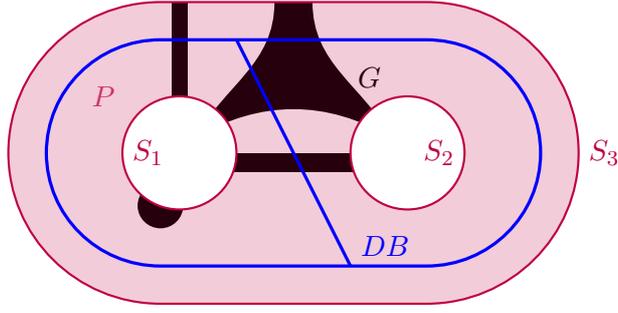
\begin{figure}[h]
\begin{center}
\begin{tikzpicture}
%Vertices
\begin{scope}[very thick]
\fill (-1.75,-0.25) rectangle (1.75,0);
\fill (-1.6,2) rectangle (-1.4,0);
\fill (-1.75,-0.7) circle (0.3cm);
\end{scope}
\fill (-1.3,0.2) to [out=30,in=150] (1.3,0.2) to [out=120,in=-90] (0.25,2) -- (-0.25,2) to [out=-90,in=60] (-1.3,0.2);
\node at (1,1) {$G$};

%Solid pant
\begin{scope}[thick, purple]
\filldraw [fill opacity = 0.2] (-1.75,-2) arc (-90:-270:2) -- (1.75,2) arc (90:-90:2) -- cycle;
\filldraw [fill = white] (-1.5,0) circle (0.75 cm);
\filldraw [fill = white] (1.5,0) circle (0.75 cm);
\end{scope}
\node at (-2.5,0.75) [purple!80!] {$P$};
\node at (-2.25,0) [purple, right] {$S_1$};
\node at (2.25,0) [purple, left] {$S_2$};
\node at (3.75,0) [purple, right] {$S_3$};

%Double Bubble 
\begin{scope}[very thick, blue]
\draw (-1.75,-1.5) arc (-90:-270:1.5) -- (1.75,1.5) arc (90:-90:1.5) -- cycle;
\draw (-0.75, 1.5) -- (0.75, -1.5);
\end{scope}
\node at (0.75,-1.5) [above right, blue] {$DB$};
\end{tikzpicture}
\caption{Definition of a double bubble $DB$ within a solid pant $P$.}
\label{pic_solid_pant}
\end{center}
\end{figure}

Each sphere involved in the previous two items has weight at most $2c^e_1 +2 c^e_2$ for some Jordan curve $\gamma_e$, and thus the width of the sphere decomposition if upper bounded by twice the sphere-cut width. We obtain a sphere decomposition of weight at most twice the treewidth of $D$, which concludes the proof.
\end{proof}

\subparagraph*{Bubble tangle.} Bubble tangles are our second main concept in this article. They will constitute an obstruction to spherewidth, by designating, for each sphere in $\Sp^3$ not intersecting the graph too many times, the side of the sphere that is easy to sweep. We first observe that some balls have to be easy to sweep: intuitively this will be the case of any unknotted segment or empty ball (see Figure~\ref{pic_ex_G-trivial}).
Let $G$ be a graph embedded in $\Sp^3$. A closed ball $B$ in $\Sp^3$ is said to be \bfdex{$G$-trivial} if its boundary is transverse to $G$ and one of the following holds (where $B(0,1)$ is the unit ball of $\mathbb{R}^3$):
\begin{itemize}
	\item $B \cap G = \varnothing$.
	\item $B \smallsetminus G$ is homeomorphic to $
	B(0,1) \smallsetminus [-1,1]\times \ens{(0,0)} \subset \mathbb{R}^3$.
	\item $B \smallsetminus G$ is homeomorphic to $
	B(0,1) \smallsetminus [-1,0]\times \ens{(0,0)} \subset \mathbb{R}^3$.
\end{itemize} 

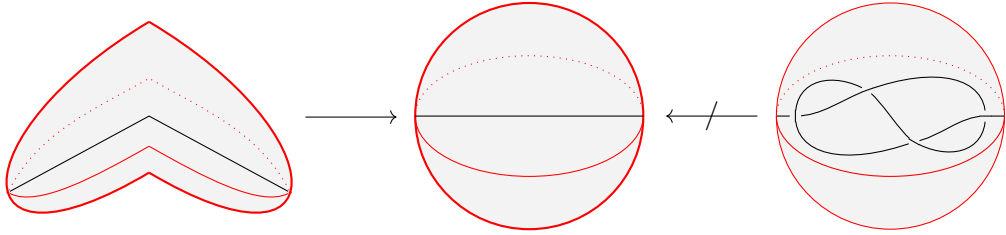
\begin{figure}[ht]
\begin{center}
%\tikzsetnextfilename{G-trivial_ball}
\begin{tikzpicture}
\draw [red, thick] (0,2) .. controls +(-2.5,-1.5) and \control{(0,0)}{(-2.5,-1.5)} node (nl) [pos = 0.55, fill, circle, inner sep = 0 pt] {};
\draw [red, thick] (0,2) .. controls +(2.5,-1.5) and \control{(0,0)}{(2.5,-1.5)} node (nr) [pos = 0.55, fill, circle, inner sep = 0 pt] {};
\fill [opacity = 0.05] (0,2) .. controls +(-2.5,-1.5) and \control{(0,0)}{(-2.5,-1.5)} .. controls ++(2.5,-1.5) and \control{(0,2)}{(2.5,-1.5)};
\draw [red] (nl) .. controls +(0.4,-0.2) and \control{(0,0.35)}{(-0.7,-0.4)};
\draw [red, dotted] (nl) .. controls +(0.4,0.8) and \control{(0,1.25)}{(-0.7,-0.4)};
\draw [red] (nr) .. controls +(-0.4,-0.2) and \control{(0,0.35)}{(0.7,-0.4)};
\draw [red, dotted] (nr) .. controls +(-0.4,0.8) and \control{(0,1.25)}{(0.7,-0.4)};
\draw (nl) -- (0,0.75) -- (nr);
\draw [->] (nr) ++(0.2,1) -- +(1.2,0);
\draw [red, thick] (5,0.75) circle (1.5);
\fill [opacity = 0.05] (5,0.75) circle (1.5);
\draw [red] (6.5,0.75) arc (0: -180: 1.5 and 0.8);
\draw [red, dotted] (6.5,0.75) arc (0:180: 1.5 and 0.8);
\draw (3.5, 0.75) -- (6.5,0.75);
\draw [<-] (6.8,0.75) -- +(1.2,0) node [midway] {$/$};
\begin{scope}[xshift = 8.5cm, yshift = 0.75 cm]
\coordinate (a) at (0,0);
\coordinate (a') at ($(a)+(-0.25,0)$);
\coordinate (b) at (20:1);
\coordinate (c) at (2.5,0);
\coordinate (c') at ($(c)+(0.25,0)$);
\coordinate (d) at ($(200:1)+(c)$);
\draw (a) -- (a');
\draw (c) -- (c');
\draw (a) .. controls +(0:0.35) and \control{(b)}{(200:0.35)};
\draw (a) .. controls +(90:0.5) and \control{(b)}{(150:0.5)};
\draw (b) .. controls +(20:1) and \control{(c)}{(90:0.5)};
\draw (c) .. controls +(180:0.35) and \control{(d)}{(10:0.35)};
\draw (c) .. controls +(0,-0.5) and \control{(d)}{(-30:0.5)};
\draw (d) .. controls +(200:1) and \control{(a)}{(-90:0.5)};
\draw (b) .. controls +(-30:0.2) and \control{(d)}{(150:0.2)};
\draw [red] (1.25,0) circle (1.5);
\draw [red] (2.75,0) arc (0: -180: 1.5 and 0.8);
\draw [red, dotted] (2.75,0) arc (0: 180: 1.5 and 0.8);
\node [fill,white, circle, inner sep = 1.5pt] (na) at (a) {};
\node [fill,white, circle, inner sep = 1.5pt, rotate = 20] (nb) at (b) {};
\node [fill,white, circle, inner sep = 1.5pt] (nc) at (c) {};
\node [fill,white, circle, inner sep = 1.5pt, rotate = -30] (nd) at (d) {};
\fill [opacity = 0.05] (1.25,0) circle (1.5);
\draw (na.south) -- (na.north);
\draw (nb.west) -- (nb.east);
\draw (nc.west) -- (nc.east);
\draw (nd.west) -- (nd.east);
\end{scope}
\end{tikzpicture}
\caption{\label{pic_ex_G-trivial} Representation of a $G$-trivial ball and a non $G$-trivial ball.}
\end{center}
\end{figure}

We can now introduce bubble tangles.

\begin{defn}
Let $G$ be an embedding of a graph in $\Sp^3$ and $n \in \N$. A \bfdex{bubble tangle $\mathcal{T}$ of order $n\geq 2$}, is a collection of closed balls in $\Sp^3$ such that: 
\begin{enumerate}[label = (T\arabic*)]
	\item \label{def_T1} For every closed ball $B$ in $\mathcal{T}$, $|C(\partial B \cap G)| < n$.
	\item \label{def_T2} For every sphere $S$ in $\Sp^3$ transverse to $G$, if $|C(S \cap G)| < n$ then exactly one of the two closed balls $\bar B_1$ is in $\T$ or $\bar B_2$ is in $\T$, where $\Sp^3 \smallsetminus S = \ens{B_1, B_2}$. 
	\item \label{def_T3} For every triple of balls $B_1,B_2$ and $B_3$ induced by a double bubble transverse to $G$, $\ens{B_1, B_2, B_3} \not \subset \T$.  
	\item \label{def_T4} For every closed ball $B$ in $\Sp^3$, if $B$ is $G$-trivial and $|C(\partial B \cap G)| < n,$ then $B \in \T$.
\end{enumerate}
\end{defn}

For every transverse sphere $S$ such that $|C(S \cap G)| < n$, a bubble tangle chooses one of the two balls having $S$ as the boundary. We think of the ball in $\T$ as being a ``small side'', since~\ref{def_T4} stipulates that balls containing trivial parts of $G$ are in $\T$, while the other one is the ``big side''. Then the key property~\ref{def_T3} enforces that no three small sides forming a double bubble should cover the entire $\Sp^3$.

\begin{remark} Tangles in graph theory are often endowed with an additional axiom, specifying that small sides should be stable under inclusion (see e.g.,~\cite[Axiom~(T3A)]{Geelen_Matroid}). Our bubble tangles are weaker in the sense that we do not enforce this axiom, but still strong enough to guarantee duality (Theorem~\ref{T:duality}) and the connection to compression-representativity (Theorem~\ref{T:representativity}). Whether such an axiom can be additionally enforced in our definition of bubble tangle while preserving these properties is left as an open problem.
\end{remark}

\section{Obstruction and duality}\label{S:duality}

In this section, we prove Theorem~\ref{T:duality}: given a graph $G$ embedded in $\Sp^3$, the highest possible order of a bubble tangle is equal to the spherewidth of $G$. In the following, $G$ is an embedding of a graph in $\Sp^3$ and the order of all bubble tangles that we consider is at least $3$, the theorem being trivial otherwise. The proof is split into two inequalities: Proposition~\ref{prop_max_tangle_l_sw} and Proposition~\ref{P:nogapside} which together immediately imply Theorem~\ref{T:duality}.

\subparagraph*{Bubble tangles as obstruction.} We first show that a bubble tangle of order $k$ and a sphere decomposition of width less than $k$ cannot both exist at the same time.

\begin{proposition}\label{prop_max_tangle_l_sw}
Let $G$ be an embedding of a graph in $\Sp^3$. If $G$ admits a bubble tangle $\T$ of order $k$ then $sw(G) \geq k$.
\end{proposition}

The proof of this proposition is similar to its graph-theoretical counterparts showing that tangles are an obstruction to branchwidth (see, e.g.,~\cite{Graph_Minors_X}). The main difference with these proofs lies in the continuous aspects of our sphere decomposition, which we control using Lemmas~\ref{lem_braid_equi} and~\ref{lem_orientable_edges}.

Let $S$ and $S'$ be two disjoint spheres in $\Sp^3$. Then $\Sp^3 \smallsetminus (S \cup S')$ has three connected components: two balls and a space $I$ homeomorphic to $\Sp^2 \times [0,1]$. The spheres $S$ and $S'$ are said to be \bfdex{braid-equivalent} if $(I \cup S \cup S') \smallsetminus G$ is homeomorphic to $S_k \times [0,1]$ where $S_k$ is the $2$-sphere with $k$ holes. The intuition behind this definition is that it means that $G$ forms a braid between $S$ and $S'$. The following lemma explains how braid-equivalent spheres interact with a bubble tangle.

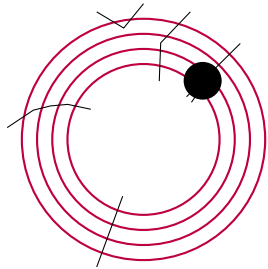
\begin{figure}[h]
\begin{center}
%\tikzsetnextfilename{braid_equi_def}
\begin{tikzpicture}
\draw [purple, thick] (0,0) circle (1cm);
\draw [purple, thick] (0,0) circle (1.2cm);
\draw [purple, thick] (0,0) circle (1.4cm);
\draw [purple, thick] (0,0) circle (1.6cm);
\fill (45:1.1) circle (0.25cm);
\draw (45: 0.8) -- (45:1.8);
\draw (40:0.9) -- (38:0.8);
\draw (75:0.8) -- (80 : 1.3) -- (70:1.8);
\draw (150:0.8) -- (155 : 1.1) -- (160:1.3) -- (165:1.5) -- (175:1.8);
\draw (-110:0.8) -- (-110:1.8);
\draw (110:1.8) -- (100:1.5) -- (90:1.8);
\end{tikzpicture}
\label{pic_braid_equi_def}
\caption{The three innermost spheres are braid-equivalent, not the fourth one.}
\end{center}
\end{figure}

\begin{lem}\label{lem_braid_equi}
Let $\T$ be a bubble tangle and $S,S'$ be two braid-equivalent spheres. Let us write $\Sp^3 \smallsetminus S = \ens{B_1,B_2}$ and $\Sp^3 \smallsetminus S' = \ens{B'_1,B'_2}$ such that $B_1 \subset B'_1$. If $B_1 \in \T$ then $B'_1 \in \T$.
\end{lem}

\begin{proof}[Proof of Lemma~\ref{lem_braid_equi}]
 The idea of the proof is to show that if $S$ and $S'$ are braid-equivalent, we can cover the space $I = B'_1 \smallsetminus  \mathring B_1$ by $G$-trivial balls in order to "grow" $S$ into $S'$, which will yield the result by invoking~\ref{def_T3} inductively. Notice that braid equivalence implies that $S$ and $S'$ are transverse to $G$, and that $|C(S \cap G)| = |C(S' \cap G)| = k$ where $k$ is less than the order of $\T$ as $S \in \T$; hence, by \ref{def_T2} either $B'_1 \in \T$ or $B'_2 \in \T$. Let $h :S_k \times [0,1] \rightarrow I \smallsetminus G$ be a homeomorphism where $S_k$ is the $2$-sphere with $k$ holes $t_1, \ldots, t_k$, and let us work on $S_k \times [0,1]$.

Let $D_1, \ldots, D_k$ be $k$ disjoint closed disks on $\Sp^2$, each one covering a hole of $S_k$: for all $i \in \inter{1}{k}, \mathring D_i \smallsetminus t_i$ is an open annulus. For $i \in \inter{1}{k}$ notice that $h((D_i \smallsetminus t_i) \times [0,1])$ is homeomorphic to $B(0,1) \smallsetminus [-1,1]\times \ens{(0,0)}$ (see Figure~\ref{pic_disks}).

\begin{figure}[h]
\begin{center}
%\tikzsetnextfilename{braid_equi_lem}
\begin{tikzpicture}
\filldraw [fill opacity = 0.4, purple] (0,0) circle (1.3cm);
\coordinate (ct1) at (45:0.75);
\coordinate (ct2) at (-5:0.65);
\coordinate (ct3) at (125:1);
\coordinate (ct4) at (125:0.55);
\coordinate (ct5) at (15:0.15);
\coordinate (ct6) at (-85:0.9);
\foreach \i in {1,2,...,6}{
	\node [inner sep = 1pt, circle, fill, black] at (ct\i) {};
	\node [right] at (ct\i) {$t_\i$};
};

\begin{scope}[xshift = 3.5cm]
\filldraw [fill opacity = 0.4, purple] (0,0) circle (1.3cm);
\coordinate (ct1) at (45:0.75);
\coordinate (ct2) at (-5:0.65);
\coordinate (ct3) at (125:1);
\coordinate (ct4) at (125:0.55);
\coordinate (ct5) at (15:0.15);
\coordinate (ct6) at (-85:0.9);
\foreach \i in {1,2,...,6}{
	\node [inner sep = 1pt, circle, fill, black] at (ct\i) {};
};
\foreach \i in {3,4,5}{
	\node [draw, inner sep = 3pt, circle] (t\i) at (ct\i) {};
	\node [fill, inner sep = 3pt, opacity = 0.2, circle] at (ct\i) {};
	\node [below left] at (ct\i) {$D_\i$};
};
\foreach \i in {1,2,6}{
	\node [draw, inner sep = 3pt, circle] (t\i) at (ct\i) {};
	\node [fill, inner sep = 3pt, opacity = 0.2, circle] at (ct\i) {};
	\node [right] at (ct\i) {$D_\i$};
};
\end{scope}
\end{tikzpicture}
\caption{Disks $D_1, \ldots, D_k$ over holes of $\Sp^2$.}
\label{pic_disks}
\end{center}
\end{figure}
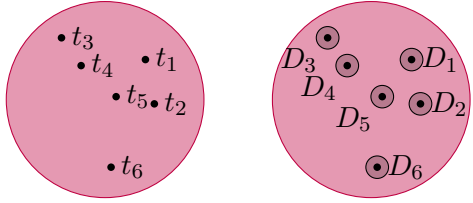

We first extend $h^{-1}$ to a continuous function $\varphi : I \rightarrow \Sp^2 \times [0,1]$ such that  $\varphi(D_i \times [0,1])$ is a closed ball of $S^3$ that is $G$-trivial, and belongs to $\T$ by $\ref{def_T4}$. It suffices to extend $h^{-1}$ on $G$. Notice that $G$ is a union of balls and edges with endpoints on these balls. On $I$, the closure of each $h(S_k \times \ens{x})$ for $x \in [0,1]$ naturally extends $h^{-1}$ to $\partial G \cap I$ by continuity: send the boundary components of $h(S_k \times \ens{x})$ to $\ens{t_i} \times \ens{x}$ with $\varphi$ when the boundary component extends the image of a neighborhood of $\ens{t_i} \times \ens{x}$. There only remains to deal with the balls of the thickened embedding. Let $B$ be a vertex ball of $G$ such that $I \cap B \neq \varnothing$ and let $X$ be one of the connected components of $B \cap I$. We layer $X$ continuously with disks such that each disk is bounded by a circle induced by one of the boundary components of  $h(S_k \times \ens{x})$. Then we send each disk to $\ens{t_i}\times \ens{x}$ with $\varphi$. 

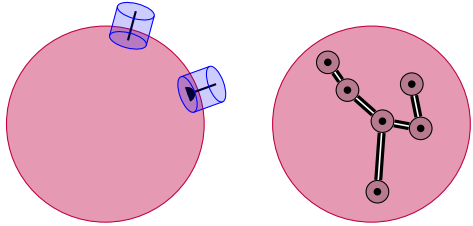
\begin{figure}[h]
\begin{center}
%\tikzsetnextfilename{braid_equi_lem2}
\begin{tikzpicture}
\begin{scope}[xshift = 3.5cm]
\filldraw [fill opacity = 0.4, purple] (0,0) circle (1.3cm);
\coordinate (ct1) at (45:0.75);
\coordinate (ct2) at (-5:0.65);
\coordinate (ct3) at (125:1);
\coordinate (ct4) at (125:0.55);[double, very thick] 
\coordinate (ct5) at (15:0.15);
\coordinate (ct6) at (-85:0.9);
\foreach \i in {1,2,...,6}{
	\node [inner sep = 1pt, circle, fill, black] at (ct\i) {};
};
\foreach \i in {3,4,5}{
	\node [draw, inner sep = 3pt, circle] (t\i) at (ct\i) {};
	\node [fill, inner sep = 3pt, opacity = 0.2, circle] at (ct\i) {};
};
\foreach \i in {1,2,6}{
	\node [draw, inner sep = 3pt, circle] (t\i) at (ct\i) {};
	\node [fill, inner sep = 3pt, opacity = 0.2, circle] at (ct\i) {};
};
\draw [double, very thick] (t1) -- (t2) -- (t5) -- (t6);
\draw [double, very thick] (t3) -- (t4) -- (t5);
\end{scope}

\filldraw [fill opacity = 0.4, purple] (0,0) circle (1.3cm);
\draw [thick] (75:1.15) -- (75:1.55);
\draw [blue, rotate around = {-15: (75:1.55)}] (75:1.55) circle (0.25 and 0.1);
\draw [blue, rotate around = {-15: (75:1.15)}] (75:1.15) circle (0.25 and 0.1);
\fill [blue, opacity = 0.2] [rotate around = {-15: (75:1.15)}] (75:1.15) ++ (0:0.25) arc  (0:-180: 0.25 and 0.1) -- ++(90:0.4) arc (180:0: 0.25 and 0.1) -- cycle;
\draw [blue] [rotate around = {-15: (75:1.15)}] (75:1.15) ++ (0:0.25) -- ++(90:0.4);
\draw [blue] [rotate around = {-15: (75:1.15)}] (75:1.15) ++ (0:-0.25) -- ++(90:0.4);

\begin{scope}[rotate = -55]
\draw [thick] (75:1.15) -- (75:1.55);
\fill [thick] [rotate around = {-15: (75:1.15)}] (75:1.15) ++ (0:0.1) arc (0:180:0.1) -- cycle;
\draw [blue, rotate around = {-15: (75:1.55)}] (75:1.55) circle (0.25 and 0.1);
\draw [blue, rotate around = {-15: (75:1.15)}] (75:1.15) circle (0.25 and 0.1);
\fill [blue, opacity = 0.2] [rotate around = {-15: (75:1.15)}] (75:1.15) ++ (0:0.25) arc  (0:-180: 0.25 and 0.1) -- ++(90:0.4) arc (180:0: 0.25 and 0.1) -- cycle;
\draw [blue] [rotate around = {-15: (75:1.15)}] (75:1.15) ++ (0:0.25) -- ++(90:0.4);
\draw [blue] [rotate around = {-15: (75:1.15)}] (75:1.15) ++ (0:-0.25) -- ++(90:0.4);
\end{scope}
\end{tikzpicture}
\caption{A $G$-trivial ball covering a strand and an example of tree connecting the disks $D_i$.}
\label{pic_braid_equi_lem}
\end{center}
\end{figure}

Now $\varphi^{-1} (D_i \times [0,1])$ is a closed ball $\beta_i$ in $\Sp^3$ by construction, and is $G$-trivial as $\varphi^{-1} (D_i \times [0,1]) \smallsetminus G = h((D_i \smallsetminus t_i) \times [0,1])$. Hence $\beta_i$ belongs to $\T$. Notice that $B_1$ intersects $\beta_i$ on a disk $D_i$, so that these balls induce a double bubble. It follows by \ref{def_T3} that $B_1 \cup \beta_1$ belongs to $\T$. As the balls $D_i$ are disjoint, $(B_1 \cup \bigcup_{j \in \inter{1}{i}}{\beta_j}) \cap \beta_{i+1} = D_{i+1}$. Hence, these balls induce a double bubble. By induction, we obtain that $B_1 \cup \bigcup_{j \in \inter{1}{k}}{\beta_j}$ is in $\T$. \\
At that point, we have that $(B_1 \cup \bigcup_{j \in \inter{1}{k}}{\beta_j}) \cap S' = h(\bigcup_{j \in \inter{1}{k}} D_i \times \ens{1})$. Now, we consider a tree $T$ connecting the disks $D_i$ as in the picture above. Then, we thicken the edges of that tree to replace each edge $e$ joining two disks with a band $r_e$. It follows that each $h(r_e \times [0,1])$ is a closed ball in $\Sp^3$ that intersects $B_1 \cup \bigcup_{j \in \inter{1}{k}}{\beta_j}$ on a $U$ shaped disk and  is disjoint from $G$ (thus it is $G$-trivial). Using both \ref{def_T3} and \ref{def_T4} on each of the balls induced by the bands, we add balls to $B_1 \cup \bigcup_{j \in \inter{1}{k}}{\beta_j}$ to get $B \in \T$ such that $B$ intersects $S'$ on a single connected component: $h((\bigcup_{e \in T} r_e \cup \bigcup_{i \in \inter{1}{k}} D_i) \times \ens{1})$. 
Finally, we add to $B$ the ball $\overline{B'_1 \smallsetminus B}$ disjoint from $G$ which intersects $S'$ on a disk : the complementary of $h((\bigcup_{e \in T} r_e \cup \bigcup_{i \in \inter{1}{k}} D_i) \times \ens{1}$. Since its intersection with $B$ is a disk, they induce a double bubble and $B'_1 \in \T$ by \ref{def_T3}.
\end{proof}

In the following, we will assume that there exists a bubble tangle $\T$ of order $k$ and a sphere decomposition $f : \Sp^3 \rightarrow T$ of $G$ of width less than $k$ in order to reach a contradiction.
Let $e = (u,v) \in E(T)$ be an edge and $x$ be a point of $e$ so that $f^{-1}(x)$ is transverse to $G$. Notice that $x$ cuts $T$ in two trees : $T_u(x)$ and $T_v(x)$ where $T_u(x)$ is the tree containing the endpoint $u$. By definition $f^{-1} (x) = S$ is a sphere in $\Sp^3$ such that $|C(G \cap S)| < k$. It follows by \ref{def_T2} that exactly one of $f^{-1} (T_u(x))$ or $f^{-1} (T_v(x))$ belongs to $\T$. We define an \bfdex{orientation} $o: T \rightarrow V(T)$ induced by $\T$ as follows: if $f^{-1} (x)$ is transverse to $G$, $o(x):=v$ if $f^{-1} (T_u(x)) \in \T$, or $o(x):=u$ if $f^{-1} (T_v(x)) \in \T$. In other words, at a point $x$ where $f^{-1} (x)$ is transverse to $G$ the orientation $o$ orients $x$ outwards, toward the  ``big side''. If $f^{-1}(x)$ has a tangency with $G$, note that for any close enough neighbor $y$ of $x$, $f^{-1}(y)$ is transverse to $G$, and we define $o(x):=o(y)$, making an arbitrary choice if needed.
As we consider edges of the tree $T$ to be intervals, we will use interval notations: we write $[u,v]$ for the edge $(u,v)$, and more generally $[x,y]$ to describe all the points on the edge between $x$ and $y$. We say that an orientation $o$ is \bfdex{consistent} if for any $x$ on some edge such that $f^{-1}(x)$ is transverse to $G$, $o$ is constant on $[x,o(x)]$.
The following lemma shows that the orientation $o$ can be assumed to be consistent on all the edges of the tree $T$.

\begin{lem}\label{lem_orientable_edges}
Let us assume that there exists a bubble tangle $\T$ of order $k$ and a sphere decomposition $f : \Sp^3 \rightarrow T$ of $G$ of width less than $k$. Then there exists a sphere decomposition to the same tree such that $o$ is consistent on $T$.
\end{lem}

\begin{proof}[Proof of Lemma~\ref{lem_orientable_edges}]
  Lemma~\ref{lem_braid_equi} shows that orientations are consistent between braid-equivalent spheres. Therefore, in order to prove Lemma~\ref{lem_orientable_edges}, we want to control what happens outside of these braid-equivalences, that is, when the quantity $|C(f^{-1} (x) \cap G)|$ varies on an edge of $T$. These variations, local maxima and minima, occur only when spheres have a tangency with $G$, since two spheres transverse to $G$ and close enough are braid-equivalent. By applying a small perturbation to $f$ if needed, we can assume that the spheres of $f$ tangent to $G$ have at most one tangency. 

Let $e = [u,v]$ be an edge of $T$ and let $x$ be some point interior to it such that $f^{-1} (x)$ has no tangency with $G$. Up to switching $u$ and $v$, we can assume that $o(x) = v$. Then, we let $t_1, \ldots, t_n \in [x,v]^n$ be all the points in $[x,v]$ where $f^{-1}(t_j)$ has a tangency with $G$. By Lemma~\ref{lem_braid_equi}, for all $y \in [x,t_1)$, we have that $o(x) = o(y) =v$ since $f^{-1} (y)$ is braid-equivalent to $f^{-1}(x)$ (since all the intersections between the spheres and $G$ are transverse between $x$ and $y$, one can glue the local homeomorphisms to get one between $f^{-1} ([x,y])$ and $S_j \times [0,1]$ for some $j$). There only remains to show that $o(x) = o(y)$ for all $y \in (t_1, t_2)$ and we will be done by induction, since this will imply by our definition of $o$ that $o(t_1)$ is also equal to $o(x)$. 
 We set $S' = f^{-1} (t_1)$ and $S_y = f^{-1} (y)$ for $y \in (t_1,t_2)$.

\begin{itemize}
\item We first consider the case where there is a local maximum on $t_1$, i.e., locally around the tangency $\tau$ the embedding of $G$ has either the ``V'' shape described in Figure~\ref{pic_orient_1} (recall that both the spheres and the embedding of the graph are piecewise-linear) or $S_1$ is tangent to a vertex.
The idea is to deal first with this tangency and then with every other strand in a similar manner as in the proof of Lemma~\ref{lem_braid_equi}. As we will use an argument similar to the proof of Lemma~\ref{lem_braid_equi}, let us introduce similar notations: $f^{-1} (x) = S$,  $B_1 = f^{-1} (T_u(x))$, $f^{-1} (t_1) = S'$, and $I$ the space between $S$ and $S_y$: $I = \Sp^3 \smallsetminus (\mathring B_1 \cup \mathring f^{-1} (T_v(y)))$.

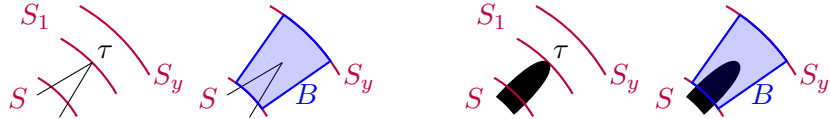
\begin{figure}[h]
\begin{center}
%\tikzsetnextfilename{orient_case_1_edge}
\begin{tikzpicture}[scale = 2]
\draw [purple, thick] (30:1) arc (30:60:1);
\draw [purple, thick] (30:0.7) arc (30:60:0.7);
\draw [purple, thick] (30:1.25) arc (30:60:1.25);
\draw (55:0.6) -- (45:1);
\draw (45:1 )-- (35:0.6);
\node [purple, above left] at (60:1) {$S_1$};
\node [purple, below left] at (60:0.7) {$S$};
\node at (45:0.95) [above right] {$\tau$};
\node [purple,right] at (30:1.2) {$S_y$};

\begin{scope}[xshift = 1.25 cm]
\draw [purple, thick] (30:0.7) arc (30:60:0.7);
\draw [purple, thick] (30:1.25) arc (30:60:1.25);
\draw (55:0.6) -- (45:1);
\draw (45:1 )-- (35:0.6);
\node [purple, below left] at (60:0.7) {$S$};
\node [purple,right] at (30:1.2) {$S_y$};
\filldraw [blue, fill opacity = 0.2, thick] (55:0.7) -- (55:1.25) arc (55:35:1.25) -- (35:0.7) arc (35:55:0.7) -- cycle;
\node [blue] at (30:1) {$B$};
\end{scope}

\begin{scope}[xshift = 3cm]
\draw [purple, thick] (30:1) arc (30:60:1);
\draw [purple, thick] (30:1.25) arc (30:60:1.25);
\node [purple, above left] at (60:1) {$S_1$};
\node [purple, below left] at (60:0.7) {$S$};
\node at (45:0.95) [above right] {$\tau$};
\node [purple,right] at (30:1.2) {$S_y$};
\fill (37:0.6) .. controls +(37:0.05) and \control{(45:1)}{(-45:0.1)} .. controls +(135:0.1) and \control{(53:0.6)}{(53:0.05)} arc (53:37:0.6);
\draw [purple, thick] (30:0.7) arc (30:60:0.7);
\end{scope}

\begin{scope}[xshift = 4.25 cm]
\draw [purple, thick] (30:1.25) arc (30:60:1.25);
\fill (37:0.6) .. controls +(37:0.05) and \control{(45:1)}{(-45:0.1)} .. controls +(135:0.1) and \control{(53:0.6)}{(53:0.05)} arc (53:37:0.6);
\draw [purple, thick] (30:0.7) arc (30:60:0.7);
\node [purple, below left] at (60:0.7) {$S$};
\node [purple,right] at (30:1.2) {$S_y$};
\filldraw [blue, fill opacity = 0.2, thick] (55:0.7) -- (55:1.25) arc (55:35:1.25) -- (35:0.7) arc (35:55:0.7) -- cycle;
\node [blue] at (30:1) {$B$};
\end{scope}
\end{tikzpicture}
\caption{The tangency $\tau$ on $S_1$ and a cover by a $G$-trivial ball.}
\label{pic_orient_1}
\end{center}
\end{figure}

Let $g$ be the connected component of $G \cap I$ containing $\tau$. Because $S$ and $G$ are transverse, $g \cap S$ is either a closed disk if $\tau$ belongs to a ball, a point if $\tau$ is a degree one vertex, or two points otherwise. Let us choose $D$ a closed disk on $S$ containing $g \cap S$ on its interior and disjoint from $G$ otherwise. We can now define $B$, a closed ball contained in $I$ intersecting $S$ on $D$, $S'$ on a disk, and containing $g$ in its interior (see Figure \ref{pic_orient_1}). Then $B$ is $G$-trivial. Indeed, either $\tau$ is the point of a ``V'' shape and $B \smallsetminus G$ is homeomorphic to $B(0,1) \smallsetminus [-1,1] \times \ens{(0,0)}$, or $B \smallsetminus G$ is homeomorphic to $B(0,1) \smallsetminus [-1,0]\times \ens{(0,0)}$ in the other cases (up to homeomorphism of $B \smallsetminus G$ the thickness of $g$ here does not matter, it can be a point or a segment or a closed ball). Furthermore, $B$ is by definition transverse to $G$ so that by \ref{def_T3}, $B \cup B_1$ belongs to $\T$.

At this stage,  the space $I \smallsetminus (B \cup G)$ is homeomorphic to $S_k \times [0,1]$ so that we can use a similar technique as in the proof of Lemma~\ref{lem_braid_equi} to deal with the other elements of $G$ in $I$ and cover what remains with balls disjoint from $B$, we apply~\ref{def_T3} inductively and conclude that $f^{-1} (y)$ belongs to $\T$.

\item If there is a local minimum on $S_1$, we let $\tau$ denote the point of tangency on $S_1$. Either $o(y) = o(x) = v$ and the proof is over, or $o(y) = u$. In the latter case, the setup is exactly the one of the previous case: $\tau$ is now a local maximum from $S_y$ to $S$ so that $o(x) = o(y) = u$ which is in contradiction with \ref{def_T2} on $f ^{-1}(x)$ by definition of $o$. 
\end{itemize}
\end{proof}

Lemma~\ref{lem_orientable_edges} ensures that for any edge $e=(u,v)$ of $T$, there exists a point $x_e$ so that all the points in $(x_e,v)$ are oriented towards $v$, while all the points in $(u,x_e)$ are oriented towards $u$. Hence, by subdividing each edge $e$ of $T$ at this $x_e$, we can think of $o$ as assigning a direction to each edge.
This directed tree is the main tool that we use in the proof of Proposition~\ref{prop_max_tangle_l_sw}.

\begin{proof}[Proof of Proposition~\ref{prop_max_tangle_l_sw}]
  Let us assume that there exists both a bubble tangle of order $k$ and a sphere decomposition $f:\Sp^3 \rightarrow T$ of width less than $k$.
  By Lemma~\ref{lem_orientable_edges}, there exists a sphere decomposition of width less than $k$ so that the orientation $o$ as defined above is consistent. Denoting by $T'$ the tree $T$ where each edge has been subdivided once, this orientation corresponds to a choice of direction for each edge of $T'$. Every directed acyclic graph, and thus in particular the tree $T'$ contains at least one sink, see Figure~\ref{pic_orient_half}.

  This sink cannot be a leaf of the tree. Indeed, let $e = [\ell,u]$ be an edge of $T$ incident to a leaf $\ell$. By definition, $f^{-1} (\ell)$ is a point disjoint from $G$, and thus for any $y$ in $(\ell, u)$ close enough to $\ell$, $f^{-1} (y)$ is a sphere disjoint from $G$. Hence $f^{-1} (T_\ell(y))$ is a $G$-trivial ball and belongs to $\T$. It follows that all edges incident to leaves of $T'$ are oriented inward. This sink cannot be a degree-two vertex either, as the tree $T'$ was defined in such a way that the two edges adjacent to a degree-two vertex are always oriented outwards. Finally, this sink cannot be a degree-three vertex as this would mean that the three balls induced by a double bubble are in $\T$, which would violate~\ref{def_T3}. We have thus reached a contradiction.\end{proof}

\begin{figure}[ht]
\begin{center}
%\tikzsetnextfilename{orient_half_tree}
\begin{tikzpicture}
\coordinate (d0) at (0,0);
\coordinate (d1) at (30:0.5);
\coordinate (d2) at (150:0.5);
\coordinate (d3) at (-90:0.5);
\coordinate (d4) at (30:1);
\coordinate (d5) at (150:1);
\coordinate (d6) at (-90:1);
\coordinate (d7) at ($(d4)+(90:0.5)$);
\coordinate (d8) at ($(d4)+(-30:0.5)$);
\coordinate (d9) at ($(d4)+(90:1)$);
\coordinate (d10) at ($(d4)+(-30:1)$);
\coordinate (d11) at ($(d5)+(90:0.5)$);
\coordinate (d12) at ($(d5)+(-150:0.5)$);
\coordinate (d13) at ($(d5)+(90:1)$);
\coordinate (d14) at ($(d5)+(-150:1)$);
\coordinate (d15) at ($(d14)+(-90:1)$);
\coordinate (d16) at ($(d14)+(-90:0.5)$);
\coordinate (d17) at ($(d14)+(150:1)$);
\coordinate (d18) at ($(d14)+(150:0.5)$);

\foreach \i in {0,1,...,18}
{
	\node [fill, circle, inner sep = 1pt] at (d\i) {};
}

\draw (d0) -- (d1) node [red, midway, sloped] {<} ;
\draw (d0) -- (d2) node [red, midway, sloped] {>} ;
\draw (d0) -- (d3) node [red, midway, sloped] {<} ;
\draw (d1) -- (d4) node [red, midway, sloped] {>} ;
\draw (d2) -- (d5) node [red, midway, sloped] {>} ;
\draw (d3) -- (d6) node [red, midway, sloped] {<} ;
\draw (d4) -- (d7) node [red, midway, sloped] {<} ;
\draw (d7) -- (d9) node [red, midway, sloped] {<} ;
\draw (d4) -- (d8) node [red, midway, sloped] {<} ;
\draw (d8) -- (d10) node [red, midway, sloped] {<} ;
\draw (d5) -- (d11) node [red, midway, sloped] {<} ;
\draw (d11) -- (d13) node [red, midway, sloped] {<} ;
\draw (d5) -- (d12) node [red, midway, sloped] {<} ;
\draw (d12) -- (d14) node [red, midway, sloped] {<} ;
\draw (d14) -- (d16) node [red, midway, sloped] {<} ;
\draw (d16) -- (d15) node [red, midway, sloped] {<} ;
\draw (d14) -- (d18) node [red, midway, sloped] {>} ;
\draw (d18) -- (d17) node [red, midway, sloped] {>} ;

\begin{scope}[xshift = -5cm]
\coordinate (d0) at (0,0);
\coordinate (d1) at (30:1);
\coordinate (d2) at (150:1);
\coordinate (d3) at (-90:1);
\coordinate (d4) at ($(d1)+(90:1)$);
\coordinate (d5) at ($(d1)+(-30:1)$);
\coordinate (d6) at ($(d2)+(90:1)$);
\coordinate (d7) at ($(d2)+(-150:1)$);
\coordinate (d8) at ($(d7)+(-90:1)$);
\coordinate (d9) at ($(d7)+(150:1)$);

\foreach \i in {0,1,...,9}
{
	\node [fill, circle, inner sep = 1pt] at (d\i) {};
}

\draw (d0) -- (d1) node [red, pos = 0.3, sloped] {<}  node [red, pos = 0.6, sloped] {>}  node [red, pos = 0.85, sloped] {>};
\draw (d0) -- (d2) node [red, pos = 0.3, sloped] {>}  node [red, pos = 0.6, sloped] {>}  node [red, pos = 0.85, sloped] {>} ;
\draw (d0) -- (d3)  node [red, pos = 0.3, sloped] {<}  node [red, pos = 0.6, sloped] {<}  node [red, pos = 0.85, sloped] {<} ;
\draw (d1) -- (d4) node [red, pos = 0.3, sloped] {<}  node [red, pos = 0.55, sloped] {<}  node [red, pos = 0.85, sloped] {<} ;
\draw (d1) -- (d5) node [red, pos = 0.3, sloped] {<}  node [red, pos = 0.6, sloped] {<}  node [red, pos = 0.85, sloped] {<};
\draw (d2) -- (d6)  node [red, pos = 0.3, sloped] {<}  node [red, pos = 0.6, sloped] {<}  node [red, pos = 0.85, sloped] {<};
\draw (d2) -- (d7)  node [red, pos = 0.3, sloped] {<}  node [red, pos = 0.6, sloped] {<}  node [red, pos = 0.85, sloped] {<} ;
\draw (d8) -- (d7)  node [red, pos = 0.3, sloped] {>}  node [red, pos = 0.6, sloped] {>}  node [red, pos = 0.85, sloped] {>} ;
\draw (d9) -- (d7)  node [red, pos = 0.3, sloped] {>}  node [red, pos = 0.6, sloped] {>}  node [red, pos = 0.9, sloped] {>} ;
\end{scope}
\end{tikzpicture}
\caption{An example of $T'$ from $T$ leading to at least one sink.}
\label{pic_orient_half}
\end{center}
\end{figure}
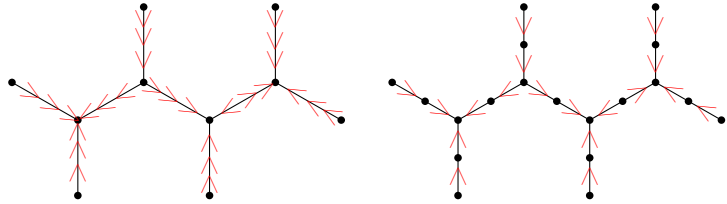

\subparagraph*{Tightness of the obstruction.} In this paragraph, we show that bubble tangles form a tight obstruction to sphere decompositions, in the sense that a bubble tangle of order $k$ exists whenever a sphere decomposition of width less than $k$ does not exist.

\begin{proposition}\label{P:nogapside}
Let $G$ be an embedding of a graph in $\Sp^3$ and $k$ be an integer at least three. If $G$ does not admit a sphere decomposition of width less than $k$, then there exists a bubble tangle of order $k$.
\end{proposition}

The idea of the proof is to show that, given a collection of closed balls satisfying the axioms~\ref{def_T1} and~\ref{def_T4} of bubble tangles, then either we can extend this collection to a bubble tangle, or there exists a \bfdex{partial sphere decomposition} of width $k$ which sweeps the space "between" the balls of the collection.
We first introduce the relevant definition.

Let $G$ a graph embedded in $\Sp^3$. A \bfdex{partial sphere decomposition} of $G$ is a continuous map $f : \Sp^3 \rightarrow T$ where $T$ is a trivalent tree with at least one edge such that: 
	\begin{itemize}
		\item For all $x \in L(T),~f^{-1}(x)$ is a point disjoint from $G$ or a closed ball $B$.
		\item For all $x \in V(T) \smallsetminus L(T),~f^{-1} (x)$ is a double bubble transverse to $G$.
		\item For all $x$ interior to an edge, $f^{-1} (x)$ is a sphere transverse or finitely tangent to $G$.
	\end{itemize}
The leaves of $T$ having preimages by $f$ which are not points are called \bfdex{non-trivial leaves}.
Let $G$ be a graph embedding in $\Sp^3$ and $\mathcal{A}$ be a collection of closed balls in $\Sp^3$. A partial sphere decomposition $f$ \bfdex{conforms} to $\mathcal{A}$ if, for all $x \in L(T),$ $f^{-1}(x)$ is either a point disjoint from $G$, or a closed ball $B$ such that there exists $A \in \mathcal{A}$ such that $\partial B$ and $\partial A$ are braid equivalent and $B \subset A$. In the latter case we say that $x$ conforms to $A$.
The width of a partial sphere decomposition is defined like the width of standard sphere decompositions: it is the supremal weight of spheres that are pre-images of points in the interiors of edges of $T$.

Now, the proof of Proposition~\ref{P:nogapside} hinges on the following key lemma. Its proof is similar to branchwidth-tangle duality proofs~\cite{Graph_Minors_X} in that it builds a bubble tangle inductively, but the continuous nature of our objects makes us rely on transfinite induction in the form of Zorn's lemma.

\begin{lem}\label{prop_extension_conformity}
Let $G$ be an embedding of a graph in $\Sp^3$, $k$ be an integer at least $3$ and $\mathcal{A}$ be a collection of closed balls in $\Sp^3$ satisfying \ref{def_T1} and \ref{def_T4}. Then one of the following is true : 
\begin{itemize}
	\item $\mathcal{A}$ extends to a bubble tangle of order $k$.
	\item there is a partial sphere decomposition of width less than $k$ that conforms to $\mathcal{A}$.
\end{itemize}
\end{lem}

The proof of Lemma~\ref{prop_extension_conformity} relies on the following preliminary lemma which allows us to slightly move the non-trivial leaves of a partial sphere decomposition.

\begin{lem}\label{lem_partial_truncate}
Let $f : \Sp^3 \rightarrow T$ be a partial sphere decomposition that conforms to a set of closed balls $\T$. Let $\ell \in L(T)$ be a non trivial leaf such that $\partial f^{-1} (\ell)$ conforms to $A \in \T$. Let $B$ be a closed ball such that $B \subset A$ and $\partial B$ is braid-equivalent to $f^{-1}(\ell)$. Then there exists a partial sphere decomposition $f' : \Sp^3 \rightarrow T$ that conforms to $\T$ such that $w(f) = w(f')$ and $f'^{-1} (\ell) = B$.
\end{lem}

\begin{proof}[Proof of Lemma~\ref{lem_partial_truncate}]
Let us denote by $I$ the subset of $\Sp^3$ that is swept by $f$, that is, the union $\bigcup f^{-1}(x)$ where $x$ ranges over all the points of $T$ except the non-trivial leaves. Now, $(I \smallsetminus B) \smallsetminus G$ is homeomorphic to $I \smallsetminus G$, since the subset of $\Sp^3 \smallsetminus G$ in $B \smallsetminus f^{-1}(\ell)$ is homeomorphic to the product of a sphere with holes with an interval, by definition of braid-equivalence. We can easily extend this homeomorphism to a homeomorphism between $I \smallsetminus B$ and $I$ as is done in the proof of Lemma~\ref{lem_braid_equi}, and then we compose $f$ with this homeomorphism to obtain a new partial sphere decomposition $f'$ so that $w(f)=w(f')$ and $f'^{-1}(\ell)=B$.
\end{proof}

\begin{proof}[Proof of Lemma~\ref{prop_extension_conformity}]
  The proof relies on Zorn's lemma. Assume that there is no partial sphere decomposition of $G$ of width less than $k$ that conforms to $\mathcal{A}$, for ease of notation we denote this assumption by (NPD). Let $X$ be the set of collection of closed balls in $\Sp^3$ satisfying \ref{def_T1}, \ref{def_T4} and (NPD) and whose boundaries are transverse to $G$. Then, we order $X$ by inclusion. Our aim is to apply Zorn's lemma on $X$ in order to get a maximal set of closed balls containing $A$ and still satisfying \ref{def_T1}, \ref{def_T4}, and (NPD).
  
  We first show that every chain of $X$ admits an upper bound in $X$. Let $C$ be a chain of $X$ and $\T = \bigcup_{c \in C} c$. If $C$ is empty, then the set of $G$-trivial balls is an upper bound by definition of $X$. Otherwise, there exists $c$ in $C$, and as $c$ satisfies \ref{def_T4} and $c \subset \T$, we also have that $\T$ satisfies \ref{def_T4}. Let $B$ be a closed ball in $\T$. Then there exists $c$ in $C$ such that $B$ belongs to $C$. Since $c$ satisfies \ref{def_T1}, we have $|C(\partial B \cap G)| < k$, and thus $\T$ satisfies \ref{def_T1} as well. Finally, we establish (NPD) for $\T$: we assume by contradiction that there exists a partial sphere decomposition $f$ of $G$ of width less than $k$ that conforms to $\T$. For each non-trivial leaf $\ell$ of $T$, by definition of conformity, there exists a ball $B_\ell \in c_\ell \in C$ such that $\partial f^{-1} (\ell)$ is braid-equivalent to $\partial B_\ell$. As $C$ is totally ordered and $L(T)$ is finite, there is an element $c$ of $C$ that contains all the sets $c_\ell$. This implies that $f$ conforms to $c$ which is absurd because $c$ satisfies (NPD). We conclude that $\T$ satisfies (NPD). Hence $\T$ belongs to $X$ and is an upper bound of $C$.

By Zorn's lemma, since every chain of $X$ admits an upper bound, it admits a maximal element: there exists $\T$ such that $\mathcal{A} \subset \T$, and $\T$ is maximal with respect to \ref{def_T1}, \ref{def_T4}, and (NPD).

Notice that $\T$ satisfies \ref{def_T3}. Indeed, if there exists $B_1,B_2,B_3 \in \T^3$ such that $B_1,B_2,B_3$ induce a double bubble $\B$ transverse to $G$ and $\Sp^3 = B_1 \cup B_2 \cup B_3$, then each of the spheres $S_1,S_2,S_3$ induced by the double are transverse to $G$. Each of them admits a braid-equivalent sphere in its neighborhood, we choose for each of them a sphere $S'_i$, braid-equivalent to $S_i$ and disjoint from $\B$ (such a sphere exists because each $S_i$ has braid-equivalent spheres on both of its sides). 

\begin{figure}[h]
\begin{center}
%\tikzsetnextfilename{trivial_partial_decomp}
\begin{tikzpicture}
%outside bubble
\fill [purple, opacity = 0.3] (-2,-1.5) rectangle (2,1.5);
\filldraw [fill = white, draw = purple] ($(-0.5,0)+(65:1.183)$) arc (65:295:1.183) arc (-115:115:1.183) -- cycle;

%double bubble
\draw [blue, very thick] (0,0.866) arc (60:300:1) arc (-120:120:1);
\draw [blue, very thick] (0,0.866) .. controls +(-50:0.3) and \control{(0,-0.866)}{(50:0.3)};
\draw [blue, dotted, very thick] (0,0.866) .. controls +(-130:0.3) and \control{(0,-0.866)}{(130:0.3)};

%left inside
\filldraw [purple, fill opacity = 0.3] ($(-0.5,0)+(65:0.75)$) arc (65:295:0.75) .. controls +(50:0.3) and \control{($(-0.5,0)+(65:0.75)$)}{(-50:0.3)};

%right inside
\filldraw [purple, fill opacity = 0.3] ($(0.5,0)+(115:0.75)$) arc (115:-115:0.75) .. controls +(130:0.3) and \control{($(0.5,0)+(115:0.75)$)}{(-130:0.3)};

%labels
\node [purple] at (-0.65,0) {$B'_1$};
\node [purple] at (0.65,0) {$B'_2$};
\node [purple] at (-1.65,1) {$B'_3$};

\begin{scope}[xshift = 4.5cm, yshift = -0.5cm]
\draw (0,0) -- (90:1.25);
\draw (0,0) -- (-30:1.25);
\draw (0,0) -- (-150:1.25);
\fill [blue] (0,0) circle (2pt);
\fill [purple] (90:1.25) circle (2pt);
\fill [purple] (-30:1.25) circle (2pt);
\fill [purple] (-150:1.25) circle (2pt);
\node at (90:1.25) [purple, above] {$B'_1$};
\node at (-30:1.25) [purple, below right] {$B'_2$};
\node at (-150:1.25) [purple, below left] {$B'_3$};
\end{scope}
\end{tikzpicture}
\caption{A partial sphere decomposition that conforms to $\B$.}
\label{pic_partial_conf}
\end{center}
\end{figure}
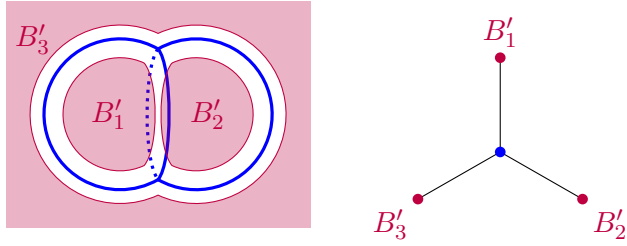

Let $B'_i$ be the closed ball disjoint from $\B$ such that $\partial B'_i = S'_i$. As $S_i$ and $S'_i$ are braid-equivalent, there is a homeomorphism $h_i : S_k \times [0,1]  \rightarrow I_i \smallsetminus G$; where $I_i = B_i \smallsetminus \mathring B'_i$. As seen in the proof of Lemma~\ref{lem_braid_equi}, we can extend $h^{-1}_i$ to a continuous function $ \phi_i : I_i \rightarrow \Sp^2 \times [0,1]$. We can now construct a partial sphere decomposition that conforms to $ B_1,B_2,B_3$ (see Figure \ref{pic_partial_conf}) on the trivalent tree made of $3$ edges $(\ens{v_0,v_1,v_2,v_3}, \ens{[v_0,v_1],[v_0,v_2],[v_0,v_3]})$.  
\begin{equation*}
f(x) =
\left\lbrace
\begin{array}{l}
v_0 \text{ if } x \in \B\\
v_i \text{ if } x \in B'_i\\
y \in [v_0,v_i] \approx [0,1] \text{ if } x \in I_i \tand h_i (x) = (\cdot,y)
\end{array}
\right.
\end{equation*}
This partial sphere in contradiction is in contradiction with $\T$ satisfying (NPD), thus $\T$ satisfies \ref{def_T3}.

Now, if $\T$ satisfies~\ref{def_T2}, it is a bubble tangle, and we are
done. Otherwise $\T$ does not satisfy \ref{def_T2}: there exists some
sphere of weight less than $k$ so that none of the two
balls that it bounds are in $\T$. In that case, we can define a partial
sphere decomposition that conforms to $\T$. Let $S$ be such a sphere
in $\Sp^3$: $|C(S \cap G)| < k$ and for $B_1,B_2$, its sides, neither
$B_1 \in \T$ nor $B_2 \in \T$. Notice that $\T \cup \ens{B_1}$ is a
collection of closed balls satisfying both \ref{def_T1} and
\ref{def_T4}. By maximality of $\T$ under \ref{def_T1}, \ref{def_T4},
and (NPD), there exists a partial sphere decomposition $f_1 : \Sp^3
\rightarrow T_1$ of width at most $k$ that conforms to $\T \cup
\ens{B_1}$. As $f_1$ does not conform to $\T$, it necessarily admits a
non-trivial leaf $\ell_1$ such that $\partial f_1^{-1} (\ell_1)$ is
braid-equivalent to $\partial B_1 = S$ and $f_1^{-1}(\ell_1) \subset B_1$. Similarly there exists a
partial sphere decomposition $f_2: \Sp^3 \rightarrow T_2$ of width at
most $k$ that conforms to $\T \cup \ens{B_2}$ with a non-trivial leaf
$\ell_2$ such that $\partial f_1^{-1} (\ell_2)$ is braid-equivalent to
$\partial B_2 = S$ and $\partial f_2^{-1}(\ell_2) \subset B_2$. Now, by Lemma~\ref{lem_partial_truncate}, $f_1$
and $f_2$ can be modified so that $f_1^{-1}(\ell_1)$ is $B_1$, and
$f_2^{-1}(\ell_2)$ is $B_2$. Now, $f_1$ and $f_2$ can be pasted
together at $\ell_1$ and $\ell_2$ to yield a single partial sphere
decomposition of width $k$ which conforms to $\T$. We reach a final contradiction
which concludes the proof.
\end{proof}

\begin{remark}
We use the axiom of choice in this proof for convenience, but it seems likely that one can just rely on the countable axiom of choice since, while there are uncountably many spheres with a low number of intersections with $G$, there are only countably many isotopy classes of those.
\end{remark}

We can now prove Proposition~\ref{P:nogapside}.

\begin{proof}[Proof of Proposition~\ref{P:nogapside}]
We denote by $\mathcal{A}$ the collection of $G$-trivial balls. By definition, $\mathcal{A}$ satisfies~\ref{def_T4}, and since $G$-trivial balls have weight at most two, it also satisfies~\ref{def_T1} for $k$ at least three. Therefore, by Lemma~\ref{prop_extension_conformity}, either $\mathcal{A}$ extends to a bubble tangle of order $k$, or there exists a partial sphere decomposition of width less than $k$ conforming to it. In the first case, we are done. In the second case, we are also done, since, given a partial sphere decomposition of width less than $k$ conforming to $G$-trivial balls, it is straightforward to sweep within the $G$-trivial balls so as to obtain a sphere decomposition of width less than $k$.
\end{proof}

\section{From compression representativity to bubble tangles}\label{S:representativity}

The goal of this section is to show Theorem~\ref{T:representativity}: when a graph $G$ is embedded on a compact, orientable, and non-zero genus surface $\Sigma$, there exists a bubble tangle naturally arising from the compression representativity of $G$ on $\Sigma$. \textit{In the following, we assume $\Sigma$ is compact, orientable, and not a sphere.}

Under these hypotheses, the idea of the proof is to show that there exists a natural choice of small side for every sphere with fewer intersections with $G$ than the compression representativity. Intuitively, such a sphere will only cut disks or ``trivial parts'' of $\Sigma$ on one of its sides, which we will designate as the small one. That is justified by the following lemma.

\begin{lem}\label{lem_sphere_compressible}
Let $\Sigma$ be a surface embedded in $\Sp^3$ and $S$ be a sphere in $\Sp^3$ that intersects $\Sigma$ transversely such that there is at least one non-contractible curve in the intersection. Then one of the non-contractible curves is compressible. 
\end{lem}

\begin{proof}
As $\Sigma$ and $S$ are transverse, the intersection of $S$ and $\Sigma$ consists of a disjoint union of simple closed curves. Each one of these curves bounds two disks on $S$. Let $\alpha$ be a curve of $S \cap \Sigma$ that is innermost in $S$, i.e. it bounds a disk $D$ in $S$ that does not contain any other curve of $S \cap \Sigma$. If $\alpha$ is non-contractible, then the disk $D$ is a compression disk for $\alpha$, and thus $\alpha$ is compressible. Otherwise, $\alpha$ bounds a disk $D_{\Sigma}$ in $\Sigma$ (see for example Epstein~\cite[Theorem~1.7]{epstein}). We deform $S$ continuously by ``pushing'' $D$ through $D_\Sigma$ while keeping $S$ embedded (see Figure~\ref{pic_torus_inessential}) until $\alpha$ disappears from $\Sigma \cap S$. 

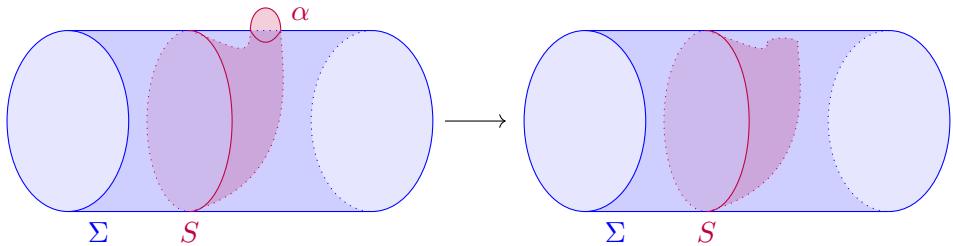
\begin{figure}[ht]
\begin{center}
%\tikzsetnextfilename{curve_removing}
\begin{tikzpicture}[scale = 0.8]
\draw [blue] (0,0) circle (1 and 1.5);
\draw [blue] (5,-1.5) arc (-90:90:1 and 1.5);
\draw [blue, dotted] (5,1.5) arc (90:270:1 and 1.5);
\draw [blue] (0,1.5) -- (3,1.5);
\draw [blue] (3.5,1.5) -- (5,1.5);
\draw [blue, dotted] (3,1.5) -- (3.5,1.5);
\draw [blue] (0,-1.5) -- (5,-1.5);
\fill [blue, opacity = 0.1] (0,-1.5) arc (-90:90:1 and 1.5) -- (5,1.5) arc (90:-90:1 and 1.5) -- cycle;
\fill [blue, opacity = 0.1] (0,-1.5) arc (-90:-270:1 and 1.5) -- (5,1.5) arc (90:270:1 and 1.5) -- cycle;
\draw [purple] (2,-1.5) arc (-90:90:0.7 and 1.5);
\draw [purple, dotted] (2,1.5) arc (90:270:0.7 and 1.5);
\draw [purple, dotted] (2,1.5) .. controls +(-20:0.8) and \control{(3,1.5)}{(-90:0.5)};
\draw [purple, dotted] (2,-1.5) .. controls +(15:2) and \control{(3.5,1.5)}{(-90:0.5)};
\draw [purple] (3,1.5) .. controls +(90:0.5) and \control{(3.5,1.5)}{(90:0.5)};
\draw [purple] (3,1.5) .. controls +(-60:0.3) and \control{(3.5,1.5)}{(-120:0.3)};
\fill [purple, opacity = 0.1] (2,1.5) .. controls +(-20:0.8) and \control{(3,1.5)}{(-90:0.5)} -- (3,1.5) .. controls +(90:0.5) and \control{(3.5,1.5)}{(90:0.5)} -- (3.5,1.5) .. controls +(-90:0.5) and \control{(2,-1.5)}{(15:2)} arc (-90:-270: 0.7 and 1.5);
\fill [purple, opacity = 0.1] (2,1.5) .. controls +(-20:0.8) and \control{(3,1.5)}{(-90:0.5)} -- (3,1.5) .. controls +(90:0.5) and \control{(3.5,1.5)}{(90:0.5)} -- (3.5,1.5) .. controls +(-90:0.5) and \control{(2,-1.5)}{(15:2)} arc (-90:90: 0.7 and 1.5);
\node at (3.5,1.5) [above right, purple] {$\alpha$};
\node at (0.5,-1.5) [below, blue] {$\Sigma$};
\node at (2,-1.5) [below, purple] {$S$};
\draw [->] (6.2,0) -- (7.2, 0);
\begin{scope}[xshift = 8.5 cm]
\draw [blue] (0,0) circle (1 and 1.5);
\draw [blue] (5,-1.5) arc (-90:90:1 and 1.5);
\draw [blue, dotted] (5,1.5) arc (90:270:1 and 1.5);
\draw [blue] (0,1.5) -- (5,1.5);
\draw [blue] (0,-1.5) -- (5,-1.5);
\fill [blue, opacity = 0.1] (0,-1.5) arc (-90:90:1 and 1.5) -- (5,1.5) arc (90:-90:1 and 1.5) -- cycle;
\fill [blue, opacity = 0.1] (0,-1.5) arc (-90:-270:1 and 1.5) -- (5,1.5) arc (90:270:1 and 1.5) -- cycle;
\draw [purple] (2,-1.5) arc (-90:90:0.7 and 1.5);
\draw [purple, dotted] (2,1.5) arc (90:270:0.7 and 1.5);
\draw [purple, dotted] (2,1.5) .. controls +(-20:0.8) and \control{(3,1.3)}{(-90:0.2)};
\draw [purple, dotted] (2,-1.5) .. controls +(15:2) and \control{(3.5,1.3)}{(-90:0.5)};
\draw [purple, dotted] (3,1.3) .. controls +(90:0.1) and \control{(3.5,1.3)}{(90:0.1)};
\fill [purple, opacity = 0.1] (2,1.5) .. controls +(-20:0.8) and \control{(3,1.3)}{(-90:0.2)} -- (3,1.3) .. controls +(90:0.1) and \control{(3.5,1.3)}{(90:0.1)} -- (3.5,1.3) .. controls +(-90:0.5) and \control{(2,-1.5)}{(15:2)} arc (-90:-270: 0.7 and 1.5);
\fill [purple, opacity = 0.1] (2,1.5) .. controls +(-20:0.8) and \control{(3,1.3)}{(-90:0.2)} -- (3,1.3) .. controls +(90:0.1) and \control{(3.5,1.3)}{(90:0.1)} -- (3.5,1.3) .. controls +(-90:0.5) and \control{(2,-1.5)}{(15:2)} arc (-90:90: 0.7 and 1.5);
\node at (0.5,-1.5) [below, blue] {$\Sigma$};
\node at (2,-1.5) [below, purple] {$S$};
\end{scope}
\end{tikzpicture}
\end{center}
\caption{\label{pic_torus_inessential} Removing a trivial curve from $S \cap \Sigma$.}
\end{figure} 

Repeating this process on a new innermost curve of $S$ will eventually yield a non-contractible compressible curve. Indeed, the number of curves in the intersection is finite (recall that both surfaces are piecewise linear), decreases at each step, and one of the curves in $\Sigma \cap S$ is non-contractible.
\end{proof}

A direct consequence of this lemma is that if $G$ is embedded on a surface $\Sigma$, a sphere $S$ intersects $\Sigma$, and the intersection has weight less than $\comprep(G,\Sigma)$, then all the simple closed curves in the intersection are contractible. Therefore, one of the two balls bounded by $S$ contains the meaningful topology of $\Sigma$, while the other one only contains spheres with holes (see Figure~\ref{F:torusknot}). In order to formalize this, we will rely on fundamental groups (see for example Hatcher~\cite{Hatcher_Algebraic_Topology} for an introduction to this concept). 
The inclusion of a subsurface $X$ on $\Sigma$ induces a morphism $i_*: \pi_1(X) \rightarrow \pi_1(\Sigma)$. If this morphism is trivial, we say that $X$ is \bfdex{$\pi_1$-trivial} with respect to $\Sigma$.

\begin{defn}[Compression bubble tangle on an embedded surface]
Let $G$ be a graph embedded on $\Sigma$, a surface embedded in $\Sp^3$ such that $\comprep(G,\Sigma) \geq 3$ and set $k = \frac{2}{3} \comprep(G,\Sigma) $. The \bfdex{compression bubble tangle} $\cT$, is the collection of balls in $\Sp^3$ defined as follows: for any sphere $S$ in $\Sp^3$ transverse to $G$ such that $ |C(S \cap G)| < k$, by Lemma~\ref{lem_sphere_compressible}, there is exactly one connected component $A$ of $\Sigma \smallsetminus S$ that is $\pi_1$-trivial. Exactly one of the open balls $B$ of $\Sp^3 \smallsetminus S$ contains $A$, and we put the closed ball in $\cT$: $\bar B \in \cT$.
\end{defn}

\begin{figure}[h]
\begin{center}
%\tikzsetnextfilename{torus_rep_sphere}
\begin{tikzpicture}[scale = 0.85]
%Torus filling
\fill [blue, opacity = 0.2] (0,0) circle (4 and 2);
\fill [white](1,0) .. controls +(150:0.5) and \control{(-1,0)}{(30:0.5)} .. controls +(-30:0.5) and \control{(1,0)}{(210:0.5)}; 

%Torus drawing and coordinate definition
\path [name path = elip, blue] (4,0) arc (0:360:4 and 2) node (c5) [pos = 0,circle, fill, inner sep = 0pt] {} node (c4) [pos = 0.166,circle, fill, inner sep = 0pt] {} node (c3) [pos = 0.333,circle, fill, inner sep = 0pt] {} node (c2) [pos = 0.5,circle, fill, inner sep = 0pt] {} node (c1) [pos = 0.666,circle, fill, inner sep = 0pt] {} node (c6) [pos = 0.833,circle, fill, inner sep = 0pt] {};
\draw [name path = highc, blue] (1,0) coordinate (a) .. controls +(150:0.5) and \control{(-1,0)}{(30:0.5)} node (i1) [pos = 0.2,circle, fill, inner sep = 0pt] {} node (i6) [pos = 0.5,circle, fill, inner sep = 0pt] {} node (i5) [pos = 0.8,circle, fill, inner sep = 0pt] {} .. controls +(-30:0.5) and \control{(a)}{(210:0.5)}  node (i4) [pos = 0.2,circle, fill, inner sep = 0pt] {} node (i3) [pos = 0.5,circle, fill, inner sep = 0pt] {} node (i2) [pos = 0.8,circle, fill, inner sep = 0pt] {};
\coordinate (aa) at ($(1,0) + (30:0.2)$);
\coordinate (bb) at ($(-1,0) + (150:0.2)$);
\draw [blue] (1,0) -- (aa);
\draw [blue] (-1,0) -- (bb);

%Knot drawing
\begin{scope}[thick]
\draw [name path = br1] (c1) .. controls +(-15:1.5) and \control{(i1)}{(10:2.5)};
\draw [name path = br6] (c6) .. controls +(15:1.5) and \control{(i6)}{(35:2.3)};
\draw [name path = br5] (c5) .. controls +(90:1) and \control{(i5)}{(50:2)};
\draw [name path = br4] (c4) .. controls +(165:1.5) and \control{(i4)}{(190:2.5)};
\draw [name path = br3] (c3) .. controls +(195:1.5) and \control{(i3)}{(215:2.3)};
\draw [name path = br2] (c2) .. controls +(-90:1) and \control{(i2)}{(230:2)};
\end{scope}
\begin{scope}[thick, opacity = 0.3]
\draw [name path = bbr1] (c1) .. controls +(165:1.5) and \control{(i6)}{(155:2)};
\draw [name path = bbr6] (c6) .. controls +(195:1.5) and \control{(i5)}{(165:2)};
\draw [name path = bbr5] (c5) .. controls +(270:1) and \control{(i4)}{(-40:1)};
\draw [name path = bbr4] (c4) .. controls +(-15:1.5) and \control{(i3)}{(-25:2)};
\draw [name path = bbr3] (c3) .. controls +(15:1.5) and \control{(i2)}{(-15:2)};
\draw [name path = bbr2] (c2) .. controls +(90:1) and \control{(i1)}{(140:1)};
\end{scope}

%intersetion
\filldraw [fill opacity = 0.2, purple, name path = sphere] (3.6,-1.1) circle (2.2);
\path [name intersections={of=elip and sphere,total=\tot}]
\foreach \s in {1,...,\tot}{coordinate (c-\s) at (intersection-\s)};
\begin{scope}
\clip (c-1)+(2,0) rectangle (c-2); 
\draw [blue,dotted] (4,0) arc (0:360:4 and 2);
\end{scope}
\begin{scope}
\clip (-5,3) -- ($(c-1) +(0,2)$) -- (c-1) -- (c-2) -- ($(c-2) +(0,-1)$) -- (-5,-3) --cycle; 
\draw [blue] (4,0) arc (0:360:4 and 2);
\end{scope}
\draw [purple, dash pattern = on 5pt off 1pt, very thick, rotate = 45] (1.75,-2.1) circle (1.3 and 0.8);

\draw [->] (5.1,0.8) -- (6.3,0.8);
\fill [opacity = 0.2, blue] (7.8,0.8) circle (1.2);
\draw [purple, dash pattern = on 5pt off 1pt, very thick] (7.8,0.8) circle (1.2);
\path (9,0.8) arc (0:360:1.2) node (p1) [pos = 0.24, fill, circle, inner sep = 0pt] {} node (p2) [pos = 0.4, fill, circle, inner sep = 0pt] {} node (p3) [pos = 0.55, fill, circle, inner sep = 0pt] {} node (p4) [pos = 0.85, fill, circle, inner sep = 0pt] {};
\draw [thick] (p1) .. controls +(230:0.4) and \control{(p2)}{(10:0.4)};
\draw [thick] (p3) .. controls +(20:0.7) and \control{(p4)}{(140:0.7)};

%Added information
\node at (7.8,0.9) [color = blue] {$\Sigma \smallsetminus A$};
\node at (-3.5,0) [color = blue] {$A$};
\node at (4,-2) [color = purple] {$B$};
\end{tikzpicture}
\end{center}
\caption{Intersection between a torus knot $T_{6,5}$ embedded on a torus and a sphere. Here the ball $B$ containing the disk on the right is in the compression bubble tangle.}
\label{F:torusknot}
\end{figure}
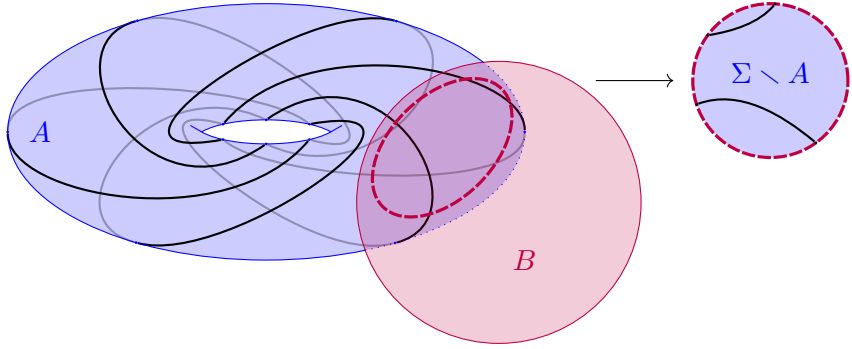  

The main step in the proof of Theorem~\ref{T:representativity} is to prove that a compression bubble tangle on the torus is indeed a bubble tangle. 

\begin{proposition}\label{prop_bubble_tangle}
A compression bubble tangle is a bubble tangle.
\end{proposition}

Note that Proposition~\ref{prop_bubble_tangle} immediately implies Theorem~\ref{T:representativity} by Proposition~\ref{prop_max_tangle_l_sw} (the theorem is trivial if $\comprep(G,\Sigma) <3$). Therefore, the remainder of this section is devoted to proving Proposition~\ref{prop_bubble_tangle}.

By definition, a compression bubble tangle satisfies \ref{def_T1} and \ref{def_T2}. We then notice that \ref{def_T4} is verified whenever the compression representativity of $G$ on $\Sigma$ is greater than $2$.

\begin{lem}\label{lem_T4}
If $\comprep(G,\Sigma) \geq 3$ then for all $G$-trivial balls $B$, $B \cap \Sigma$ is $\pi_1$-trivial.
\end{lem}

\begin{proof}[Proof of Lemma~\ref{lem_T4}]
Let $B$ be a $G$-trivial ball. Then its boundary $S$ has weight at most two. By Lemma~\ref{lem_sphere_compressible} and the definition of compression-representativity, this implies that $S$ intersects $\Sigma$ only on contractible curves. Therefore, $S$ bounds two balls, exactly one of which, denoted by $B'$ is $\pi_1$-trivial. This implies that $B'$ is $G$-trivial. Now either $B=B'$ and we are done, or they are different, but this would directly imply that $G$ is made of two trivial pieces, and is thus an unknot, plus possibly some trees attached to it. Since an unknot has compression-representativity one (see for example~\cite[Example~3.2]{ozawa}), this would contradict the assumption that the compression-representativity is at least three.
\end{proof}

The hard part of the proof is to show that \ref{def_T3} is satisfied. This is more delicate than it seems at first glance, since any surface can be obtained by gluing three disks, and these three disks can even come from a double bubble: we provide an example in Appendix~\ref{A:torus} and Figures~\ref{pic_torus_3_disks}, \ref{pic_double_bubble_example_2}, and \ref{pic_double_bubble_example_3} in the case of the torus.

Henceforth, we will proceed by contradiction and assume that we can cover $\Sp^3$ by three closed balls $B_1,B_2,B_3$ of $\cT$ that induce a double bubble $DB$ transverse to $\Sigma$ and $G$. Thus $\Sigma$ is covered by three surfaces with boundary: $\Sigma \cap B_1, \Sigma \cap B_2$ and $\Sigma \cap B_3$ which are $\pi_1$-trivial by definition of $\cT$. In the following, we write $S_i = \partial B_i$. We first show that we can furthermore assume that these surfaces are a disjoint union of closed disks on $\Sigma$.

\begin{lem}\label{lem_Sigma_disks}
  Let $G$ be a graph embedded on $\Sigma$, a surface embedded in $\Sp^3$. Let $\cT$ be the compression bubble tangle associated to $G$ and $\Sigma$. If there is a double bubble $DB$ transverse to $\Sigma$, inducing three balls $B_1,B_2,B_3 \in \cT^3$ such that $B_1 \cup B_2 \cup B_3 = \Sp^3$, then we can isotope the double bubble so that we additionally have that $B_i \cap \Sigma$ is a union of closed disks.
\end{lem}

\begin{proof}[Proof of Lemma~\ref{lem_Sigma_disks}]
 By Proposition~\ref{lem_sphere_compressible}, each simple closed curve $c$ of $\bigcup{C(\partial B_i \cap \Sigma)}$ is a contractible curve of $\Sigma$. Hence it bounds a unique closed disk $D_c$ on $\Sigma$ (the other connected component being a surface with non-zero genus and a puncture), and we can define a \bfdex{slope}\footnote{This terminology mirrors the one of Robertson and Seymour in~\cite{Graph_minors_XI}.} $s$ that associates the disk $D_c$ to each $c$:
\begin{equation*}
\begin{array}{rcl}
s:~\oper{\bigcup}{i \in \ens{1,2,3}}{}{C(\partial B_i \cap \Sigma)} & \rightarrow & \P(\Sigma)\\
 c & \mapsto & \text{the closed disk }D_c \text{ of $\Sigma$ such that } \partial D_c = c
\end{array}
\end{equation*}

Let $c$ be a simple closed curve of $S_i \cap \Sigma$ such that $s(c) \not \subset B_i$ and that is innermost in $S_i$ with respect to that property, i.e., $c$ bounds a disk $D_c$ in $S_i$ such that all $c' \in C(\mathring D_c \cap \Sigma)$ satisfy $s(c') \subset B_i$ (see Figure~\ref{pic_surface_disks}). Denote by $D_{c'}$ the disk associated to each $c'$ on $S_i$.

\begin{figure}[H]
\begin{center}
%\tikzsetnextfilename{pic_torus_disks}
\begin{tikzpicture}[scale = 0.75]
\fill [fill = purple!60!] (0,0) circle (2.4);  
\fill [fill = green!60!purple] (0,0) circle (1.25);  

\fill [fill = green!60!purple, draw = purple] (0,0) circle (1.25);  
\fill [blue!60!] (0,1.25) -- (0,-1.25) arc  (-90:90:1.25);
\draw (0,1.25) -- (0,-1.25);
\draw [color = purple] (0,0) circle (1.25);
\node at (45:1.21) [right, color = purple!70!black] {c};
\draw [dashed] (0,0) circle (1.7);
\node at (0:2) {$\alpha$};
\filldraw [fill = purple!60!, draw = purple] (0,0) circle (0.5);
\node at (-45:0.4) [right, color = purple!70!black] {c'}; 

\node at (3,0) {$\rightarrow$};

\fill [fill = purple!60!] (6,0) circle (2.4); 

\begin{scope}[xshift = -7cm]
\fill [fill = purple!60!] (0,0) circle (2.4);    
\fill [fill = green!60!purple] (0,0) circle (1.25);  
\node at (45:1.25) [right] {c};
\fill [fill = green!60!purple, draw = purple] (0,0) circle (1.25);  
\node at (45:1.25) [right] {c};
\fill [blue!60!] (0,1.25) -- (0,-1.25) arc  (-90:90:1.25);
\draw (0,1.25) -- (0,-1.25);
\draw [color = black] (0,0) circle (1.25);
\filldraw [fill = purple!60!, draw = purple] (0,0) circle (0.5);
\fill [color = gray, opacity = 0.85] (0,0) circle (1.25);
\node at (0,-0.4) {$s(c)$};
\node [color = purple!60!black] at (-90:1.8) {$B_i \cap \Sigma$};
\end{scope}
\end{tikzpicture}
\caption{An example of $c$, innermost with the property that $s(c) \not \subset B_i$. The green and blue parts depict the intersection between $\Sigma$ and the two other balls induced by the double bubble.} 
\label{pic_surface_disks}
\end{center}
\end{figure}
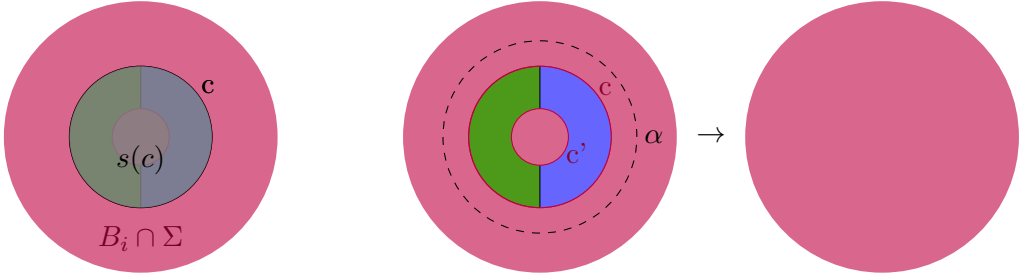

The idea is now to push $s(c)$, a disk of $\Sigma$, through $D_c$ until ${s(c)} \subset B_i$ (see Figure~\ref{pic_surface_disks_2}). In order to remove both the interior of $D_c$ and its boundary, we define $\alpha$, a simple closed curve homotopic to $c$ in $\Sigma \cap B_i$ lying in a tubular neighborhood of $c$ and disjoint from $S_i$. Then $D_\alpha = s(\alpha)$ satisfies $D_c \subset D_\alpha$. 

To be more precise, let $\Delta$ be the disk of $S_i \smallsetminus c$ not disjoint from $s(c)$. As $c$ is innermost on $S_i$ for the property that $s(c) \not \subset B_i$, $\Delta \smallsetminus \Sigma$ has one connected component $W$ which contains $c$ in its boundary, and possibly other connected components which are open disks (any other surface would contradict the fact that $c$ is innermost). These disks may cobound with $D_\alpha$ some closed balls of $B_i$. But in any case, it suffices for $D_\alpha$ to cross through $W$ during the push (see Figure~\ref{pic_surface_disks_2}). As $c$ is innermost on $S_i$ with respect to the property above, $\Sp^3 \smallsetminus (W \cup D_\alpha)$ has two connected components. Let $A$ be the one disjoint from $B_i$, it is necessary disjoint from $\Sigma$ (again because $c$ is innermost). Hence the push described here makes $D_\alpha$ sweep $A$ until it crosses $W$ to end in $B_i$.

This transformation can be made a piecewise-linear isotopy since $A \cup \mathring W$ is disjoint from $\Sigma$ and both $W$ and $\Sigma$ are piecewise-linear. We apply this transformation on $\Sigma$ to obtain a new embedding $\Sigma'$. The number of elements in $C(S_i \cap \Sigma')$ is smaller than $C(S_i \cap \Sigma)$ so that repeating this process will eventually end.

\begin{figure}[H]
\begin{center}
%\tikzsetnextfilename{pic_torus_disks_2}
\begin{tikzpicture}[scale = 0.75]
\draw [blue] (-1,0) -- (7,0);
\node [blue, left] at (-1,0) {$\Sigma$};
\draw [ultra thick, dashed] (0,0) .. controls +(75:1) and \control{(4,0)}{(105:1)};
\draw [ultra thick, dashed] (5,0) .. controls +(75:0.75) and \control{(6,0)}{(105:0.75)};
\draw [color = purple] (-0.8,-1.5) ..controls +(3.5,0) and \control{(2,0)}{(-75:1)} .. controls +(105:0.6) and \control{(1,0)}{(75:0.6)} .. controls +(-105:1) and \control{(0,0)}{(-105:1)} .. controls +(75:1) and \control{(4,0)}{(105:1)} .. controls +(-75:0.4) and \control{(5,0)}{(-105:0.4)} .. controls +(75:0.75) and \control{(6,0)}{(105:0.75)} .. controls +(-75:1) and \control{(6.8,-1.5)}{(-1,0)}; 
\fill [color = purple, opacity = 0.3] (-0.8,-1.5) ..controls +(3.5,0) and \control{(2,0)}{(-75:1)} .. controls +(105:0.6) and \control{(1,0)}{(75:0.6)} .. controls +(-105:1) and \control{(0,0)}{(-105:1)} .. controls +(75:1) and \control{(4,0)}{(105:1)} .. controls +(-75:0.4) and \control{(5,0)}{(-105:0.4)} .. controls +(75:0.75) and \control{(6,0)}{(105:0.75)} .. controls +(-75:1) and \control{(6.8,-1.5)}{(-1,0)} -- (6.8,1.3) -- (-0.8,1.3) -- cycle; 
\node at (0,-1.1) [color = purple!80!black] {$B_i$};
\node at (7.5,0) {$\rightarrow$};

\node at (4.5,0.5) [black] {$W$};

\begin{scope}[xshift = 9.5cm]
\draw [blue](-1,0) -- (-0.6,0) .. controls + (45:2) and \control{(4.3,0)}{(135:2)} -- (4.7,0) .. controls +(75:1) and \control{(6.3,0)}{(105:1)} -- (7,0);
\node [blue, left] at (-1,0) {$\Sigma$};
\draw [color = purple] (-0.8,-1.5) ..controls +(3.5,0) and \control{(2,0)}{(-75:1)} .. controls +(105:0.6) and \control{(1,0)}{(75:0.6)} .. controls +(-105:1) and \control{(0,0)}{(-105:1)} .. controls +(75:1) and \control{(4,0)}{(105:1)} .. controls +(-75:0.4) and \control{(5,0)}{(-105:0.4)} .. controls +(75:0.75) and \control{(6,0)}{(105:0.75)} .. controls +(-75:1) and \control{(6.8,-1.5)}{(-1,0)}; 
\fill [color = purple, opacity = 0.3] (-0.8,-1.5) ..controls +(3.5,0) and \control{(2,0)}{(-75:1)} .. controls +(105:0.6) and \control{(1,0)}{(75:0.6)} .. controls +(-105:1) and \control{(0,0)}{(-105:1)} .. controls +(75:1) and \control{(4,0)}{(105:1)} .. controls +(-75:0.4) and \control{(5,0)}{(-105:0.4)} .. controls +(75:0.75) and \control{(6,0)}{(105:0.75)} .. controls +(-75:1) and \control{(6.8,-1.5)}{(-1,0)} -- (6.8,1.3) -- (-0.8,1.3) -- cycle; 
\node at (0,-1.1) [color = purple!80!black] {$B_i$};
\end{scope}
\end{tikzpicture}
\caption{Pushing the surface $\Sigma$ until $c$ is removed from it intersection with $B_i$.}
\label{pic_surface_disks_2}
\end{center}
\end{figure}
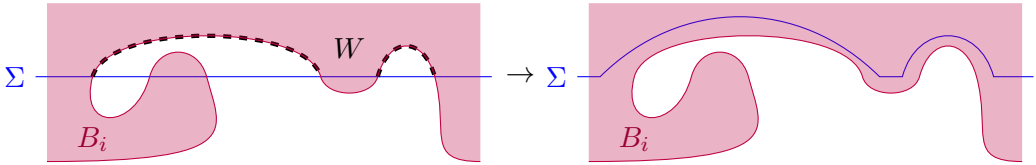

During this transformation, for all $j \in \ens{1,2,3}$ the only modifications of $\Sigma \cap B_j$ were made on $s(c)$ such that $DB \cap \Sigma' = DB \cap \Sigma \smallsetminus s(c)$ and the transversality is preserved. 
Furthermore, $|S_i \cap G|$ did not increase so that the balls $B_j$ still verify \ref{def_T1} and are still in $\cT$ since they are either not intersecting $G$ (hence $G$-trivial) or they contain the $\pi_1$-trivial subparts of $\Sigma'$.

We repeat this transformation on $\Sigma$ until there remains no simple closed curve of $S_i \cap \Sigma$ such that $s(c) \not \subset B_i \cap \Sigma$ for all $i \in \ens{1,2,3}$. We obtain that $B_i \cap \Sigma$ is a disjoint union of closed disks. Finally, we have described the isotopy on $\Sigma$ for convenience, but up to applying a homeomorphism, we can instead keep $\Sigma$ fixed and apply an isotopy on $DB$ and obtain the same properties. This concludes the proof.
\end{proof}

 Then we define $\Gamma$ \bfdex{induced by the double bubble $DB$} to be the intersection of the double bubble with $\Sigma$: where vertices are the intersection of the common boundary of the three disks with $\Sigma$ and edges are the intersections of $\Sigma$ with the disks. By Lemma~\ref{lem_Sigma_disks}, we can assume that this graph is trivalent and cellularly embedded. It is naturally weighted by endowing each edge with its weight, i.e., the number of connected components in its intersection with $G$. Let us now state the lemma we will use for the sake of contradiction.

\begin{lem}\label{lem_contrad_compress}
The total weight of $\Gamma$ is less than $\comprep(G, \Sigma)$: 
\begin{equation*}
\sum_{e \in E(\Gamma)} |C(e \cap G)| < \comprep(G, \Sigma).
\end{equation*}
\end{lem}

\begin{proof}
Since each edge of $\Gamma$ bounds exactly two faces of $\Gamma$, i.e, disks of $\Sigma$; and $\Gamma = DB \cap \Sigma$ we get the following equality:

\begin{equation}\label{eq_prop_contrad_torus_disjoint_1}
|C(S_1 \cap G)| + |C(S_3 \cap G)| + |C(S_3 \cap G)| = 2\oper{\sum}{e \in E(\Gamma)}{}{|C(e \cap G)|}
\end{equation}

By definition of $\cT$, each ball $B_i$ satisfies $|C(S_i \cap G)| < \frac{2}{3} \comprep (G,\Sigma)$ so that:

\begin{equation}\label{eq_prop_contrad_torus_disjoint_2}
|C(S_1 \cap G)| + |C(S_3 \cap G)| + |C(S_3 \cap G)| < 3 \cdot \frac{2}{3} \comprep (G,\Sigma) = 2 \comprep (G,\Sigma).
\end{equation}

Combining (\ref{eq_prop_contrad_torus_disjoint_1}) and (\ref{eq_prop_contrad_torus_disjoint_2}) concludes the proof: $2\oper{\sum}{e \in E(\Gamma)}{}{|C(e \cap G)|} < 2 \comprep (G,\Sigma)$.
\end{proof}

Hence, if $\Gamma$ contained a simple closed curve that is compressible, we would obtain the contradiction that we are looking for. The rest of the proof almost consists of finding such a compressible curve, leading to the following proposition.

\begin{proposition}\label{prop_1cur_compress}
There exists a set of edges $X$ on $\Gamma$ such that:
\begin{equation*}
\sum_{e \in X} |C(e \cap G)| \geq \comprep(G,\Sigma).
\end{equation*}
\end{proposition}

The proof of Proposition~\ref{prop_1cur_compress} is the technical crux of Theorem~\ref{T:representativity}. It consists in defining a merging process, which gradually merges two balls of a double bubble, and proving that at some point in this merging process, one ball will intersect $\Sigma$ in a non-trivial way, and thus yield a compressible curve via Lemma~\ref{lem_sphere_compressible}. An additional difficulty is that this curve might be non-simple in $\Gamma$;  we circumvent this issue by finding a fractional version of such a curve instead, which will be strong enough to prove Proposition~\ref{prop_1cur_compress}.

\begin{proof}[Proof of Proposition~\ref{prop_1cur_compress}] In order to prove Proposition~\ref{prop_1cur_compress}, we consider three balls $B_1,B_2,B_3$ inducing a double bubble $DB$ on their boundaries, and define a merging process that gradually merges two of these balls in a controlled way. This is illustrated in Figure~\ref{pic_merging_process} with a double bubble intersecting a torus.\footnote{In this figure, the double bubble does not induce a cellularly embedded graph for clarity purposes, since even on a torus the resulting picture would be too intricate to describe the merging process (compare with pictures in Appendix~\ref{A:torus}).} The point of this merging process is to yield a family of balls, one of which will have a non-trivial intersection with $\Sigma$ and thus allow us to (almost) find a compressible curve.

The merging process depends on the shape of one specific disk $D$ of the double bubble, say the one between $B_1$ and $B_2$, which we will call a \bfdex{membrane} to avoid confusion with the many disks that we deal with. Each connected component of $D \cap \Sigma$ is a separating curve of $D$ with both ends on $\partial D$ since intersections between $D$ and $\Sigma$ are edges of $\Gamma$. Considering these arcs as the embeddings of edges of a graph on $D$, we consider the dual of this graph, which we call the \bfdex{membrane tree} and denote $MT(D)$. Its set of vertices are connected components of $D \smallsetminus \Sigma$, $V(MT(D)) = C(D \smallsetminus \Sigma)$, and there is an edge between $f,f' \in V(MT(D))^2$ if $\bar f  \cap \bar{f'} \neq \varnothing$. In the following we will consider $E(MT(D)) = C(D \cap MT(D))$ as there is a natural one-to-one correspondence between the sets (see Figure~\ref{pic_membrane_tree}). Note that $MT(D)$ is indeed a tree since it is the (weak) dual of an outerplanar graph. 

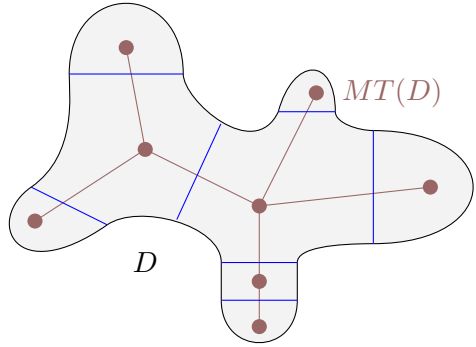
\begin{figure}[H]
\begin{center}
%\tikzsetnextfilename{membrane_tree_example}
\begin{tikzpicture}[scale = 1]
%disk
\coordinate (d1) at (1.5,1);
\coordinate (d2) at (0,1);
\coordinate (d3) at (-0.5,-0.5);
\coordinate (d4) at (0.5,-1);
\coordinate (d5) at (2,-1.5);
\coordinate (d6) at (3,-1.5);
\coordinate (d7) at (2,-2);
\coordinate (d8) at (3,-2);
\coordinate (d9) at (4,-1.25);
\coordinate (d10) at (4,0.25);
\coordinate (d11) at (3.5,0.5);
\coordinate (d12) at (2.75,0.5);
\draw (d1) .. controls +(90:1.25) and \control{(d2)}{(90:1.25)};
\draw (d2) .. controls +(-90:0.25) and \control{(d3)}{(35:0.75)};
\draw (d3) .. controls +(-145:0.75) and \control{(d4)}{(-145:1.5)};
\draw (d4) .. controls +(35:0.5) and \control{(d5)}{(90:0.5)} node [midway, inner sep = 0] (d13) {};
\draw (d5) -- (d7);
\draw (d7) .. controls +(-90:0.75) and \control{(d8)}{(-90:0.75)};
\draw (d8) -- (d6);
\draw (d6) .. controls +(90:0.2) and \control{(d9)}{(180:0.5)};
\draw (d9) .. controls +(0:1.75) and \control{(d10)}{(0:1.75)};
\draw (d10) .. controls +(180:0.20) and \control{(d11)}{(-90:0.15)};
\draw (d11) .. controls +(90:0.75) and \control{(d12)}{(75:0.75)};
\draw (d12) .. controls +(-115:0.75) and \control{(d1)}{(-90:0.35)} node [midway, inner sep = 0] (d14) {};
\fill [opacity = 0.05] (d1) .. controls +(90:1.25) and \control{(d2)}{(90:1.25)}  .. controls +(-90:0.25) and \control{(d3)}{(35:0.75)} .. controls +(-145:0.75) and \control{(d4)}{(-145:1.5)} .. controls +(35:0.5) and \control{(d5)}{(90:0.5)} -- (d7) .. controls +(-90:0.75) and \control{(d8)}{(-90:0.75)} -- (d6) .. controls +(90:0.2) and \control{(d9)}{(180:0.5)} .. controls +(0:1.75) and \control{(d10)}{(0:1.75)}.. controls +(180:0.20) and \control{(d11)}{(-90:0.15)} .. controls +(90:0.75) and \control{(d12)}{(75:0.75)} .. controls +(-115:0.75) and \control{(d1)}{(-90:0.35)};

%vertices
\coordinate (v1) at (0.75,1.35);
\coordinate (v2) at (1,0);
\coordinate (v3) at (-0.45,-0.95);
\coordinate (v4) at (2.5,-0.75);
\coordinate (v5) at (2.5,-1.75);
\coordinate (v6) at (2.5,-2.35);
\coordinate (v7) at (4.75,-0.5);
\coordinate (v8) at (3.25,0.75);
\foreach \i in {1,2,...,8}
{
	\node at (v\i) [fill, circle, inner sep = 2pt, brown!80!blue] {};
}
%tree
\begin{scope}[brown!80!blue]
\draw (v1) -- (v2) -- (v3);
\draw (v2) -- (v4) -- (v5) -- (v6);
\draw (v8) -- (v4) -- (v7);
\end{scope}

%Intersection disk torus
\begin{scope}[blue]
\draw (d1) -- (d2);
\draw (d3) -- (d4);
\draw (d5) -- (d6);
\draw (d7) -- (d8);
\draw (d9) -- (d10);
\draw (d11) -- (d12);
\draw (d13) -- (d14);
\end{scope}

%labels
\node at (1, -1.5) {$D$};
\node at (4.25, 0.75) [brown!80!blue] {$MT(D)$};
\end{tikzpicture}
\end{center}
\caption{The membrane tree (brown) of a membrane $D$ whose edges of $\Gamma$ stem from $D \cap \Sigma$ (blue).}
\label{pic_membrane_tree}
\end{figure}

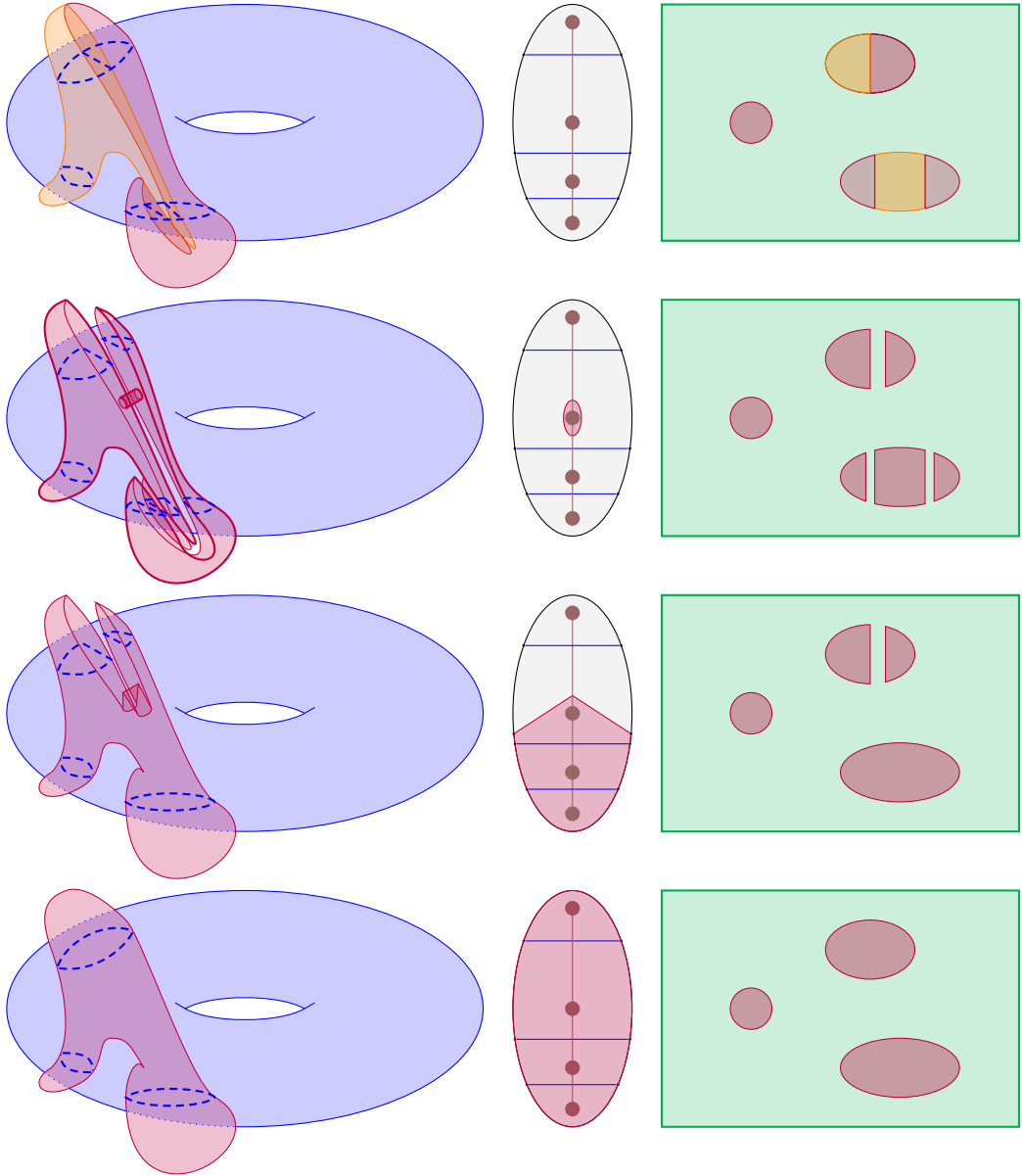
\begin{figure}[h!]
\begin{center}
%\tikzsetnextfilename{merging_process}
\begin{tikzpicture}[scale =0.8]
\begin{scope}[xshift = 4cm, yshift = 2.5cm]
%Torus filling
\fill [blue, opacity = 0.2] (0,0) circle (4 and 2);
\fill [white](1,0) .. controls +(150:0.5) and \control{(-1,0)}{(30:0.5)} .. controls +(-30:0.5) and \control{(1,0)}{(210:0.5)}; 

%Torus drawing and coordinate definition
\path [name path = elip, blue] (4,0) arc (0:360:4 and 2);
\draw [name path = highc, blue] (1,0) .. controls +(150:0.5) and \control{(-1,0)}{(30:0.5)} .. controls +(-30:0.5) and \control{(1,0)}{(210:0.5)};
\coordinate (aa) at ($(1,0) + (30:0.2)$);
\coordinate (bb) at ($(-1,0) + (150:0.2)$);
\draw [blue] (1,0) -- (aa);
\draw [blue] (-1,0) -- (bb);

%sphere 
\coordinate (g0) at (-3,2);
\coordinate (g1) at (-3.25,1);
\coordinate (g2) at (-3,0);
\coordinate (g3) at (-3.25,-1);
\coordinate (g4) at (-2.75,-1.25);
\coordinate (g5) at (-2.25,-0.5);
\coordinate (g6) at (-1.7,-1);
\coordinate (g7) at (-1,-1.75);
\coordinate (r1) at (-2,1.5);
\coordinate (r2) at (-1.25,-0.5);
\coordinate (r3) at (-0.5,-1.5);
\coordinate (r4) at (-2,-1.5);
\coordinate (d1) at (-1.2,-1.75);

\path (r1) .. controls +(-45:0.4) and  \control{(r2)}{(110:0.5)} node (n2) [pos = 0.15, inner sep = 0, circle] {};
\path (g1) .. controls +(-70:0.4) and  \control{(g3)}{(45:0.65)} node (n1) [pos = 0.2, inner sep = 0, circle] {}  node (n4) [pos = 0.85, inner sep = 0, circle] {};
\path (g4) .. controls +(30:0.5) and \control{(g5)}{(180:0.25)} node (n3) [pos = 0.2, inner sep = 0, circle] {};

\filldraw [fill opacity = 0.25, purple] (g0) .. controls +(-160:0.3) and \control{(g1)}{(110:0.75)} .. controls +(-70:0.4) and  \control{(g3)}{(45:0.65)} .. controls +(210:0.5) and \control{(g4)}{(-150:0.75)} .. controls +(30:0.5) and \control{(g5)}{(180:0.25)} .. controls +(0:0.25) and \control{(g6)}{(120:0.5)} .. controls +(120:0.25) and \control{(r4)}{(90:0.35)} .. controls +(-90:2.5) and \control{(r3)}{(-30:1.5)} .. controls +(150:0.5) and \control{(r1)}{(-45:0.4)} .. controls +(135:0.45) and \control{(g0)}{(20:0.3)};

%intersection tore
\draw [blue, thick, densely dashed] (n4) .. controls +(0:0.3) and \control{(n3)}{(90:0.3)} .. controls +(180:0.3) and \control{(n4)}{(-90:0.3)};
\draw [name path = disk2, blue, thick, densely dashed] (n1) .. controls +(70:0.5) and \control{(n2)}{(160:0.5)} .. controls +(-120:0.5) and \control{(n1)}{(-10:0.5)};
\draw [name path = disk3, blue, thick, densely dashed] (r3) .. controls +(140:0.3) and \control{(r4)}{(40:0.3)} .. controls +(-40:0.3) and \control{(r3)}{(-140:0.3)};

%Superposition avec la sphère
\path [name path = garc] (g0) .. controls +(-160:0.3) and \control{(g1)}{(110:0.75)};
\path [name path = darc] (g0) .. controls +(20:0.3) and \control{(r1)}{(135:0.45)};
\path [name intersections={of=darc and elip,total=\tot}]
\foreach \s in {1,...,\tot}{coordinate (uarc-1) at (intersection-\s)};
\path [name intersections={of=garc and elip,total=\tot}]
\foreach \s in {1,...,\tot}{coordinate (uarc-2) at (intersection-\s)};
\path [name path = arc] (g3) .. controls +(210:0.5) and \control{(g4)}{(-150:0.75)};
\path [name intersections={of=arc and elip,total=\tot}]
\foreach \s in {1,...,\tot}{coordinate (arc-\s) at (intersection-\s)};
\path [name path = larc] (r3) .. controls +(-30:1.5) and \control{(r4)}{(-90:2.5)};
\path [name intersections={of=larc and elip,total=\tot}]
\foreach \s in {1,...,\tot}{coordinate (larc-\s) at (intersection-\s)};
\begin{scope}
\clip (uarc-1) rectangle (uarc-2);
\draw [blue, dotted] (4,0) arc (0:360:4 and 2);
\end{scope}
\begin{scope}
\clip (arc-1) rectangle (arc-2);
\draw [blue, dotted] (4,0) arc (0:360:4 and 2);
\end{scope}
\begin{scope}
\clip (larc-1) rectangle (larc-2);
\draw [blue, dotted] (4,0) arc (0:360:4 and 2);
\end{scope}
\begin{scope}
\clip (uarc-1) -- ++(0,0.6) -- (4.1,2.1) -- (4.1,-2.1) -- ($(larc-1)+(0,-1)$) -- ($(larc-1)+(0,1)$) -- ($(larc-2)+(0,1)$) -- ($(larc-2)+(0,-1)$) -- ($(arc-2)+(0,-1)$) -- ($(arc-2)+(0,0.5)$) -- ($(arc-1)+(0,1)$) -- ($(arc-1)+(0,-1)$) -- (-4.1,-1.3) -- (-4.1,2.1) -- ($(uarc-2)+(0,1)$) -- ($(uarc-2)+(0,-1)$) -- cycle;
\draw [blue] (4,0) arc (0:360:4 and 2);
\end{scope}
\end{scope}

\begin{scope}[xshift = 4cm, yshift = 7.5cm]
%Torus filling
\fill [blue, opacity = 0.2] (0,0) circle (4 and 2);
\fill [white](1,0) .. controls +(150:0.5) and \control{(-1,0)}{(30:0.5)} .. controls +(-30:0.5) and \control{(1,0)}{(210:0.5)}; 

%Torus drawing and coordinate definition
\path [name path = elip, blue] (4,0) arc (0:360:4 and 2);
\draw [name path = highc, blue] (1,0) .. controls +(150:0.5) and \control{(-1,0)}{(30:0.5)} .. controls +(-30:0.5) and \control{(1,0)}{(210:0.5)};
\coordinate (aa) at ($(1,0) + (30:0.2)$);
\coordinate (bb) at ($(-1,0) + (150:0.2)$);
\draw [blue] (1,0) -- (aa);
\draw [blue] (-1,0) -- (bb);

%sphere 
\coordinate (g0) at (-3,2);
\coordinate (g1) at (-3.25,1);
\coordinate (g2) at (-3,0);
\coordinate (g3) at (-3.25,-1);
\coordinate (g4) at (-2.75,-1.25);
\coordinate (g5) at (-2.25,-0.5);
\coordinate (g6) at (-1.7,-1);
\coordinate (g7) at (-1,-1.75);
\coordinate (r1) at (-2,1.5);
\coordinate (r2) at (-1.25,-0.5);
\coordinate (r3) at (-0.5,-1.5);
\coordinate (r4) at (-2,-1.5);
\coordinate (d1) at (-1.2,-1.75);
\coordinate (o0) at (-2.5,1.875);
\coordinate (o1) at (-1.85,-1);
\coordinate (o2) at (-0.75,-1.75);
\coordinate (c1) at (-2.05,0.35);
\coordinate (c2) at ($(-2.05,0.35)+(30:0.3)$);
\path [name path =md_1] (g7) .. controls +(120:1) and \control{(g0)}{(-40:1)} node [circle, inner sep= 0, pos = 0.45] (f1') {};
\coordinate (f1) at (f1');
\path [name path = boundr] (o2) .. controls +(130:1) and \control{(o0)}{(-40:1)} node [inner sep= 0, pos = 0.495] (f2') {};
\coordinate (f2) at (f2');
\path (g0) .. controls +(-120:0.5) and \control{(d1)}{(110:2.5)} node [circle, inner sep= 0, pos = 0.688] (f3') {};
\coordinate (f3) at (f3');
\path [name path = u4] (o0) .. controls +(-120:0.3) and \control{($(o2)+(-0.13,0)$)}{(115:2.5)} node [circle, inner sep= 0, pos = 0.66] (f4') {};
\coordinate (f4) at (f4');
\path (g1) .. controls +(-70:0.4) and  \control{(g3)}{(45:0.65)} node (n1) [pos = 0.2, inner sep = 0, circle] {}  node (n4) [pos = 0.85, inner sep = 0, circle] {};
\path (g4) .. controls +(30:0.5) and \control{(g5)}{(180:0.25)} node (n3) [pos = 0.2, inner sep = 0, circle] {};
\path [name path = md_1] (g7) .. controls +(120:1) and \control{(g0)}{(-40:1)};

\filldraw [fill opacity = 0.25, purple] (g0) .. controls +(-160:0.3) and \control{(g1)}{(110:0.75)} .. controls +(-70:0.4) and  \control{(g3)}{(45:0.65)} .. controls +(210:0.5) and \control{(g4)}{(-150:0.75)} .. controls +(30:0.5) and \control{(g5)}{(180:0.25)} .. controls +(0:0.25) and \control{(g6)}{(120:0.5)} .. controls +(120:0.25) and \control{(r4)}{(90:0.35)} .. controls +(-90:2.5) and \control{(r3)}{(-30:1.5)} .. controls +(150:0.5) and \control{(r1)}{(-45:0.4)} .. controls +(135:0.25) and \control{(o0)}{(-25:0.3)} .. controls +(-40:0.25) and \control{(f2)}{(110:1)} to [out = -70, in = -80] (f1) .. controls +(115:1.1) and \control{(g0)}{(-45:0.75)};
\draw [purple, name path = boundl, thin] (o0) .. controls +(-120:0.15) and \control{(f4)}{(120:0.5)} to [out = -70, in = -80] (f3) .. controls +(120:0.5) and \control{(g0)}{(-120:0.5)};
\begin{scope}
\path [name intersections={of=u4 and boundl,total=\tot}]
\foreach \s in {1,...,\tot}{coordinate (i) at (intersection-\s)};
\clip (i) -- ($(i) + (0,-0.3)$) -- ++ (0.75,0) -- ($(o0)+(1,0)$) -- ($(o0)+(-0.29,0)$) -- cycle ;
\fill [fill opacity = 0.25, purple] (o0) .. controls +(-120:0.15) and \control{(f4)}{(120:0.5)} to [out = -70, in = -80] (f3) .. controls +(120:0.5) and \control{(g0)}{(-120:0.5)} .. controls +(-45:0.75) and \control{(f1)}{(115:1.1)} to [out = -80, in = -70] (f2) .. controls +(110:1) and \control{(o0)}{(-40:0.25)};
\end{scope}
\begin{scope}[purple]
\fill [fill opacity = 0.15] (f1) -- (c1) -- (c2) -- (f2) to [out = -70, in = -80] (f1) -- cycle;
\fill [fill opacity = 0.15] (f3) -- (c1) -- (c2) -- (f4) to [out = -70, in = -80] (f3) -- cycle;
\draw [thin] (f1) -- (c1) -- (c2) -- (f2);
\draw [thin] (f3) -- (c1);
\draw [thin] (c2) -- (f4);
\end{scope}

%intersection tore
\path (r1) .. controls +(-45:0.4) and  \control{(r2)}{(110:0.5)} node (n2) [pos = 0.15, inner sep = 0, circle] {};
\draw [blue, thick, densely dashed] (n4) .. controls +(0:0.3) and \control{(n3)}{(90:0.3)} .. controls +(180:0.3) and \control{(n4)}{(-90:0.3)};

\path [name path = disk2] (n1) .. controls +(70:0.5) and \control{(n2)}{(160:0.5)} .. controls +(-120:0.5) and \control{(n1)}{(-10:0.5)};
\path [name path = mg_1] (g0) .. controls +(-120:0.5) and \control{(d1)}{(110:2.5)};
\path [name path = md_1] (g7) .. controls +(120:1) and \control{(g0)}{(-40:1)};
\path [name intersections={of=disk2 and mg_1,total=\tot}]
\foreach \s in {1,...,\tot}{coordinate (disk2g-\s) at (intersection-\s)};
\path [name intersections={of=disk2 and md_1,total=\tot}]
\foreach \s in {1,...,\tot}{coordinate (disk2d-\s) at (intersection-\s)};
\draw [blue, thick, densely dashed] (n1) .. controls +(70:0.3) and \control{(disk2g-1)}{(-155:0.2)} -- (disk2d-2) .. controls +(-135:0.3) and \control{(n1)}{(-10:0.5)};
\path [name intersections={of=disk2 and boundl,total=\tot}]
\foreach \s in {1,...,\tot}{coordinate (disk2bg-\s) at (intersection-\s)};
\path [name intersections={of=disk2 and boundr,total=\tot}]
\foreach \s in {1,...,\tot}{coordinate (disk2bd-\s) at (intersection-\s)};
\draw [blue, thick, densely dashed] (n2) .. controls +(160:0.15) and \control{(disk2bg-1)}{(30:0.15)} -- (disk2bd-2) .. controls +(30:0.1) and \control{(n2)}{(-120:0.1)};

\draw [blue, thick, densely dashed] (r3) .. controls +(140:0.3) and \control{(r4)}{(40:0.3)} .. controls +(-40:0.3) and \control{(r3)}{(-140:0.3)};

%Superposition avec la sphère
\path [name path = garc] (g0) .. controls +(-160:0.3) and \control{(g1)}{(110:0.75)};
\path [name path = darc] (g0) .. controls +(20:0.3) and \control{(r1)}{(135:0.45)};
\path [name intersections={of=darc and elip,total=\tot}]
\foreach \s in {1,...,\tot}{coordinate (uarc-1) at (intersection-\s)};
\path [name intersections={of=garc and elip,total=\tot}]
\foreach \s in {1,...,\tot}{coordinate (uarc-2) at (intersection-\s)};
\path [name path = arc] (g3) .. controls +(210:0.5) and \control{(g4)}{(-150:0.75)};
\path [name intersections={of=arc and elip,total=\tot}]
\foreach \s in {1,...,\tot}{coordinate (arc-\s) at (intersection-\s)};
\path [name path = larc] (r3) .. controls +(-30:1.5) and \control{(r4)}{(-90:2.5)};
\path [name intersections={of=larc and elip,total=\tot}]
\foreach \s in {1,...,\tot}{coordinate (larc-\s) at (intersection-\s)};
\path [name intersections={of=boundl and elip,total=\tot}]
\foreach \s in {1,...,\tot}{coordinate (llarc-\s) at (intersection-\s)};
\path [name intersections={of=md_1 and elip,total=\tot}]
\foreach \s in {1,...,\tot}{coordinate (lrarc-1) at (intersection-\s)};
\begin{scope}
\clip (uarc-1) rectangle (uarc-2);
\draw [blue, dotted] (4,0) arc (0:360:4 and 2);
\end{scope}
\begin{scope}
\clip (arc-1) rectangle (arc-2);
\draw [blue, dotted] (4,0) arc (0:360:4 and 2);
\end{scope}
\begin{scope}
\clip (larc-1) rectangle (larc-2);
\draw [blue, dotted] (4,0) arc (0:360:4 and 2);
\end{scope}
\begin{scope}
\clip (uarc-1) -- ++(0,0.6) -- (4.1,2.1) -- (4.1,-2.1) -- ($(larc-1)+(0,-1)$) -- ($(larc-1)+(0,1)$) -- ($(larc-2)+(0,1)$) -- ($(larc-2)+(0,-1)$) -- ($(arc-2)+(0,-1)$) -- ($(arc-2)+(0,0.5)$) -- ($(arc-1)+(0,1)$) -- ($(arc-1)+(0,-1)$) -- (-4.1,-1.3) -- (-4.1,2.1) -- ($(uarc-2)+(0,1)$) -- ($(uarc-2)+(0,-1)$) -- ($(lrarc-1)+(0,-0.5)$) -- ($(lrarc-1)+(0,0.5)$) -- ($(llarc-1)+(0,0.5)$) -- ($(llarc-1)+(0,-0.5)$) --cycle;
\draw [blue] (4,0) arc (0:360:4 and 2);
\end{scope}
\end{scope}

\begin{scope}[xshift = 4cm, yshift = 12.5cm]
%Torus filling
\fill [blue, opacity = 0.2] (0,0) circle (4 and 2);
\fill [white](1,0) .. controls +(150:0.5) and \control{(-1,0)}{(30:0.5)} .. controls +(-30:0.5) and \control{(1,0)}{(210:0.5)}; 

%Torus drawing and coordinate definition
\path [name path = elip, blue] (4,0) arc (0:360:4 and 2);
\draw [name path = highc, blue] (1,0) .. controls +(150:0.5) and \control{(-1,0)}{(30:0.5)} .. controls +(-30:0.5) and \control{(1,0)}{(210:0.5)};
\coordinate (aa) at ($(1,0) + (30:0.2)$);
\coordinate (bb) at ($(-1,0) + (150:0.2)$);
\draw [blue] (1,0) -- (aa);
\draw [blue] (-1,0) -- (bb);

%sphere 
\coordinate (g0) at (-3,2);
\coordinate (g1) at (-3.25,1);
\coordinate (g2) at (-3,0);
\coordinate (g3) at (-3.25,-1);
\coordinate (g4) at (-2.75,-1.25);
\coordinate (g5) at (-2.25,-0.5);
\coordinate (g6) at (-1.7,-1);
\coordinate (g7) at (-1,-1.75);
\coordinate (r1) at (-2,1.5);
\coordinate (r2) at (-1.25,-0.5);
\coordinate (r3) at (-0.5,-1.5);
\coordinate (r4) at (-2,-1.5);
\coordinate (d1) at (-1.2,-1.75);

\begin{scope}[purple, thick]
\draw (g0) .. controls +(-160:0.3) and \control{(g1)}{(110:0.75)};
\draw (g1) .. controls +(-70:0.4) and  \control{(g3)}{(45:0.65)} node (n1) [pos = 0.2, inner sep = 0, circle] {}  node (n4) [pos = 0.85, inner sep = 0, circle] {};
\draw (g3) .. controls +(210:0.5) and \control{(g4)}{(-150:0.75)};
\draw (g4) .. controls +(30:0.5) and \control{(g5)}{(180:0.25)} node (n3) [pos = 0.2, inner sep = 0, circle] {};
\draw (g5) .. controls +(0:0.25) and \control{(g6)}{(120:0.5)};
\draw (g6) .. controls +(-60:1) and \control{(g7)}{(-60:1)};
\draw [name path = md_1] (g7) .. controls +(120:1) and \control{(g0)}{(-40:1)};
\end{scope}
\begin{scope}[purple, thin]
\draw (g0) .. controls +(-120:0.5) and \control{(d1)}{(110:2.5)};
\draw (d1) .. controls +(-50:1.5) and \control{(g6)}{(-110:0.75)};
\end{scope}
\fill [purple, opacity = 0.25] (g0) .. controls +(-160:0.3) and \control{(g1)}{(110:0.75)} .. controls +(-70:0.4) and  \control{(g3)}{(45:0.65)} .. controls +(210:0.5) and \control{(g4)}{(-150:0.75)} .. controls +(30:0.5) and \control{(g5)}{(180:0.25)} .. controls +(0:0.25) and \control{(g6)}{(120:0.5)} .. controls +(-110:0.75) and \control{(g7)}{(-60:1.25)} .. controls +(120:1) and \control{(g0)}{(-40:1)};

\coordinate (o0) at (-2.5,1.875);
\coordinate (o1) at (-1.85,-1);
\coordinate (o2) at (-0.75,-1.75);
\draw [purple, thick, name path = boundr] (o0) .. controls +(-25:0.3) and \control{(r1)}{(135:0.45)} .. controls +(-45:0.4) and  \control{(r2)}{(110:0.5)} .. controls +(-70:0.5) and \control{(r3)}{(150:0.5)} .. controls +(-30:1.5) and \control{(r4)}{(-90:2.5)} .. controls +(90:0.35) and \control{(o1)}{(-135:0.1)}.. controls +(-45:0.2) and \control{($(d1)+(0.1,-0.35)$)}{(110:0.2)} .. controls +(-70:0.5) and \control{(o2)}{(-50:1)} .. controls +(130:1) and \control{(o0)}{(-40:1)};
\draw [purple, name path = boundl, thin] (o0) .. controls +(-120:0.3) and \control{($(o2)+(-0.13,0)$)}{(115:2.5)} .. controls +(-65:1.5) and \control{(o1)}{(-110:0.75)};
\fill [purple, opacity = 0.25] (o0) .. controls +(-25:0.3) and \control{(r1)}{(135:0.45)} .. controls +(-45:0.4) and  \control{(r2)}{(110:0.5)} .. controls +(-70:0.5) and \control{(r3)}{(150:0.5)} .. controls +(-30:1.5) and \control{(r4)}{(-90:2.5)} .. controls +(90:0.35) and \control{(o1)}{(-135:0.1)}.. controls +(-45:0.2) and \control{($(d1)+(0.1,-0.35)$)}{(110:0.2)} .. controls +(-70:0.5) and \control{(o2)}{(-50:1)} .. controls +(130:1) and \control{(o0)}{(-40:1)};
\path [name intersections={of=boundr and boundl,total=\tot}]
\foreach \s in {1,...,\tot}{coordinate (i) at (intersection-\s)};
\begin{scope}
\clip ($(o0)+(-0.5,0)$) -- ($(o0)+(0.5,0)$) -- (0,-2.6) -- ($(i)+(0,-0.5)$) -- ($(i)+(0,0.1)$) -- cycle;
\fill [purple, opacity = 0.25] (o0) .. controls +(-120:0.3) and \control{($(o2)+(-0.13,0)$)}{(115:2.5)} .. controls +(-65:1.5) and \control{(o1)}{(-110:0.75)} .. controls +(-45:0.2) and \control{($(d1)+(0.1,-0.35)$)}{(110:0.2)} .. controls +(-70:0.5) and \control{(o2)}{(-50:1)} .. controls +(130:1) and \control{(o0)}{(-40:1)};
\end{scope}

%cylinder
\begin{scope}[purple, xshift = -2.05cm, yshift = 0.25cm, rotate = 30]
\draw [thick] (0,0.1) arc (90 :-90: 0.05 and 0.1) -- (0.3,-0.1) arc (-90:90: 0.05 and 0.1) -- cycle;
\draw (0,-0.1) arc (-90 :-270: 0.05 and 0.1);
\draw (0.3,-0.1) arc (-90 :-270: 0.05 and 0.1);
\fill [opacity = 0.25] (0,0.1) -- (0.3,0.1) arc (90 :-90: 0.05 and 0.1) -- (0,-0.1) arc (-90 :-270: 0.05 and 0.1);
\fill [opacity = 0.25] (0,0.1) -- (0.3,0.1) arc (90 :270: 0.05 and 0.1) -- (0,-0.1) arc (-90 :90: 0.05 and 0.1);
\end{scope}

%intersection tore
\path (r1) .. controls +(-45:0.4) and  \control{(r2)}{(110:0.5)} node (n2) [pos = 0.15, inner sep = 0, circle] {};
\draw [blue, thick, densely dashed] (n4) .. controls +(0:0.3) and \control{(n3)}{(90:0.3)} .. controls +(180:0.3) and \control{(n4)}{(-90:0.3)};

\path [name path = disk2] (n1) .. controls +(70:0.5) and \control{(n2)}{(160:0.5)} .. controls +(-120:0.5) and \control{(n1)}{(-10:0.5)};
\path [name path = mg_1] (g0) .. controls +(-120:0.5) and \control{(d1)}{(110:2.5)};
\path [name path = md_1] (g7) .. controls +(120:1) and \control{(g0)}{(-40:1)};
\path [name intersections={of=disk2 and mg_1,total=\tot}]
\foreach \s in {1,...,\tot}{coordinate (disk2g-\s) at (intersection-\s)};
\path [name intersections={of=disk2 and md_1,total=\tot}]
\foreach \s in {1,...,\tot}{coordinate (disk2d-\s) at (intersection-\s)};
\draw [blue, thick, densely dashed] (n1) .. controls +(70:0.3) and \control{(disk2g-1)}{(-155:0.2)} -- (disk2d-2) .. controls +(-135:0.3) and \control{(n1)}{(-10:0.5)};
\path [name intersections={of=disk2 and boundl,total=\tot}]
\foreach \s in {1,...,\tot}{coordinate (disk2bg-\s) at (intersection-\s)};
\path [name intersections={of=disk2 and boundr,total=\tot}]
\foreach \s in {1,...,\tot}{coordinate (disk2bd-\s) at (intersection-\s)};
\draw [blue, thick, densely dashed] (n2) .. controls +(160:0.15) and \control{(disk2bg-1)}{(30:0.15)} -- (disk2bd-2) .. controls +(30:0.1) and \control{(n2)}{(-120:0.1)};

\path [name path = disk3] (r3) .. controls +(140:0.3) and \control{(r4)}{(40:0.3)} .. controls +(-40:0.3) and \control{(r3)}{(-140:0.3)};
\path [name path = md_2] (g6) .. controls +(-60:1) and \control{(g7)}{(-60:1)};
\path [name path = mg_2] (d1) .. controls +(-50:1.5) and \control{(g6)}{(-110:0.75)};
\path [name intersections={of=disk3 and mg_2,total=\tot}]
\foreach \s in {1,...,\tot}{coordinate (disk3g-\s) at (intersection-\s)};
\path [name intersections={of=disk3 and md_2,total=\tot}]
\foreach \s in {1,...,\tot}{coordinate (disk3d-\s) at (intersection-\s)};
\path [name intersections={of=disk3 and mg_1,total=\tot}]
\foreach \s in {1,...,\tot}{coordinate (disk3bg-\s) at (intersection-\s)};
\path [name intersections={of=disk3 and md_1,total=\tot}]
\foreach \s in {1,...,\tot}{coordinate (disk3bd-\s) at (intersection-\s)};
\draw [blue, thick, densely dashed] (disk3bg-1) -- (disk3bd-2) -- (disk3d-2)-- (disk3g-1) -- cycle;
\path [name intersections={of=disk3 and boundl,total=\tot}]
\foreach \s in {1,...,\tot}{coordinate (disk3ll-\s) at (intersection-\s)};
\path [name intersections={of=disk3 and boundr,total=\tot}]
\foreach \s in {1,...,\tot}{coordinate (disk3lr-\s) at (intersection-\s)};
\draw [blue, thick, densely dashed] (r3) .. controls +(140:0.3) and  \control{(disk3ll-1)}{(0:0.1)} -- (disk3lr-7) .. controls +(0:0.15) and \control{(r3)}{(-140:0.1)};
\draw [blue, thick, densely dashed] (r4) .. controls +(40:0.1) and  \control{(disk3ll-2)}{(180:0.1)} -- (disk3lr-6) .. controls +(180:0.1) and \control{(r4)}{(-40:0.15)};

%Superposition avec la sphère
\path [name path = garc] (g0) .. controls +(-160:0.3) and \control{(g1)}{(110:0.75)};
\path [name path = darc] (g0) .. controls +(20:0.3) and \control{(r1)}{(135:0.45)};
\path [name intersections={of=darc and elip,total=\tot}]
\foreach \s in {1,...,\tot}{coordinate (uarc-1) at (intersection-\s)};
\path [name intersections={of=garc and elip,total=\tot}]
\foreach \s in {1,...,\tot}{coordinate (uarc-2) at (intersection-\s)};
\path [name path = arc] (g3) .. controls +(210:0.5) and \control{(g4)}{(-150:0.75)};
\path [name intersections={of=arc and elip,total=\tot}]
\foreach \s in {1,...,\tot}{coordinate (arc-\s) at (intersection-\s)};
\path [name path = larc] (r3) .. controls +(-30:1.5) and \control{(r4)}{(-90:2.5)};
\path [name intersections={of=larc and elip,total=\tot}]
\foreach \s in {1,...,\tot}{coordinate (larc-\s) at (intersection-\s)};
\path [name intersections={of=boundl and elip,total=\tot}]
\foreach \s in {1,...,\tot}{coordinate (llarc-\s) at (intersection-\s)};
\path [name intersections={of=md_1 and elip,total=\tot}]
\foreach \s in {1,...,\tot}{coordinate (lrarc-1) at (intersection-\s)};
\path [name intersections={of=md_2 and elip,total=\tot}]
\foreach \s in {1,...,\tot}{coordinate (lrarc-2) at (intersection-\s)};
\begin{scope}
\clip (uarc-1) rectangle (uarc-2);
\draw [blue, dotted] (4,0) arc (0:360:4 and 2);
\end{scope}
\begin{scope}
\clip (arc-1) rectangle (arc-2);
\draw [blue, dotted] (4,0) arc (0:360:4 and 2);
\end{scope}
\begin{scope}
\clip (larc-1) rectangle (larc-2);
\draw [blue, dotted] (4,0) arc (0:360:4 and 2);
\end{scope}
\begin{scope}
\clip (uarc-1) -- ++(0,0.6) -- (4.1,2.1) -- (4.1,-2.1) -- ($(larc-1)+(0,-1)$) -- ($(larc-1)+(0,1)$) -- ($(llarc-2)+(0,1)$) -- ($(llarc-2)+(0,-1)$) -- ($(lrarc-2)+(0,-1)$) -- ($(lrarc-2)+(0,1)$) -- ($(larc-2)+(0,1)$) -- ($(larc-2)+(0,-1)$) -- ($(arc-2)+(0,-1)$) -- ($(arc-2)+(0,0.5)$) -- ($(arc-1)+(0,1)$) -- ($(arc-1)+(0,-1)$) -- (-4.1,-1.3) -- (-4.1,2.1) -- ($(uarc-2)+(0,1)$) -- ($(uarc-2)+(0,-1)$) -- ($(lrarc-1)+(0,-0.5)$) -- ($(lrarc-1)+(0,0.5)$) -- ($(llarc-1)+(0,0.5)$) -- ($(llarc-1)+(0,-0.5)$) --cycle;
\draw [blue] (4,0) arc (0:360:4 and 2);
\end{scope}
\end{scope}

\begin{scope}[xshift = 4cm, yshift = 17.5cm]
%Torus filling
\fill [blue, opacity = 0.2] (0,0) circle (4 and 2);
\fill [white](1,0) .. controls +(150:0.5) and \control{(-1,0)}{(30:0.5)} .. controls +(-30:0.5) and \control{(1,0)}{(210:0.5)}; 

%Torus drawing and coordinate definition
\path [name path = elip, blue] (4,0) arc (0:360:4 and 2);
\draw [name path = highc, blue] (1,0) .. controls +(150:0.5) and \control{(-1,0)}{(30:0.5)} .. controls +(-30:0.5) and \control{(1,0)}{(210:0.5)};
\coordinate (aa) at ($(1,0) + (30:0.2)$);
\coordinate (bb) at ($(-1,0) + (150:0.2)$);
\draw [blue] (1,0) -- (aa);
\draw [blue] (-1,0) -- (bb);

%sphere 
\coordinate (g0) at (-3,2);
\coordinate (g1) at (-3.25,1);
\coordinate (g2) at (-3,0);
\coordinate (g3) at (-3.25,-1);
\coordinate (g4) at (-2.75,-1.25);
\coordinate (g5) at (-2.25,-0.5);
\coordinate (g6) at (-1.7,-1);
\coordinate (g7) at (-1,-1.75);
\coordinate (r1) at (-2,1.5);
\coordinate (r2) at (-1.25,-0.5);
\coordinate (r3) at (-0.5,-1.5);
\coordinate (r4) at (-2,-1.5);
\coordinate (d1) at (-1.2,-1.75);

\begin{scope}[orange]
\draw (g0) .. controls +(-160:0.3) and \control{(g1)}{(110:0.75)};
\draw (g1) .. controls +(-70:0.4) and  \control{(g3)}{(45:0.65)} node (n1) [pos = 0.2, inner sep = 0, circle] {}  node (n4) [pos = 0.85, inner sep = 0, circle] {};
\draw (g3) .. controls +(210:0.5) and \control{(g4)}{(-150:0.75)};
\draw (g4) .. controls +(30:0.5) and \control{(g5)}{(180:0.25)} node (n3) [pos = 0.2, inner sep = 0, circle] {};
\draw (g5) .. controls +(0:0.25) and \control{(g6)}{(120:0.5)};
\draw (g6) .. controls +(-60:1) and \control{(g7)}{(-60:1)};
\draw [name path = md_1] (g7) .. controls +(120:1) and \control{(g0)}{(-40:1)};
\end{scope}
\begin{scope}[purple]
\draw (g0) .. controls +(20:0.3) and \control{(r1)}{(135:0.45)};
\draw (r1) .. controls +(-45:0.4) and  \control{(r2)}{(110:0.5)} node (n2) [pos = 0.15, inner sep = 0, circle] {};
\draw (r2) .. controls +(-70:0.5) and \control{(r3)}{(150:0.5)};
\draw (r3) .. controls +(-30:1.5) and \control{(r4)}{(-90:2.5)};
\draw (r4) .. controls +(90:0.35) and \control{(g6)}{(120:0.25)};
\end{scope}
\begin{scope}[purple!50!orange]
\draw (g0) .. controls +(-120:0.5) and \control{(d1)}{(110:2.5)};
\draw (d1) .. controls +(-50:1.5) and \control{(g6)}{(-110:0.75)};
\end{scope}
\fill [orange, opacity = 0.25] (g0) .. controls +(-160:0.3) and \control{(g1)}{(110:0.75)} .. controls +(-70:0.4) and  \control{(g3)}{(45:0.65)} .. controls +(210:0.5) and \control{(g4)}{(-150:0.75)} .. controls +(30:0.5) and \control{(g5)}{(180:0.25)} .. controls +(0:0.25) and \control{(g6)}{(120:0.5)} .. controls +(-60:1) and \control{(g7)}{(-60:1)} .. controls +(120:1) and \control{(g0)}{(-40:1)};
\fill [purple, opacity = 0.25] (g0) .. controls +(20:0.3) and \control{(r1)}{(135:0.45)} .. controls +(-45:0.4) and  \control{(r2)}{(110:0.5)} .. controls +(-70:0.5) and \control{(r3)}{(150:0.5)} .. controls +(-30:1.5) and \control{(r4)}{(-90:2.5)} .. controls +(90:0.35) and \control{(g6)}{(120:0.25)}.. controls +(-60:1) and \control{(g7)}{(-60:1)} .. controls +(120:1) and \control{(g0)}{(-40:1)};
\fill [purple!50!orange, opacity = 0.25] (g0) .. controls +(-120:0.5) and \control{(d1)}{(110:2.5)} .. controls +(-50:1.5) and \control{(g6)}{(-110:0.75)} .. controls +(-60:1) and \control{(g7)}{(-60:1)} .. controls +(120:1) and \control{(g0)}{(-40:1)};

%intersection tore
\draw [blue, thick, densely dashed] (n4) .. controls +(0:0.3) and \control{(n3)}{(90:0.3)} .. controls +(180:0.3) and \control{(n4)}{(-90:0.3)};
\draw [name path = disk2, blue, thick, densely dashed] (n1) .. controls +(70:0.5) and \control{(n2)}{(160:0.5)} .. controls +(-120:0.5) and \control{(n1)}{(-10:0.5)};
\path [name path = mg_1] (g0) .. controls +(-120:0.5) and \control{(d1)}{(110:2.5)};
\path [name path = md_1] (g7) .. controls +(120:1) and \control{(g0)}{(-40:1)};
\path [name intersections={of=disk2 and mg_1,total=\tot}]
\foreach \s in {1,...,\tot}{coordinate (disk2g-\s) at (intersection-\s)};
\path [name intersections={of=disk2 and md_1,total=\tot}]
\foreach \s in {1,...,\tot}{coordinate (disk2d-\s) at (intersection-\s)};
\draw [blue, thick, densely dashed] (disk2g-1) -- (disk2d-2);

\draw [name path = disk3, blue, thick, densely dashed] (r3) .. controls +(140:0.3) and \control{(r4)}{(40:0.3)} .. controls +(-40:0.3) and \control{(r3)}{(-140:0.3)};
\path [name path = md_2] (g6) .. controls +(-60:1) and \control{(g7)}{(-60:1)};
\path [name path = mg_2] (d1) .. controls +(-50:1.5) and \control{(g6)}{(-110:0.75)};
\path [name intersections={of=disk3 and mg_2,total=\tot}]
\foreach \s in {1,...,\tot}{coordinate (disk3g-\s) at (intersection-\s)};
\path [name intersections={of=disk3 and md_2,total=\tot}]
\foreach \s in {1,...,\tot}{coordinate (disk3d-\s) at (intersection-\s)};
\draw [blue, thick, densely dashed] (disk3g-1) -- (disk3d-2);
\path [name intersections={of=disk3 and mg_1,total=\tot}]
\foreach \s in {1,...,\tot}{coordinate (disk3bg-\s) at (intersection-\s)};
\path [name intersections={of=disk3 and md_1,total=\tot}]
\foreach \s in {1,...,\tot}{coordinate (disk3bd-\s) at (intersection-\s)};
\draw [blue, thick, densely dashed] (disk3bg-1) -- (disk3bd-2);

%Superposition avec la sphère
\path [name path = garc] (g0) .. controls +(-160:0.3) and \control{(g1)}{(110:0.75)};
\path [name path = darc] (g0) .. controls +(20:0.3) and \control{(r1)}{(135:0.45)};
\path [name intersections={of=darc and elip,total=\tot}]
\foreach \s in {1,...,\tot}{coordinate (uarc-1) at (intersection-\s)};
\path [name intersections={of=garc and elip,total=\tot}]
\foreach \s in {1,...,\tot}{coordinate (uarc-2) at (intersection-\s)};
\path [name path = arc] (g3) .. controls +(210:0.5) and \control{(g4)}{(-150:0.75)};
\path [name intersections={of=arc and elip,total=\tot}]
\foreach \s in {1,...,\tot}{coordinate (arc-\s) at (intersection-\s)};
\path [name path = larc] (r3) .. controls +(-30:1.5) and \control{(r4)}{(-90:2.5)};
\path [name intersections={of=larc and elip,total=\tot}]
\foreach \s in {1,...,\tot}{coordinate (larc-\s) at (intersection-\s)};
\begin{scope}
\clip (uarc-1) rectangle (uarc-2);
\draw [blue, dotted] (4,0) arc (0:360:4 and 2);
\end{scope}
\begin{scope}
\clip (arc-1) rectangle (arc-2);
\draw [blue, dotted] (4,0) arc (0:360:4 and 2);
\end{scope}
\begin{scope}
\clip (larc-1) rectangle (larc-2);
\draw [blue, dotted] (4,0) arc (0:360:4 and 2);
\end{scope}
\begin{scope}
\clip (uarc-1) -- ++(0,0.6) -- (4.1,2.1) -- (4.1,-2.1) -- ($(larc-1)+(0,-1)$) -- ($(larc-1)+(0,1)$) -- ($(larc-2)+(0,1)$) -- ($(larc-2)+(0,-1)$) -- ($(arc-2)+(0,-1)$) -- ($(arc-2)+(0,0.5)$) -- ($(arc-1)+(0,1)$) -- ($(arc-1)+(0,-1)$) -- (-4.1,-1.3) -- (-4.1,2.1) -- ($(uarc-2)+(0,1)$) -- ($(uarc-2)+(0,-1)$) -- cycle;
\draw [blue] (4,0) arc (0:360:4 and 2);
\end{scope}
\end{scope}

\begin{scope}[xshift = 9.5cm, yshift = 17.5cm]
\path (0,2) arc (90:145:1 and 2) node (n1) [fill, circle, inner sep = 0] {};
\path (0,2) arc (90:35:1 and 2) node (n2) [fill, circle, inner sep = 0] {};
\path (0,2) arc (90:-15:1 and 2) node (n3) [fill, circle, inner sep = 0] {}; 
\path (0,2) arc (90:195:1 and 2) node (n4) [fill, circle, inner sep = 0] {}; 
\path (0,2) arc (90:-40:1 and 2) node (n5) [fill, circle, inner sep = 0] {}; 
\path (0,2) arc (90:220:1 and 2) node (n6) [fill, circle, inner sep = 0] {}; 
\filldraw [fill opacity = 0.05] (0,0) circle (1 and 2);
\draw [blue] (n1) -- (n2);
\draw [blue] (n3) -- (n4);
\draw [blue] (n5) -- (n6);

\node at (0,1.7) [fill, circle, brown!80!blue, inner sep = 2] {};
\node at (0,0) [fill, circle, brown!80!blue, inner sep = 2] {};
\node at (0,-1) [fill, circle, brown!80!blue, inner sep = 2] {};
\node at (0,-1.7) [fill, circle, brown!80!blue, inner sep = 2] {};
\draw [brown!80!blue] (0,1.7) -- (0,-1.7);
\end{scope}

\begin{scope}[xshift = 9.5cm, yshift = 12.5cm]
\path (0,2) arc (90:145:1 and 2) node (n1) [fill, circle, inner sep = 0] {};
\path (0,2) arc (90:35:1 and 2) node (n2) [fill, circle, inner sep = 0] {};
\path (0,2) arc (90:-15:1 and 2) node (n3) [fill, circle, inner sep = 0] {}; 
\path (0,2) arc (90:195:1 and 2) node (n4) [fill, circle, inner sep = 0] {}; 
\path (0,2) arc (90:-40:1 and 2) node (n5) [fill, circle, inner sep = 0] {}; 
\path (0,2) arc (90:220:1 and 2) node (n6) [fill, circle, inner sep = 0] {}; 
\filldraw [fill opacity = 0.05] (0,0) circle (1 and 2);
\draw [blue] (n1) -- (n2);
\draw [blue] (n3) -- (n4);
\draw [blue] (n5) -- (n6);

\filldraw [fill = purple, fill opacity = 0.25, draw = purple] (0.0,0.0) circle (0.15 and 0.3);
\node at (0,1.7) [fill, circle, brown!80!blue, inner sep = 2] {};
\node at (0,0) [fill, circle, brown!80!blue, inner sep = 2] {};
\node at (0,-1) [fill, circle, brown!80!blue, inner sep = 2] {};
\node at (0,-1.7) [fill, circle, brown!80!blue, inner sep = 2] {};
\draw [brown!80!blue] (0,1.7) -- (0,-1.7);

\end{scope}

\begin{scope}[xshift = 9.5cm, yshift = 7.5cm]
\path (0,2) arc (90:145:1 and 2) node (n1) [fill, circle, inner sep = 0] {};
\path (0,2) arc (90:35:1 and 2) node (n2) [fill, circle, inner sep = 0] {};
\path (0,2) arc (90:-15:1 and 2) node (n3) [fill, circle, inner sep = 0] {}; 
\path (0,2) arc (90:195:1 and 2) node (n4) [fill, circle, inner sep = 0] {}; 
\path (0,2) arc (90:-40:1 and 2) node (n5) [fill, circle, inner sep = 0] {}; 
\path (0,2) arc (90:220:1 and 2) node (n6) [fill, circle, inner sep = 0] {};
\path (0,2) arc (90:190:1 and 2) node (n7') [fill, circle, inner sep = 0] {}; 
\path (0,2) arc (90:-10:1 and 2) node (n8') [fill, circle, inner sep = 0] {}; 
\filldraw [fill opacity = 0.05] (0,0) circle (1 and 2);

\coordinate (t) at (0,0.3);
\coordinate (n7) at (n7');
\coordinate (n8) at (n8');
\filldraw [fill = purple, fill opacity = 0.25, draw = purple] (n7) -- (t) -- (n8) arc (-10:-170: 1 and 2) -- cycle;
\node at (0,1.7) [fill, circle, brown!80!blue, inner sep = 2] {};
\node at (0,0) [fill, circle, brown!80!blue, inner sep = 2] {};
\node at (0,-1) [fill, circle, brown!80!blue, inner sep = 2] {};
\node at (0,-1.7) [fill, circle, brown!80!blue, inner sep = 2] {};
\draw [brown!80!blue] (0,1.7) -- (0,-1.7);
\draw [blue] (n1) -- (n2);
\draw [blue] (n3) -- (n4);
\draw [blue] (n5) -- (n6);
\end{scope}

\begin{scope}[xshift = 9.5cm, yshift = 2.5cm]
\path (0,2) arc (90:145:1 and 2) node (n1) [fill, circle, inner sep = 0] {};
\path (0,2) arc (90:35:1 and 2) node (n2) [fill, circle, inner sep = 0] {};
\path (0,2) arc (90:-15:1 and 2) node (n3) [fill, circle, inner sep = 0] {}; 
\path (0,2) arc (90:195:1 and 2) node (n4) [fill, circle, inner sep = 0] {}; 
\path (0,2) arc (90:-40:1 and 2) node (n5) [fill, circle, inner sep = 0] {}; 
\path (0,2) arc (90:220:1 and 2) node (n6) [fill, circle, inner sep = 0] {};
\filldraw [fill opacity = 0.05] (0,0) circle (1 and 2);
\draw [blue] (n1) -- (n2);
\draw [blue] (n3) -- (n4);
\draw [blue] (n5) -- (n6);

\node at (0,1.7) [fill, circle, brown!80!blue, inner sep = 2] {};
\node at (0,0) [fill, circle, brown!80!blue, inner sep = 2] {};
\node at (0,-1) [fill, circle, brown!80!blue, inner sep = 2] {};
\node at (0,-1.7) [fill, circle, brown!80!blue, inner sep = 2] {};
\draw [brown!80!blue] (0,1.7) -- (0,-1.7);
\filldraw [purple, fill opacity = 0.25] (0,0) circle (1 and 2);
\end{scope}

\begin{scope}[xshift = 14cm, yshift = 17.5cm]
\filldraw [fill opacity = 0.2, green!70!blue, thick] (-3,-2) rectangle (3,2);
\path (2,-1) arc (0:65:1 and 0.5) coordinate (n1);
\path (2,-1) arc (0:-65:1 and 0.5) coordinate (n2);
\path (2,-1) arc (0:115:1 and 0.5) coordinate (n3);
\path (2,-1) arc (0:-115:1 and 0.5) coordinate (n4);

\filldraw [fill opacity = 0.35, orange] (n1) arc (65:115:1 and 0.5) -- (n4) arc (-115:-65:1 and 0.5) -- cycle;
\filldraw [purple, fill opacity = 0.25, fill = purple] (n2) arc (-65:65:1 and 0.5) -- (n2);
\filldraw [purple, fill opacity = 0.25, fill = purple] (n4) arc (-115:-245:1 and 0.5) -- (n4);

\draw (0.5, 1) circle (0.75 and 0.5);
\filldraw [fill opacity = 0.35, purple] (0.5,1.5) arc (90:-90:0.75 and 0.5) -- cycle;
\filldraw [fill opacity = 0.35, orange] (0.5,1.5) arc (90:270:0.75 and 0.5) -- cycle;
\filldraw [fill opacity = 0.35, purple] (-1.5,0) circle (0.35);
\end{scope}

\begin{scope}[xshift = 14cm, yshift = 12.5cm]
\filldraw [fill opacity = 0.2, green!70!blue, thick] (-3,-2) rectangle (3,2);
\path (2,-1) arc (0:65:1 and 0.5) coordinate (n1);
\path (2,-1) arc (0:-55:1 and 0.5) coordinate (n2');
\path (2,-1) arc (0:-115:1 and 0.5) coordinate (n4);
\path (2,-1) arc (0:-125:1 and 0.5) coordinate (n4');

\filldraw [fill opacity = 0.35, purple] (n1) arc (65:115:1 and 0.5) -- (n4) arc (-115:-65:1 and 0.5) -- cycle;
\filldraw [fill opacity = 0.35, purple] (n2') arc (-55:55:1 and 0.5) -- cycle;
\filldraw [fill opacity = 0.35, purple] (n4') arc (-125:-235:1 and 0.5) -- cycle;

\path (1.25,1) arc (0:70:0.75 and 0.5) coordinate (n5);
\filldraw [fill opacity = 0.35, purple] (n5) arc (70:-70:0.75 and 0.5) -- cycle;
\filldraw [fill opacity = 0.35, purple] (0.5,1.5) arc (90:270:0.75 and 0.5) -- cycle;
\filldraw [fill opacity = 0.35, purple] (-1.5,0) circle (0.35);
\end{scope}

\begin{scope}[xshift = 14cm, yshift = 7.5cm]
\filldraw [fill opacity = 0.2, green!70!blue, thick] (-3,-2) rectangle (3,2);

\filldraw [fill opacity = 0.35, purple] (1,-1) circle (1 and 0.5);
\path (1.25,1) arc (0:70:0.75 and 0.5) coordinate (n5);
\filldraw [fill opacity = 0.35, purple] (n5) arc (70:-70:0.75 and 0.5) -- cycle;
\filldraw [fill opacity = 0.35, purple] (0.5,1.5) arc (90:270:0.75 and 0.5) -- cycle;
\filldraw [fill opacity = 0.35, purple] (-1.5,0) circle (0.35);
\end{scope}

\begin{scope}[xshift = 14cm, yshift = 2.5cm]
\filldraw [fill opacity = 0.2, green!70!blue, thick] (-3,-2) rectangle (3,2);
\filldraw [fill opacity = 0.35, purple] (-1.5,0) circle (0.35);
\filldraw [fill opacity = 0.35, purple] (1,-1) circle (1 and 0.5);
\filldraw [fill opacity = 0.35, purple] (0.5,1) circle (0.75 and 0.5);
\end{scope} 
\end{tikzpicture}
\end{center}
\vspace{-4em}
\caption{Three views of the merging process. Left: a 3D view of a double bubble (purple and orange) intersecting $\Sigma$ (blue). Middle: the membrane tree, which is progressively covered by a subtree. Right: The unfolded torus and the graph $\Gamma$ evolving throughout the merging process. It is first covered by 3 balls (orange, purple, green), then 2 (purple and green) after the cylinder merges two of them. The green ball is not shown on the left view for clarity purpose.}
\label{pic_merging_process}
\end{figure}  

The merging process starts by shrinking a bit $B_2$ and surrounding what remains by $B_3$ so that $B_2$ is now a closed ball within what was previously $B_2$ while $B_3$ is homeomorphic to $\mathbb{S}^2 \times \mathbb{S}^1$. 
On $\Sigma$, the edges that came from $D = B_1 \cap B_2$ are thus now replaced by thin bands belonging to the transformed $B_3$ with two copies of each edge as boundaries of the bands. We assume the shrinking to be small enough so that the copied edges intersect $G$ the same way the original edges do: this is always possible because the double bubble is transverse to $G$ and the vertices of $G$ have been thickened. It follows that one face of $\Gamma$ is now $B_3 \cap \Sigma$, which is $\Sigma$ with as many boundaries as there are connected components in $(B_2 \cup B_1) \cap \Sigma$ (see the unfolded torus on the right of the second row compared to the one on the right of the first row in Figure~\ref{pic_merging_process}). All the other faces of $\Gamma$ are now disks whose boundaries consist of edges with $B_3$.

We now pick an arbitrary vertex $v$ of $MT(D)$ and connect $B_1$ with $B_2$ by a small solid cylinder disjoint from $\Sigma$ whose ends are glued on $B_1$ and $B_2$ in the face indicated by $v$. This new ball is denoted by $B_{1,2}$. This step can be seen on the second row of Figure~\ref{pic_merging_process}.

For $e$ an edge of $MT(D)$ adjacent to $v$, we then grow this cylinder so that it entirely covers $e$. This amounts to connecting the surfaces of $\Sigma \cap B_{1,2}$ (intuitively, the bands between the copies of edges disappear, see the right column of Figure~\ref{pic_merging_process}). Indeed, when the edge $e$ of $MT(D)$ is covered by the growing cylinder, the two surfaces containing the initial disks of $B_1$ and $B_2$ incident to $e$ are glued along $e$. This process is then iterated by adding more edges of $MT(D)$, thus growing a subtree of $MT(D)$, until that subtree has fully grown and is $MT(D)$. This is illustrated in the third and four rows of Figure~\ref{pic_merging_process}. Eventually, the whole membrane is covered so that $B_{1,2} = B_1 \cup B_2$

We track the evolution of the ball $B_{1,2}$ during the merging process, which is parameterized by the subtree $T$ of $MT(D)$ that has been merged. This leads to the following \bfdex{merging function} $m$:

\begin{equation*}
\begin{array}{cccl}
m : & MT(D)_{\subset} & \rightarrow & \mathcal{P}(\Sigma)\\
& T & \mapsto & B_{1,2} (T) \cap \Sigma
\end{array}
\end{equation*} 
where $MT(D)_{\subset}$ is the set of subtrees of $MT(D)$ and $B_{1,2}(T)$ is the closed ball of the merging process parameterized by $T$. By definition, we have that $B_{1,2}(T)$ is contained in $B_1 \cup B_2$, transverse to $\Sigma$, and $B_{1,2}(T) \cap \Sigma = E(T)$.

It will be convenient in the following to define what kind of topology of $\Sigma$ is captured by a closed connected surface $W$ embedded on $\Sigma$. For example, an annulus can be embedded on the torus in two ways: either it contains a non-contractible curve of the torus, or it is $\pi_1$-trivial and contains irrelevant holes.
The following definition aims to ``fill'' these irrelevant holes. Each contractible boundary component $b$ of $W$ bounds a disk $D_b$ on $\Sigma$. We define $\obullet{W} = W \cup \oper{\bigcup}{b \in C(\partial W) \\ b \text{ contractible}}{}{D_b}$, the \bfdex{filled surface} induced by $W$. 

It directly follows from the definition that $\partial \obullet{W} \subset \partial W$. It is also worth noticing that if $W$ is $\pi_1$-trivial, then $\obullet{W} = D_b$ for one of the $b$.

 We can now finally use all these objects to find a fractional packing of curves in $\Gamma$ with enough weight. Let $T$ be a minimal tree of $MT(D)_{\subset}$ such that $\obullet{m(T)}$ is not a union of disks. Note that such a tree exists, since for $T=MT(D)$, we have $B_{1,2}(T)=B_1\cup B_2$ and thus $\obullet{m(T)}=\Sigma$. Necessarily, the surface $\obullet{m(T)}$ contains exactly one component that is not a disk. Indeed, let $e$ be an edge incident to a leaf of $T$. By minimality of $T$, $\obullet{m(T \smallsetminus e)}$ is a union of disjoint disks. By definition of the merging process, $m(T)$ is $m(T \smallsetminus e)$, where two copies of an edge have been merged back. If these copies were on two different $\pi_1$-trivial surfaces, their merging would have been $\pi_1$-trivial too. Hence the two copies of an edge were on the same disk of $\obullet{m(T \smallsetminus e)}$, so that the merging of $e$ glues together two distinct segments on the boundary of that disk. This operation yields either a Möbius band, which is impossible since $\Sigma$ is orientable, or an annulus $a$.

By definition of $\obullet{m(T)}$, the boundaries of $a$ are non-contractible curves of $\Sigma$. 
Lemma~\ref{lem_sphere_compressible} ensures that one of these two curves is compressible. Since they bound an annulus on $\Sigma$, they are homotopic on $\Sigma$ and thus are both compressible. Let us denote them by $b$ and $b'$.

At any stage of the merging process, some edges that were originally in $\Gamma$ are now duplicated. Therefore, while the curves $b$ and $b'$ are simple and disjoint, they might be using some of those duplicated edges of $\Gamma$, but each such edge is used at most twice. Therefore we have, denoting by $X$ the set of edges of $\Gamma$ used by $b$ and $b'$.

\begin{equation*}
\sum_{e \in X} |C(e \cap \Gamma)| \geq  \frac{1}{2} \sum_{e \in b} |C(e \cap \Gamma)| + \frac{1}{2} \sum_{e \in b'} |C(e \cap \Gamma)| \geq 2\cdot \frac{1}{2} \comprep(G,\Sigma),
\end{equation*}
where the last inequality comes from the fact that $b$ and $b'$ are compressible and the definition of $\comprep(G,\Sigma)$. This last inequality concludes the proof.
\end{proof}

This proposition directly implies Proposition~\ref{prop_bubble_tangle}, and thus Theorem~\ref{T:representativity}:

\begin{proof}[Proof of Proposition~\ref{prop_bubble_tangle}]
  A compression bubble tangle immediately satisfies the bubble tangle axioms \ref{def_T1} and \ref{def_T2} by definition, and \ref{def_T4} by Lemma~\ref{lem_T4}. For the axiom \ref{def_T3}, assume by contradiction that there exist three closed balls $B_1,B_2,B_3 \in \cT$ covering $\Sp^3$ and inducing a double bubble transverse to $\Sigma$. By Lemma~\ref{lem_Sigma_disks}, we can assume the graph $\Gamma$ induced by the intersection of the double bubble with $\Sigma$ is cellularly embedded. Then by Proposition \ref{prop_1cur_compress}, the total weight of $\Gamma$ is at least $\comprep(G,\Sigma)$. This is a contradiction with Lemma~\ref{lem_contrad_compress}.
\end{proof}

\section{Examples}\label{S:examples}

A torus knot $T_{p,q}$ is a knot embedded on a surface of an unknotted torus $\To$ in $\mathbb{S}^3$, for example a standard torus of revolution. It winds $p$ times around the revolution axis, and $q$ times around the core of the torus. We refer to Figure~\ref{F:torusknot} for an illustration of $T_{6,5}$. We can now combine our results to lower bound the treewidth of any diagram of a torus knot, thus proving Corollary~\ref{C:torusknots}.

\begin{proof}[Proof of Corollary~\ref{C:torusknots}]
Combining Proposition~\ref{P:swtw}, Theorem~\ref{T:duality} and Theorem~\ref{T:representativity}, it suffices to lower bound the compression-representativity of $T_{p,q}$ by $\min(p,q)$. The torus $\To$ has two curves which obviously bound compression disks, we prove that there are no other compressible curves. In $\mathbb{S}^3$, the torus $\To$ bounds one solid torus $\mathbb{D}^2 \times \mathbb{S}^1$ on each side, and in a solid torus $S$, exactly one of the homotopy classes of simple closed curves on the boundary is compressible, namely the one in the kernel of the inclusion map $i_*: \pi_1(\To)=\mathbb{Z}^2 \rightarrow \pi_1(S)=\mathbb{Z}$. The result follows by observing that $T_{p,q}$ intersects $p$ times the first of these homotopy classes and $q$ times the other one, and thus has compression representativity $\min(p,q)$.
\end{proof}

More generally, the same argument can be applied to lower bound the treewidth of the $(p,q)$-cabling~\cite[Section~5.2]{knotbook} of any nontrivial knot. We refer to Ozawa~\cite[Theorem~6]{ozawa} for examples of spatial embeddings of any graph with high compression representativity, and thus high spherewidth.

We conclude by observing that the proof of Theorem~\ref{T:representativity} offers more flexibility than what the theorem states and can also be applied in some settings where the compression-representativity is low. For example, a connected sum of two knots $K_1 \# K_2$ has compression-representativity two (see~\cite[Corollary~9]{ozawa}), but if one these two knots, say $K_1$, has high compression-representativity separately, then we can still define a bubble-tangle of high order by considering as big sides the balls containing the surface that $K_1$ is embedded on.

\subparagraph*{Acknowledgements.} We would like to thank Pierre Dehornoy and Saul Schleimer for helpful discussions, and the anonymous reviewers for their feedback which allowed us to significantly improve the paper.

\bibliographystyle{plainurl}
\bibliography{biblio}

\appendix

\section{Covering a torus with three disks from a double bubble}\label{A:torus}

\begin{figure}[ht]
\begin{center}
%\tikzsetnextfilename{torus_3_disks}
\begin{tikzpicture}[scale = 0.8]
\filldraw [fill = purple!90! , draw = black] (-4,-3) rectangle (4,3);
\foreach \i in {0,1,...,5}
{
	\coordinate (n\i) at (60*\i: 2);
}
\filldraw [fill = green!60!purple , draw = black] (n0) -- (n1) -- (n2) -- (n3) -- (n4) -- (n5) -- cycle;

\filldraw [fill = blue!60! , draw = black] (n0) -- (4,0) -- (4,-2) -- (n5) -- cycle;
\filldraw [fill = blue!60! , draw = black] (n3) -- (-4,0) -- (-4,-2) .. controls +(1.5,-0.25) and \control{(-2,-3)}{(-0.5,0.5)} -- (-0.5,-3) -- (n4) -- cycle;
\filldraw [fill = blue!60! , draw = black] (n2) -- (-2,3) -- (-0.5,3) -- (n1) -- cycle;

\end{tikzpicture}
\end{center}
\caption{Covering a torus with three disks from a double bubble.}
\label{pic_torus_3_disks}
\end{figure}
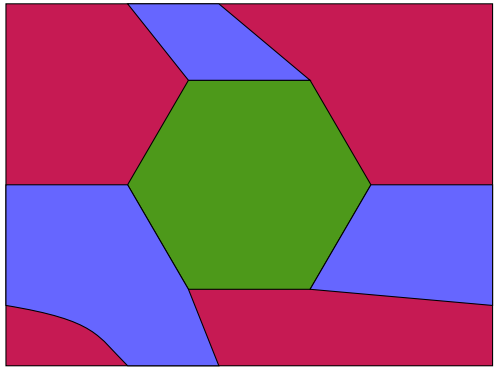

A torus can be covered by a double bubble which induces three disks on the torus, as pictured in Figure~\ref{pic_torus_3_disks}. This double bubble is described in Figures~\ref{pic_double_bubble_example_1}, \ref{pic_double_bubble_example_2}, \ref{pic_double_bubble_example_3}, \ref{pic_double_bubble_example_4} and one can prove that this is up to homeomorphism the only double bubble realizing such a cover. 

\begin{figure}[ht]
\begin{center}
%\tikzsetnextfilename{double_bubble_example_1}
\begin{tikzpicture}[scale = 0.75]
\fill [fill opacity = 0.2, blue] (0,0) -- (6,-1) -- (6,1) -- (0,2) --cycle;
\filldraw [fill opacity = 0.2, blue] (6,-1) -- (7.5,0) -- (7.5,2) -- (6,1);
\filldraw [fill opacity = 0.2, blue] ($(7.5,0)+(0.5,0.333)$) -- (9,1) -- (9,3) -- ($(7.5,2)+(0.5,0.333)$);
\fill [opacity = 0.2, blue] (0,0) -- (6,-1) -- (7.5,0) -- ++(-1.25,0.208) -- ++(0.5,0.333) -- ++(1.25,-0.208) -- (9,1) -- (3,2) --cycle;
\fill [blue, opacity = 0.2] (7.5,0) -- ++(-1.25,0.208) -- ++(0,2) -- ++(1.25,-0.208) -- cycle;
\fill [blue, opacity = 0.2, xshift =0.5cm, yshift=0.333cm] (7.5,0) -- ++(-1.25,0.208) -- ++(0,2) -- ++(1.25,-0.208) -- cycle;
\fill [blue, opacity = 0.2] (7.5,0)++(-1.25,0.208) -- ++(0.5,0.333) -- ++(0,2) -- ++(-0.5,-0.333) --cycle;
\draw [blue] (7.5,0) -- ++(-1.25,0.208) -- ++(0.5,0.333) -- ++(1.25,-0.208);
\draw [blue] (7.5,0) ++(-1.25,0.208) -- ++(0,2);
\draw [blue, xshift =0.5cm, yshift=0.333cm] (7.5,0) ++(-1.25,0.208) -- ++(0,2);
\draw [blue, xshift =0.5cm, yshift=0.333cm] (7.5,0) -- ++(0,2);
\draw [blue] (0,0) -- (6,-1) -- (6,1);
\draw [blue] (0,0) -- (0,2);

\fill [fill opacity = 0.2, blue] (0,0) -- (3,2) -- (3,4) -- (0,2) --cycle;
\fill [fill opacity = 0.2, blue] (3,2) -- (9,1) -- (9,3) -- (3,4) --cycle;
\filldraw [fill opacity = 0.2, blue, yshift = 2cm] (0,0) -- (6,-1) -- (7.5,0) -- ++(-1.25,0.208) -- ++(0.5,0.333) -- ++(1.25,-0.208) -- (9,1) -- (3,2) --cycle;
\draw [blue, dotted] (0,0) -- (3,2) -- (3,4);
\draw [dotted, blue] (3,2) --(9,1);

\begin{scope}[xshift = 9.5cm]
\begin{scope}[red, opacity = 0.2]
\fill (7.5,0.75) ++(-1, 0.166) ++(0,0.75) -- ++(0.5,0.333) -- ++(1,-0.166) -- ++(-0.5,-0.333) -- cycle;
\fill (7.5,0.75) ++(-1, 0.166) -- ++(0.5,0.333) -- ++(0,0.75) -- ++(-0.5,-0.333) -- cycle;
\fill (7.5,0.75) ++(-1, 0.166) -- ++(1,-0.166) -- ++(0.5,0.333) -- ++(-1,0.166) -- cycle;
\filldraw (7.5, 0.75) ++(0.5, 0.333) ++(0, 0.75) -- ($(6,1) +(2,1.333)$) -- ++(-0.5, -0.333) -- ++(0,-0.2) -- (6,0.8) -- (6,0.5) --cycle;
\filldraw (7.5, 0.75) ++(0,0.75) -- ++(0,0.3) -- (6,0.8) -- (6,0.5) --cycle;
\filldraw [red, opacity = 0.4] (6,0.5) -- (6,0.8) -- (0,1.8) -- (0,1.5) --cycle;
\filldraw [red, opacity = 0.4] (0,1.5) -- (3,3.5) -- (3,3.8) -- (0,1.8) --cycle;
\filldraw [red, opacity = 0.4] (3,3.5) -- (9, 2.5) -- (9,2.8) -- (3,3.8) --cycle;
\filldraw [red, opacity = 0.4] (9, 2.5) -- (8,1.3) -- ++(0,0.3) -- (9,2.8) --cycle;
\filldraw (7.5,1.5) -- ++(-1,0.166) -- ++(0,-0.75) -- (7.5,0.75) -- (7.5,0) -- ++(-1.25,0.208) -- ($(7.5,2)+(-1.25,0.208)$) --(7.5,2) --cycle;
\filldraw [xshift = 0.5cm, yshift = 0.333 cm] (7.5,1.5) -- ++(-1,0.166) -- ++(0,-0.75) -- (7.5,0.75) -- (7.5,0) -- ++(-1.25,0.208) -- ($(7.5,2)+(-1.25,0.208)$) --(7.5,2) --cycle;
\filldraw (7.5,0) -- ++(0.5,0.333) -- ++(-1.25,0.208) -- ++(0,2) -- ++(-0.5,-0.333) -- ++(0,-2) -- ++(1.25,-0.208) -- cycle; 
\filldraw (7.5,2) -- ++(-1.25,0.208) -- ++(0.5,0.333) -- ++(1.25,-0.208) -- cycle;
\filldraw (7.5,0) -- ++(0.5,0.333) -- ($(7.5,0.75) +(0.5,0.333)$) -- (7.5,0.75) -- cycle;
\end{scope}
\end{scope}
\end{tikzpicture}
\caption{The spheres $S_1$ (blue) and $S_2$ (red) up to thickening the ribbon part.}
\label{pic_double_bubble_example_1}
\end{center}
\end{figure}
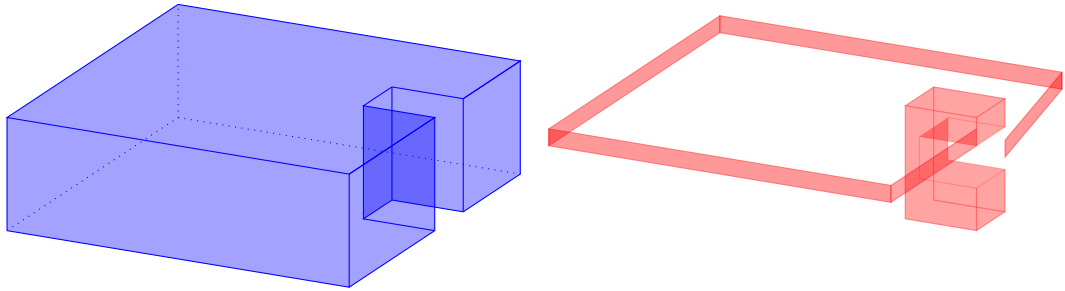

\begin{figure}[ht]
\begin{center}
%\tikzsetnextfilename{double_bubble_example_2}
\begin{tikzpicture}[scale = 0.75]
\filldraw [fill opacity = 0.2, blue!50!red!50!] (0,0) -- (6,-1) -- (6,1) -- (0,2) --cycle;
\filldraw [fill opacity = 0.2, blue!50!red!50!] (6,-1) -- (9,1) -- (9,3) -- (6,1);
\fill [fill opacity = 0.2, blue!50!red!50!] (0,0) -- (6,-1) -- (9,1) -- (3,2) --cycle;
\fill [fill opacity = 0.2, blue!50!red!50!] (0,0) -- (3,2) -- (3,4) -- (0,2) --cycle;
\fill [fill opacity = 0.2, blue!50!red!50!] (3,2) -- (9,1) -- (9,3) -- (3,4) --cycle;
\filldraw [fill opacity = 0.2, blue!50!red!50!] (0,2) -- (6,1) -- (9,3) -- (3,4) --cycle;
\draw [blue!50!red!50!, dotted] (0,0) -- (3,2) -- (3,4);
\draw [dotted, blue!50!red!50!] (3,2) --(9,1);

%Encoche
\fill [fill = white, opacity = 0.05] (7.5, 0.75) -- ++(0.5, 0.333) -- ++(0, 0.75) -- ++(-0.5, -0.333) -- cycle;
\fill [blue!50!red!50!, opacity = 0.2] (7.5,0.75) -- ++(-1, 0.166) -- ++(0,0.75) -- ++(1,-0.166) -- cycle;
\fill [blue!50!red!50!, opacity = 0.2] (7.5,0.75) ++(-1, 0.166) ++(0,0.75) -- ++(0.5,0.333) -- ++(1,-0.166) -- ++(-0.5,-0.333) -- cycle;
\fill [blue!50!red!50!, opacity = 0.2] (7.5,0.75) ++(-1, 0.166) -- ++(0.5,0.333) -- ++(0,0.75) -- ++(-0.5,-0.333) -- cycle;
\fill [blue!50!red!50!, opacity = 0.2] (7.5,0.75) ++(-1, 0.166) -- ++(1,-0.166) -- ++(0.5,0.333) -- ++(-1,0.166) -- cycle;
\fill [blue!50!red!50!, opacity = 0.2] (7.5,0.75) ++(0.5, 0.333) -- ++(-1, 0.166) -- ++(0,0.75) -- ++(1,-0.166) -- cycle;

\draw [blue!50!red!50!] (7.5, 0.75) -- ++(0.5, 0.333) -- ++(0, 0.75) -- ++(-0.5, -0.333) -- cycle;
\draw [blue!50!red!50!, dotted, thick] (7.5,0.75) -- ++(-1, 0.166) -- ++(0,0.75) -- ++(1,-0.166);
\draw [blue!50!red!50!, dotted, thick] (7.5,0.75) ++(-1, 0.166) ++(0,0.75) -- ++(0.5,0.333) -- ++(1,-0.166);
\draw [blue!50!red!50!, dotted, thick] (7.5,0.75) ++(-1, 0.166) -- ++(0.5,0.333) -- ++(1,-0.166);
\draw [blue!50!red!50!, dotted, thick] (7.5,0.75) ++(-1, 0.166) ++(0.5,0.333) -- ++(0,0.75);

\begin{scope}[xshift = 9.5cm]
\filldraw [fill opacity = 0.2, blue] (0,0) -- (6,-1) -- (6,1) -- (0,2) --cycle;
\filldraw [fill opacity = 0.2, blue] (6,-1) -- (9,1) -- (9,3) -- (6,1);
\fill [fill opacity = 0.2, blue] (0,0) -- (6,-1) -- (9,1) -- (3,2) --cycle;
\fill [fill opacity = 0.2, blue] (0,0) -- (3,2) -- (3,4) -- (0,2) --cycle;
\fill [fill opacity = 0.2, blue] (3,2) -- (9,1) -- (9,3) -- (3,4) --cycle;
\filldraw [fill opacity = 0.2, blue] (0,2) -- (6,1) -- (9,3) -- (3,4) --cycle;
\draw [blue, dotted] (0,0) -- (3,2) -- (3,4);
\draw [dotted, blue] (3,2) --(9,1);

%Encoche
\fill [fill = white, opacity = 0.05] (7.5, 0.75) -- ++(0.5, 0.333) -- ++(0, 0.75) -- ++(-0.5, -0.333) -- cycle;
\fill [blue, opacity = 0.2] (7.5,0.75) -- ++(-1, 0.166) -- ++(0,0.75) -- ++(1,-0.166) -- cycle;
\fill [blue!50!red!50!, opacity = 0.2] (7.5,0.75) ++(-1, 0.166) ++(0,0.75) -- ++(0.5,0.333) -- ++(1,-0.166) -- ++(-0.5,-0.333) -- cycle;
\fill [blue!50!red!50!, opacity = 0.2] (7.5,0.75) ++(-1, 0.166) -- ++(0.5,0.333) -- ++(0,0.75) -- ++(-0.5,-0.333) -- cycle;
\fill [blue!50!red!50!, opacity = 0.2] (7.5,0.75) ++(-1, 0.166) -- ++(1,-0.166) -- ++(0.5,0.333) -- ++(-1,0.166) -- cycle;
\fill [blue, opacity = 0.2] (7.5,0.75) ++(0.5, 0.333) -- ++(-1, 0.166) -- ++(0,0.75) -- ++(1,-0.166) -- cycle;

\draw [blue] (7.5, 0.75) -- ++(0.5, 0.333) -- ++(0, 0.75) -- ++(-0.5, -0.333) -- cycle;
\draw [blue, dotted, thick] (7.5,0.75) -- ++(-1, 0.166) -- ++(0,0.75) -- ++(1,-0.166);
\draw [blue, dotted, thick] (7.5,0.75) ++(-1, 0.166) ++(0,0.75) -- ++(0.5,0.333) -- ++(1,-0.166);
\draw [blue, dotted, thick] (7.5,0.75) ++(-1, 0.166) -- ++(0.5,0.333) -- ++(1,-0.166);
\draw [blue, dotted, thick] (7.5,0.75) ++(-1, 0.166) ++(0.5,0.333) -- ++(0,0.75);

\begin{scope}[red, opacity = 0.2]
\fill (7.5,0.75) ++(-1, 0.166) ++(0,0.75) -- ++(0.5,0.333) -- ++(1,-0.166) -- ++(-0.5,-0.333) -- cycle;
\fill (7.5,0.75) ++(-1, 0.166) -- ++(0.5,0.333) -- ++(0,0.75) -- ++(-0.5,-0.333) -- cycle;
\fill (7.5,0.75) ++(-1, 0.166) -- ++(1,-0.166) -- ++(0.5,0.333) -- ++(-1,0.166) -- cycle;
\filldraw (7.5, 0.75) ++(0.5, 0.333) ++(0, 0.75) -- ($(6,1) +(2,1.333)$) -- ++(-0.5, -0.333) -- ++(0,-0.2) -- (6,0.8) -- (6,0.5) --cycle;
\filldraw (7.5, 0.75) ++(0,0.75) -- ++(0,0.3) -- (6,0.8) -- (6,0.5) --cycle;
\filldraw [red, opacity = 0.4] (6,0.5) -- (6,0.8) -- (0,1.8) -- (0,1.5) --cycle;
\filldraw [red, opacity = 0.4] (0,1.5) -- (3,3.5) -- (3,3.8) -- (0,1.8) --cycle;
\filldraw [red, opacity = 0.4] (3,3.5) -- (9, 2.5) -- (9,2.8) -- (3,3.8) --cycle;
\filldraw [red, opacity = 0.4] (9, 2.5) -- (8,1.3) -- ++(0,0.3) -- (9,2.8) --cycle;
\filldraw (7.5,1.5) -- ++(-1,0.166) -- ++(0,-0.75) -- (7.5,0.75) -- (7.5,0) -- ++(-1.25,0.208) -- ($(7.5,2)+(-1.25,0.208)$) --(7.5,2) --cycle;
\filldraw [xshift = 0.5cm, yshift = 0.333 cm] (7.5,1.5) -- ++(-1,0.166) -- ++(0,-0.75) -- (7.5,0.75) -- (7.5,0) -- ++(-1.25,0.208) -- ($(7.5,2)+(-1.25,0.208)$) --(7.5,2) --cycle;
\filldraw (7.5,0) -- ++(0.5,0.333) -- ++(-1.25,0.208) -- ++(0,2) -- ++(-0.5,-0.333) -- ++(0,-2) -- ++(1.25,-0.208) -- cycle; 
\filldraw (7.5,2) -- ++(-1.25,0.208) -- ++(0.5,0.333) -- ++(1.25,-0.208) -- cycle;
\filldraw (7.5,0) -- ++(0.5,0.333) -- ($(7.5,0.75) +(0.5,0.333)$) -- (7.5,0.75) -- cycle;
\end{scope}
\end{scope}
\end{tikzpicture}
\caption{The ball $B_1 \cup B_2$ in purple, $B_3$ is the complementary.}
\label{pic_double_bubble_example_2}
\end{center}
\end{figure}
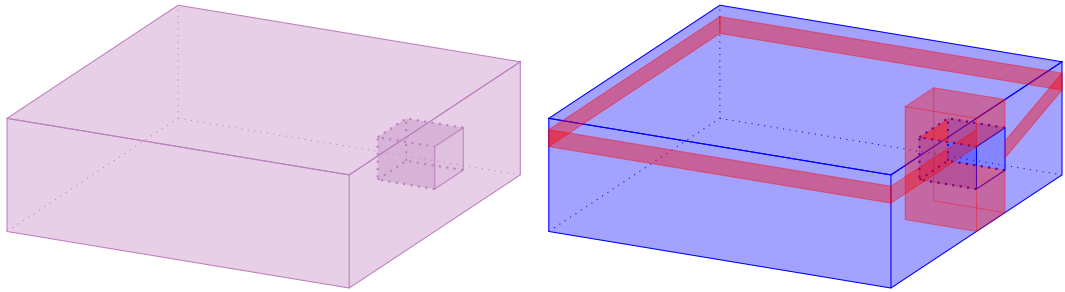

\begin{figure}[ht]
\begin{center}
%\tikzsetnextfilename{double_bubble_example_3}
\begin{tikzpicture}[scale = 0.75]
\filldraw [fill opacity = 0.2, blue] (0,0) -- (6,-1) -- (6,1) -- (0,2) --cycle;
\filldraw [fill opacity = 0.2, blue] (6,-1) -- (9,1) -- (9,3) -- (6,1);
\fill [fill opacity = 0.2, blue] (0,0) -- (6,-1) -- (9,1) -- (3,2) --cycle;
\fill [fill opacity = 0.2, blue] (0,0) -- (3,2) -- (3,4) -- (0,2) --cycle;
\fill [fill opacity = 0.2, blue] (3,2) -- (9,1) -- (9,3) -- (3,4) --cycle;
\filldraw [fill opacity = 0.2, blue] (0,2) -- (6,1) -- (9,3) -- (3,4) --cycle;
\draw [blue, dotted] (0,0) -- (3,2) -- (3,4);
\draw [dotted, blue] (3,2) --(9,1);

%Encoche
\fill [fill = white, opacity = 0.05] (7.5, 0.75) -- ++(0.5, 0.333) -- ++(0, 0.75) -- ++(-0.5, -0.333) -- cycle;
\fill [blue, opacity = 0.2] (7.5,0.75) -- ++(-1, 0.166) -- ++(0,0.75) -- ++(1,-0.166) -- cycle;
\fill [blue!50!red!50!, opacity = 0.2] (7.5,0.75) ++(-1, 0.166) ++(0,0.75) -- ++(0.5,0.333) -- ++(1,-0.166) -- ++(-0.5,-0.333) -- cycle;
\fill [blue!50!red!50!, opacity = 0.2] (7.5,0.75) ++(-1, 0.166) -- ++(0.5,0.333) -- ++(0,0.75) -- ++(-0.5,-0.333) -- cycle;
\fill [blue!50!red!50!, opacity = 0.2] (7.5,0.75) ++(-1, 0.166) -- ++(1,-0.166) -- ++(0.5,0.333) -- ++(-1,0.166) -- cycle;
\fill [blue, opacity = 0.2] (7.5,0.75) ++(0.5, 0.333) -- ++(-1, 0.166) -- ++(0,0.75) -- ++(1,-0.166) -- cycle;

\draw [blue] (7.5, 0.75) -- ++(0.5, 0.333) -- ++(0, 0.75) -- ++(-0.5, -0.333) -- cycle;
\draw [blue, dotted, thick] (7.5,0.75) -- ++(-1, 0.166) -- ++(0,0.75) -- ++(1,-0.166);
\draw [blue, dotted, thick] (7.5,0.75) ++(-1, 0.166) ++(0,0.75) -- ++(0.5,0.333) -- ++(1,-0.166);
\draw [blue, dotted, thick] (7.5,0.75) ++(-1, 0.166) -- ++(0.5,0.333) -- ++(1,-0.166);
\draw [blue, dotted, thick] (7.5,0.75) ++(-1, 0.166) ++(0.5,0.333) -- ++(0,0.75);

%Torus

\def\eps{0.1}
\begin{scope}[thick, yshift = 0cm]
\draw ($(0,0) +(\eps,\eps)$) -- ($(6,-1) +(-\eps, \eps)$) -- ($(6,1) +(-\eps, -\eps)$) -- ($(0,2) +(\eps,-\eps)$) --cycle;
\draw ($(6,-1) +(-\eps, \eps)$) -- ($(9,1) +(-\eps, \eps)$) -- ($(9,3) +(-\eps, -\eps)$) -- ($(6,1) +(-\eps, -\eps)$);
\draw ($(0,2) +(\eps, -\eps)$) -- ($(3,4) +(\eps, -\eps)$) --  ($(9,3) +(-\eps, -\eps)$);
\draw ($(0,0) +(\eps, \eps)$) -- ($(3,2) +(\eps, \eps)$) -- ($(3,4) +(\eps, -\eps)$);
\draw ($(3,2) +(\eps, \eps)$) -- ($(9,1) +(-\eps, \eps)$);
\end{scope}

\begin{scope}[thick, yshift = 0cm]
\draw ($(1,-0.166) +(\eps,\eps) +(0.5,0.333)$) -- ($(5,-0.833) +(-\eps, \eps) +(0.5,0.333)$) -- ($(5,1.166) +(-\eps, -\eps) +(0.5,0.333)$) -- ($(1,1.833) +(\eps,-\eps) +(0.5,0.333)$) --cycle;
\draw ($(5,-0.833) +(-\eps, \eps) +(0.5,0.333)$) -- ($(8,1.166) +(-\eps, \eps) +(-0.5,-0.333)$) -- ($(8,3.166) +(-\eps, -\eps) +(-0.5,-0.333)$) -- ($(5,1.166) +(0.5,0.333) +(-\eps, -\eps)$);
\draw ($(1,1.833) +(\eps, -\eps) +(0.5,0.333)$) -- ($(4,3.833) +(\eps, -\eps) +(-0.5,-0.333)$) --  ($(8,3.166) +(-\eps, -\eps) +(-0.5,-0.333)$);
\draw ($(1,-0.166) +(0.5,0.333) +(\eps, \eps)$) -- ($(4,1.833) +(-0.5,-0.333) +(\eps, \eps)$) -- ($(4,3.833) +(-0.5,-0.333) +(\eps, -\eps)$);
\draw ($(4,1.833) +(-0.5,-0.333) +(\eps, \eps)$) -- ($(8,1.166) +(-0.5,-0.333) +(-\eps, \eps)$);
\end{scope}

%Torus

\def\eps{0.1}
\begin{scope}[thick, xshift=9.5cm]
\draw ($(0,0) +(\eps,\eps)$) -- ($(6,-1) +(-\eps, \eps)$) -- ($(6,1) +(-\eps, -\eps)$) -- ($(0,2) +(\eps,-\eps)$) --cycle;
\draw ($(6,-1) +(-\eps, \eps)$) -- ($(9,1) +(-\eps, \eps)$) -- ($(9,3) +(-\eps, -\eps)$) -- ($(6,1) +(-\eps, -\eps)$);
\draw ($(0,2) +(\eps, -\eps)$) -- ($(3,4) +(\eps, -\eps)$) --  ($(9,3) +(-\eps, -\eps)$);
\draw ($(0,0) +(\eps, \eps)$) -- ($(3,2) +(\eps, \eps)$) -- ($(3,4) +(\eps, -\eps)$);
\draw ($(3,2) +(\eps, \eps)$) -- ($(9,1) +(-\eps, \eps)$);

\draw ($(1,-0.166) +(\eps,\eps) +(0.5,0.333)$) -- ($(5,-0.833) +(-\eps, \eps) +(0.5,0.333)$) -- ($(5,1.166) +(-\eps, -\eps) +(0.5,0.333)$) -- ($(1,1.833) +(\eps,-\eps) +(0.5,0.333)$) --cycle;
\draw ($(5,-0.833) +(-\eps, \eps) +(0.5,0.333)$) -- ($(8,1.166) +(-\eps, \eps) +(-0.5,-0.333)$) -- ($(8,3.166) +(-\eps, -\eps) +(-0.5,-0.333)$) -- ($(5,1.166) +(0.5,0.333) +(-\eps, -\eps)$);
\draw ($(1,1.833) +(\eps, -\eps) +(0.5,0.333)$) -- ($(4,3.833) +(\eps, -\eps) +(-0.5,-0.333)$) --  ($(8,3.166) +(-\eps, -\eps) +(-0.5,-0.333)$);
\draw ($(1,-0.166) +(0.5,0.333) +(\eps, \eps)$) -- ($(4,1.833) +(-0.5,-0.333) +(\eps, \eps)$) -- ($(4,3.833) +(-0.5,-0.333) +(\eps, -\eps)$);
\draw ($(4,1.833) +(-0.5,-0.333) +(\eps, \eps)$) -- ($(8,1.166) +(-0.5,-0.333) +(-\eps, \eps)$);
\end{scope}

\begin{scope}[red, opacity = 0.2]
\fill (7.5,0.75) ++(-1, 0.166) ++(0,0.75) -- ++(0.5,0.333) -- ++(1,-0.166) -- ++(-0.5,-0.333) -- cycle;
\fill (7.5,0.75) ++(-1, 0.166) -- ++(0.5,0.333) -- ++(0,0.75) -- ++(-0.5,-0.333) -- cycle;
\fill (7.5,0.75) ++(-1, 0.166) -- ++(1,-0.166) -- ++(0.5,0.333) -- ++(-1,0.166) -- cycle;
\filldraw (7.5, 0.75) ++(0.5, 0.333) ++(0, 0.75) -- ($(6,1) +(2,1.333)$) -- ++(-0.5, -0.333) -- ++(0,-0.2) -- (6,0.8) -- (6,0.5) --cycle;
\filldraw (7.5, 0.75) ++(0,0.75) -- ++(0,0.3) -- (6,0.8) -- (6,0.5) --cycle;
\filldraw [red, opacity = 0.4] (6,0.5) -- (6,0.8) -- (0,1.8) -- (0,1.5) --cycle;
\filldraw [red, opacity = 0.4] (0,1.5) -- (3,3.5) -- (3,3.8) -- (0,1.8) --cycle;
\filldraw [red, opacity = 0.4] (3,3.5) -- (9, 2.5) -- (9,2.8) -- (3,3.8) --cycle;
\filldraw [red, opacity = 0.4] (9, 2.5) -- (8,1.3) -- ++(0,0.3) -- (9,2.8) --cycle;
\filldraw (7.5,1.5) -- ++(-1,0.166) -- ++(0,-0.75) -- (7.5,0.75) -- (7.5,0) -- ++(-1.25,0.208) -- ($(7.5,2)+(-1.25,0.208)$) --(7.5,2) --cycle;
\filldraw [xshift = 0.5cm, yshift = 0.333 cm] (7.5,1.5) -- ++(-1,0.166) -- ++(0,-0.75) -- (7.5,0.75) -- (7.5,0) -- ++(-1.25,0.208) -- ($(7.5,2)+(-1.25,0.208)$) --(7.5,2) --cycle;
\filldraw (7.5,0) -- ++(0.5,0.333) -- ++(-1.25,0.208) -- ++(0,2) -- ++(-0.5,-0.333) -- ++(0,-2) -- ++(1.25,-0.208) -- cycle; 
\filldraw (7.5,2) -- ++(-1.25,0.208) -- ++(0.5,0.333) -- ++(1.25,-0.208) -- cycle;
\filldraw (7.5,0) -- ++(0.5,0.333) -- ($(7.5,0.75) +(0.5,0.333)$) -- (7.5,0.75) -- cycle;
\end{scope}

\begin{scope}[xshift = 5.5cm, yshift = -5.5cm]
\fill [purple!75!] (0,0) rectangle (6,4);
\fill[blue!60] (0,3) rectangle (4.5,4);
\fill[blue!60] (0,0) rectangle (4.5,2.5);
\fill[blue!60] (5,0) -- (5,1.5) -- (5.5,2.5) --(6,2.5) -- (6,0) -- cycle;
\fill[blue!60] (5,2) -- (5,4) -- (6,4) --(6,3) -- (5.5,3) -- cycle;
\draw (0,0) rectangle (6,4);
\fill (4.5,2.5) rectangle (5,1);
\begin{scope}[thick]
\draw  (0,2.5) -- (4.5,2.5) -- (4.5,0);
\draw (0,3) --  (4.5,3) -- (4.5,4);
\draw (5,0) -- (5,1.5) -- (5.5,2.5) --(6,2.5);
\draw (6,3) -- (5.5,3)--(5,2) -- (5,4);
\end{scope}
\end{scope}

\end{tikzpicture}
\caption{The double bubble and a torus inside, covered by three disks (up to thickening the ribbon part of $S_2$ by pushing it into $S_1$).}
\label{pic_double_bubble_example_3}
\end{center}
\end{figure}
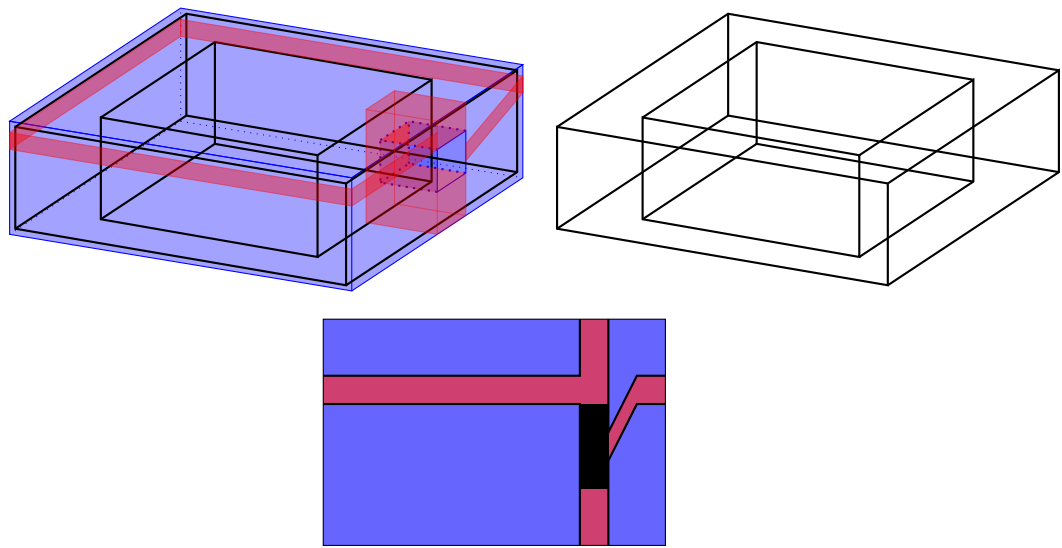

\begin{figure}[ht]
\begin{center}
%\tikzsetnextfilename{double_bubble_example_4}
\begin{tikzpicture}[scale = 0.75]
\filldraw [red, opacity = 0.2] (7.5, 0.75) ++(0,0.75) -- ++(0,0.3) -- (6,0.8) -- (6,0.5) --cycle;
\filldraw [red, opacity = 0.2] (6,0.5) -- (6,0.8) -- (0,1.8) -- (0,1.5) --cycle;
\filldraw [red, opacity = 0.2] (0,1.5) -- (3,3.5) -- (3,3.8) -- (0,1.8) --cycle;
\filldraw [red, opacity = 0.2] (3,3.5) -- (9, 2.5) -- (9,2.8) -- (3,3.8) --cycle;
\filldraw [red, opacity = 0.2] (9, 2.5) -- (8,1.3) -- ++(0,0.3) -- (9,2.8) --cycle;
\filldraw [red, opacity = 0.2] (7.5,1.5) -- ++(-1,0.166) -- ++(0,-0.75) -- (7.5,0.75) -- (7.5,0) -- ++(-1.25,0.208) -- ($(7.5,2)+(-1.25,0.208)$) --(7.5,2) --cycle;
\filldraw [xshift = 0.5cm, yshift = 0.333 cm, red, opacity = 0.2] (7.5,1.5) -- ++(-1,0.166) -- ++(0,-0.75) -- (7.5,0.75) -- (7.5,0) -- ++(-1.25,0.208) -- ($(7.5,2)+(-1.25,0.208)$) --(7.5,2) --cycle;
\filldraw [red, opacity = 0.2] ($(7.5,0)+(0.5,0.333)+(-1.25,0.208)$) -- ++(0,2) -- ++(-0.5,-0.333) -- ++(0,-2) -- cycle;
\draw [purple, thick] (7.5,2) -- ++(-1.25,0.208) -- ++(0.5,0.33) -- ++(1.25,-0.208) -- ++ (0,-0.5) -- ++(-1,0.166) -- ++ (0,-0.75) -- ++(1,-0.166) -- ++(0,-0.75) --++(-1.25,0.208) -- ++(-0.5,-0.33) -- ++(1.25,-0.208) -- ++(0,0.75) -- ++(-1,0.166) -- ++(0,0.75) -- ++(1,-0.166) -- (6,0.5) -- (0,1.5) -- (3,3.5) -- (9,2.5) -- (8,1.3) -- (8,1.6) -- (9,2.8) -- (3,3.8) -- (0,1.8) -- (6,0.8) -- (7.5,1.8) -- cycle;

\begin{scope}[xshift = 9.5cm]
\fill [red, opacity = 0.2] (7.5,0.75) ++(-1, 0.166) ++(0,0.75) -- ++(0.5,0.333) -- ++(1,-0.166) -- ++(-0.5,-0.333) -- cycle;
\fill [red, opacity = 0.2] (7.5,0.75) ++(-1, 0.166) -- ++(0.5,0.333) -- ++(0,0.75) -- ++(-0.5,-0.333) -- cycle;
\fill [red, opacity = 0.2] (7.5,0.75) ++(-1, 0.166) -- ++(1,-0.166) -- ++(0.5,0.333) -- ++(-1,0.166) -- cycle;
\filldraw [red, opacity = 0.2] (7.5, 0.75) ++(0.5, 0.333) ++(0, 0.75) -- ($(6,1) +(2,1.333)$) -- ++(-0.5, -0.333) -- ++(0,-0.2) -- (6,0.8) -- (6,0.5) --cycle;

\filldraw [red, opacity = 0.2] (6,0.5) -- (6,0.8) -- (0,1.8) -- (0,1.5) --cycle;
\filldraw [red, opacity = 0.2] (0,1.5) -- (3,3.5) -- (3,3.8) -- (0,1.8) --cycle;
\filldraw [red, opacity = 0.2] (3,3.5) -- (9, 2.5) -- (9,2.8) -- (3,3.8) --cycle;
\filldraw [red, opacity = 0.2] (9, 2.5) -- (8,1.3) -- ++(0,0.3) -- (9,2.8) --cycle;
\filldraw [red, opacity = 0.2] (7.5,0) -- ++(0.5,0.333) -- ++(-1.25,0.208) -- ++(-0.5,-0.333) -- cycle;
\filldraw [red, opacity = 0.2] (7.5,0) -- ++(0,0.75) -- ++(0.5,0.33) -- ++(0,-0.75) -- cycle;
\filldraw [red, opacity = 0.2] (7.5,2) -- ++(0.5,0.333) -- ++(-1.25,0.208) -- ++(-0.5,-0.333) -- cycle;
\draw [purple, thick] (7.5,2) -- ++(-1.25,0.208) -- ++(0.5,0.33) -- ++(1.25,-0.208) -- ++ (0,-0.5) -- ++(-1,0.166) -- ++ (0,-0.75) -- ++(1,-0.166) -- ++(0,-0.75) --++(-1.25,0.208) -- ++(-0.5,-0.33) -- ++(1.25,-0.208) -- ++(0,0.75) -- ++(-1,0.166) -- ++(0,0.75) -- ++(1,-0.166) -- (6,0.5) -- (0,1.5) -- (3,3.5) -- (9,2.5) -- (8,1.3) -- (8,1.6) -- (9,2.8) -- (3,3.8) -- (0,1.8) -- (6,0.8) -- (7.5,1.8) -- cycle;
\end{scope}
\end{tikzpicture}
\caption{The disks $S_1 \cap S_2$ and $S_3 \cap S_2$. }
\label{pic_double_bubble_example_4}
\end{center}
\end{figure}
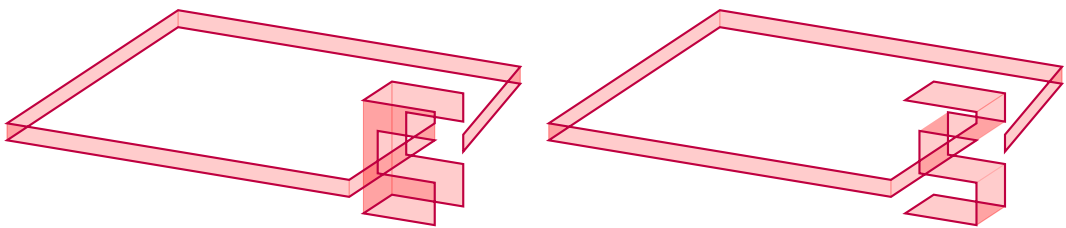

\end{document}